\documentclass[prd,tightenlines,nofootinbib,superscriptaddress]{revtex4}

\usepackage{amsfonts,amssymb,amsthm,bbm}

\usepackage{amsmath}

\usepackage{hyperref}

\usepackage{color,psfrag}
\usepackage[dvips]{graphicx}

\usepackage{tkz-euclide}
\usepackage{tikz}
\usetikzlibrary{calc}
\usetikzlibrary{decorations.pathmorphing}
\usetikzlibrary{shapes.geometric}
\usetikzlibrary{arrows,decorations.markings}
\usetikzlibrary{shadings,intersections}
\usetikzlibrary{matrix}

\usepackage{import}

\usepackage{subcaption}
\newcommand{\bC}{{\mathbb C}}
\newcommand{\N}{{\mathbb N}}
\newcommand{\R}{{\mathbb R}}

\newcommand{\cA}{{\mathcal A}}

\newcommand{\cE}{{\mathcal E}}
\newcommand{\cF}{{\mathcal F}}

\newcommand{\cL}{{\mathcal L}}
\newcommand{\cH}{{\mathcal H}}
\newcommand{\cM}{{\mathcal M}}
\newcommand{\cN}{{\mathcal N}}
\newcommand{\cR}{{\mathcal R}}
\newcommand{\cO}{{\mathcal O}}

\newcommand{\cV}{{\mathcal V}}
\newcommand{\cD}{{\mathcal D}}
\newcommand{\cC}{{\mathcal C}}
\newcommand{\cS}{{\mathcal S}}

\newcommand{\cW}{{\mathcal W}}
\newcommand{\cI}{{\mathcal I}}

\newcommand{\SU}{\mathrm{SU}}

\newcommand{\ISU}{\mathrm{ISU}}

\newcommand{\SL}{\mathrm{SL}}
\newcommand{\SO}{\mathrm{SO}}

\newcommand{\bU}{\mathrm{U}}

\newcommand{\be}{\begin{equation}}
\newcommand{\ee}{\end{equation}}
\newcommand{\beq}{\begin{eqnarray}}
\newcommand{\eeq}{\end{eqnarray}}
\newcommand{\bes}{\begin{eqnarray}}
\newcommand{\ees}{\end{eqnarray}}

\newcommand{\mat} [2] {\left ( \begin{array}{#1}#2\end{array} \right ) }

\renewcommand{\u}{{\mathfrak{u}}}
\newcommand{\su}{{\mathfrak{su}}}

\renewcommand{\sl}{{\mathfrak{sl}}}
\newcommand{\so}{{\mathfrak{so}}}

\newcommand{\la}{\langle}
\newcommand{\ra}{\rangle}

\newcommand{\tr}{{\mathrm{Tr}}}
\newcommand{\f}{\frac}

\def\nn{\nonumber}
\def\pp{\partial}

\def\rd{\mathrm{d}}

\newcommand{\id}{\mathbb{I}}

\def\ie{\textit{i.e.}}
\def\eg{\textit{e.g.}}

\usepackage[disable]{todonotes} 
\usepackage{mathtools}
 \newcommand{\comment}[1]{}  

\usepackage{comment}

\usepackage{pgfplots}
\usepackage{tikz-3dplot}

\tdplotsetmaincoords{60}{115}

\def\tz{\tilde{z}}
\def\tw{\tilde{w}}
\def\txi{\tilde{\xi}}
\def\bxi{\bar{\xi}}

\def\teta{\tilde{\eta}}

\newcommand{\ba}{\begin{array}}
\newcommand{\ea}{\end{array}}
\usepackage{advdate}

\def\kz{\left|z\right>}

\def\kzb{\left|\bar{z}\right>}

\def\bz{\left<z\right|}

\def\bzb{\left<\bar{z}\right|}

\def\zb{\bar{z}}
\def\zt{\tilde{z}}
\def\ztb{\bar{\tilde{z}}}
\def\kzd{\left|z\right]}

\def\bzd{\left[z\right|}

\def\tX{\tilde{X}}

\def\Re{\text{Re}}
\def\I{\text{Im}}

\def\z{\mathbf{z}}

\def\r{\mathbf{r}}

\def\HO{\text{HO}}
\def\tet{\text{tet}}
\def\sl{\text{sl}}
\def\cohe{\text{cohe}}
\def\phys{\text{phys}}
\def\Ad{\text{Ad}}
\def\tr{\text{Tr}}
\def\lp{\ell_{\text{p}}}
\def\G{\text{G}}

\def\B{\text{B}}

\def\e{{\bf e}}
\def\A{{\bf A}}

\newcommand{\hmm}{\textcolor{black}}

\newcommand{\Mat} [2] {\left\{ \begin{array}{#1}#2\end{array} \right\} }

\newcommand{\ri}{\MakeUppercase{\romannumeral 1}}
\newcommand{\rii}{\MakeUppercase{\romannumeral 2}}
\newcommand{\riii}{\MakeUppercase{\romannumeral3}}

\newtheorem{theorem}{Theorem}[section]
\newtheorem{lemma}[theorem]{Lemma}
\newtheorem{prop}[theorem]{Proposition}
\newtheorem{coro}[theorem]{Corollary}

\begin{document}

\title{3D Quantum Gravity from Holomorphic Blocks}

\author{{\bf Etera R. Livine}}\email{etera.livine@ens-lyon.fr}
\affiliation{Universit\'e de Lyon, ENS de Lyon, CNRS, Laboratoire de Physique LPENSL, 69007 Lyon, France}

\author{{\bf Qiaoyin Pan}}\email{qpan@fau.edu}
\affiliation{Department of Physics, Florida Atlantic University, 777 Glades Road, Boca Raton, FL 33431, USA}
\affiliation{Perimeter Institute, 31 Caroline St North, Waterloo N2L 2Y5, Ontario, Canada}
\affiliation{Department of Applied Mathematics, University of Waterloo, Waterloo, Ontario, Canada}

\date{\today}

\begin{abstract}

Three-dimensional gravity is a topological field theory, which can be quantized as the Ponzano-Regge state-sum model built from the $\{3nj\}$-symbols of the recoupling of the $\SU(2)$ representations, in which spins are interpreted as quantized edge lengths in Planck units. It describes the flat spacetime as gluing of three-dimensional cells with a fixed boundary metric encoding length scale. In this paper, we revisit the Ponzano-Regge model formulated in terms of spinors and rewrite the quantum geometry of 3D cells with holomorphic recoupling symbols. These symbols, known as Schwinger's generating function for the $\{6j\}$-symbols, are simply the squared inverse of the partition function of the 2D Ising model living on the boundary of the 3D cells.
They can furthermore be interpreted, in their critical regime,  as scale-invariant basic elements of geometry. 
We show how to glue them together into a discrete topological quantum field theory.
This reformulation of the path integral for 3D quantum gravity, with a rich pole structure of the elementary building blocks, opens a new door toward the study of phase transitions and continuum limits in 3D quantum gravity, and offers a new twist on the construction of a duality between 3D quantum gravity and a 2d conformal theory.

\end{abstract}

\maketitle
\tableofcontents

\section*{Introduction}

Spinfoam models (see $\eg$ \cite{Baez:1999sr,Livine:2010zx,Perez:2012wv} for reviews) provide a rigorous, background-independent and non-perturbative path integral quantization of gravitational theories based on discrete topological quantum field theories and state-sum models. They define probability amplitudes for histories of quantum states of geometry defined as entangled collections of discrete excitations.
They can be considered as a quantum version of Regge calculus for discretized general relativity \cite{Regge:2000wu}. They provide transition amplitudes for spin network states in loop quantum gravity \cite{Reisenberger:1996pu,Baez:1999sr,Engle:2007wy}. They also provide the triangulation weights in sum-over-random-geometries approaches to quantum gravity, such as matrix models, tensor models and group field theories \cite{Reisenberger:2000fy,Reisenberger:2000zc}. Finally, they are a natural mathematical framework for defining topological invariants $\eg$ \cite{Barrett:2008wh}, and have been shown to be related to non-commutative geometry, $\eg$ \cite{Freidel:2005me}, to lattice field theories with quantum group gauge symmetries, $\eg$ \cite{Bonzom:2014wva,Dupuis:2020ndx}, and to higher gauge theories \cite{Girelli:2007tt,Baratin:2014era,Asante:2019lki,Girelli:2021zmt}.

Retrospectively, the first spinfoam model was constructed by Ponzano and Regge \cite{Ponzano:1968se} and defines a discrete topologically-invariant path integral for three space-time dimensional gravity in Euclidean signature with vanishing cosmological constant \cite{Ooguri:1991ni,Freidel:2004vi,Barrett:2008wh}.
Let us underline that this is not a Wick-rotated path integral but it is truly the quantum theory of a positive signature metric with probability amplitudes in $\exp[iS_{grav}]$.
{\it A posteriori}, the Ponzano-Regge model has been understood as the discretized path integral for 3D gravity written in terms of veirbein-connection variables as a topological BF theory with gauge group $\SU(2)$ \cite{Freidel:2004vi}.
There exists a Lorentzian version of this model based on the $\SU(1,1)$ gauge group \cite{Freidel:2000uq,Davids:2000kz,Freidel:2005bb,Girelli:2015ija}. One can also take into account a non-vanishing cosmological constant through a q-deformation of the gauge group \cite{Bonzom:2014bua,Bonzom:2014wva,Dupuis:2020ndx}. This yields the Turaev-Viro topological invariant when the quantum deformation parameter $q$ is a root of unity \cite{Turaev:1992hq}. Through this relation, the Ponzano-Regge state-sum has been understood as a special case of the Reshetikhin-Turaev invariants \cite{Freidel:2004nb} and thereby related to the quantization of 3D gravity as a Chern-Simons theory as advocated by Witten \cite{Witten:1988hc}.

\smallskip

The Ponzano-Regge model is constructed as a path integral over discrete 3D geometries. Considering a 3D piecewise linear cellular complex $\Delta$, which can be thought of as a  cellular decomposition of a Riemannian 3D manifold, one has a hierarchy of cells with dimensions between 0 and 3, which one dresses with algebraic data following the logic from algebraic topology. Then one builds a probability amplitude for the 3D geometry from the algebraic data, such that it is topologically invariant in the sense that it does not depend on the details of the 3D cellular complex but only on its topology (and on boundary data). So the hierarchy of the Ponzano-Regge state-sum is:
\begin{itemize}

\item {\it 0-cells (points)}:\\
This level is actually put aside by spinfoam models, which focus on geometrical structures with co-dimensions less than or equal to 2. This allows for conical singularities $\eg$ \cite{Crane:2001kf}, which can be appropriately controlled in sums over random discrete geometries $\eg$ \cite{Gurau:2010nd}.

\item {\it 1-cells (edges)}:\\
Edges have the topology of a basic segment, devoid of any winding information.
One dresses edges with irreducible representations of the Lie group $\SU(2)$. These are labelled by half-integers $j\in\f\N2$, usually referred to as {\it spins}. The Hilbert space $\cV_{j}$ carrying the representation of spin $j$ is of dimension $d_{j}=(2j+1)$ and its standard basis is given the spin basis labelled by the spin $j$ and the magnetic moment $m$ running from $-j$ to $+j$ by integer steps. \\
The spin $j_{e}$ carried by an edge $e$ gives its quantized length in Planck unit, $\ell_{e}=j_{e} \ell_{Planck}$. A state in $\cV_{j}$ is then interpreted as a quantum 3-vector of length $j$. This is the key to the geometrical interpretation of the Ponzano-Regge model.

\item {\it 2-cells (faces)}:\\
Faces are assumed to have the topology of a two-dimensional disk with a $\cS^{1}$ boundary.
They are dressed with {\it intertwiners}, that is $\SU(2)$-invariant states living in the tensor product of the spins living on the edges forming the face's boundary:
\be
\cI_{f}=\textrm{Inv}_{\SU(2)}\bigg{[}\bigotimes_{e\in\pp f}\cV_{j_{e}}\bigg{]}
\,.
\ee
In the case of a triangle, consisting of three edges carrying spins $j_{1},j_{2},j_{3}$, this intertwiner state is one-dimensional if the spins satisfy the triangular inequalities and 0-dimensional in all other cases. When it is non-trivial, the unique intertwiner state is given by the Clebsh-Gordan coefficients, or equivalently the Wigner $\{3j\}$-symbols, encoding the recouplings of the three spins.
In general, intertwiner states are interpreted as the quantum version of polygons.

\item {\it 3-cells (elementary 3D regions or bubbles)}:\\
Elementary 3-cells $\sigma$ have the topology of a 3-ball. Their boundary has the topology of a 2-sphere made of faces glued together along edges. We focus on the boundary of each 3-cell and think of them as bubbles. For each bubble, we usually introduce the dual graph $\Gamma=(\pp\sigma)_{1}^{*}$, defined formally as the 1-skeleton of the dual of the boundary 2-complex: each face is represented as a (dual) node and each edge is represented as a (dual) link linking two nodes. Each edge or link comes with its spin.  Each face or node comes with its intertwiner state. Such graph with representations on its links and intertwiners at its nodes, $\Gamma_{\{j_{l},\iota_{n}\}}$, is called a {\it spin network}, as illustrated on fig.\ref{fig:boundaryspinnet}.
\begin{figure}[h!]
\begin{tikzpicture}[one end extended/.style={shorten >=-#1},
 one end extended/.default=0.7cm]
\def\x{1.5};
	\coordinate (O) at (0,0);
	\coordinate (A) at ([shift=(-30:\x)]O);
	\coordinate (B) at ([shift=(-90:\x)]A);
	\coordinate (C) at ([shift=(-150:\x)]B);
	\coordinate (D) at ([shift=(-210:\x)]C);
	\coordinate (E) at ([shift=(90:\x)]D);
	\coordinate (F) at ([shift=(18:\x)]A);
	\coordinate (G) at ([shift=(-54:\x)]F);
	\coordinate (H) at ([shift=(-18:\x)]B);
	\coordinate (I) at ([shift=(-20:1.5*\x)]C);
	\coordinate (J) at ([shift=(-5:2*\x)]H);

\draw[very thick] (O) -- (A);
\draw[very thick] (A) -- (B);
\draw[very thick] (B) -- (C);
\draw[very thick] (C) -- (D);
\draw[very thick] (D) -- (E);
\draw[very thick] (E) -- (O);
\draw[very thick] (A) -- (F);
\draw[very thick] (F) -- (G);
\draw[very thick] (G) -- (H);
\draw[very thick] (H) -- (B);
\draw[very thick] (C) -- (I);
\draw[very thick] (H) -- (I);
\draw[very thick] (J) -- (G);
\draw[very thick] (H) -- (J);

\coordinate (o) at ([shift=(-90:\x)]O);
\coordinate (a) at ([shift=(0:1.62*\x)]o);
\coordinate (b) at ([shift=(-120:1.2*\x)]a);
\coordinate (c) at ([shift=(-34:1.1*\x)]a);

\draw[red,one end extended] (o) -- node[pos=1.3,right]{$j_1$} ($(O)!(o)!(E)$);
\draw[red,one end extended] (o) -- node[pos=1.3,left]{$j_2$} ($(O)!(o)!(A)$);
\draw[red,one end extended] (o) -- node[pos=1.3,left]{$j_5$} ($(C)!(o)!(D)$);
\draw[red,one end extended] (o) -- node[pos=1.3,above]{$j_6$} ($(D)!(o)!(E)$);
\draw[red,one end extended] (a) -- node[pos=1.3,right]{$j_7$} ($(A)!(a)!(F)$);
\draw[red,one end extended] (a) -- node[pos=1.3,above]{$j_8$} ($(F)!(a)!(G)$);
\draw[red,one end extended] (b) -- node[pos=1.3,right]{$j_{11}$} ($(H)!(b)!(I)$);
\draw[red,one end extended] (b) -- node[pos=1.3,left]{$j_{12}$} ($(C)!(b)!(I)$);
\draw[red,one end extended] (c) -- node[pos=1.3,right]{$j_{13}$} ($(G)!(c)!(J)$);
\draw[red,one end extended] (c) -- node[pos=2,right]{$j_{14}$} ($(H)!(c)!(J)$);

\draw[red] (o) -- node[pos=0.4,above]{$j_3$} (a);
\draw[red] (o) -- node[pos=0.5,left]{$j_4$} (b);
\draw[red] (a) -- node[pos=0.5,right]{$j_9$} (b);
\draw[red] (a) -- node[pos=0.5,above]{$j_{10}$} (c);

\draw[red] (o) node{$\bullet$} node[above]{\large $\iota_1$};
\draw[red] (a) node{$\bullet$} node[above left]{\large $\iota_2$};
\draw[red] (b) node{$\bullet$} node[right]{\large $\iota_3$};
\draw[red] (c) node{$\bullet$} node[right]{\large $\iota_4$};

\end{tikzpicture}
\caption{(A portion of) the cellular decomposition of a 2-sphere made of faces glued along edges ({\it in black and thick}) and its dual graph $\Gamma$ ({\it in red and thin}) made of links and nodes. Each link $l$ is dressed with a spin $j_l$ and each node $n$ is dressed with an intertwiner $\iota_n$, which together represent a spin network $\Gamma_{\{j_l,\iota_n\}}$.}
\label{fig:boundaryspinnet}
\end{figure}
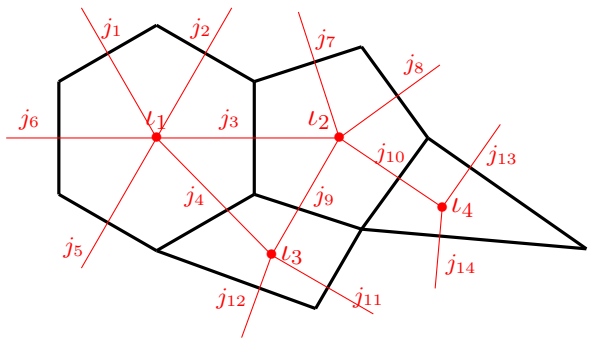
We define the probability amplitude for the geometry of the 3-cell as the evaluation of its boundary spin network:
\be
\cA_{\sigma}[\{j_{e},\iota_{f}\}_{e,f\in\pp\sigma}]
=
\tr_{\{j_{e}\}_{e\in\pp\sigma}} \, \bigotimes_{f\in\pp\sigma} \iota_{f}
=
\tr_{\{j_{n}\}_{n\in\Gamma}} \, \bigotimes_{n\in\Gamma} \iota_{n}
=
\cE_{\Gamma}[\{j_{l},\iota_{n}\}_{l,n\in\Gamma}]
\,.
\ee
The trace $\tr$ here is a slightly abusive notation. It means gluing the intertwiner states using the inner product on each edge, in the tensor product Hilbert space $\bigotimes_{e\in\pp\sigma}\cV_{j_{e}}$.
In the case of a 3-simplex, or tetrahedron, the boundary graph consists of four nodes connected to each other through six links, as illustrated in fig.\ref{fig:spinnetwork}. The links carry six spins $j_{1,..,6}$ while the 3-valent nodes carry the corresponding unique intertwiner state numerically given by the Clebsh-Gordan coefficients. The resulting spin network evaluation is the celebrated Wigner's $\{6j\}$-symbol.
\\
In general, one could glue the intertwiner states by inserting $\SU(2)$ group elements $g_{l}$ (or even $\SL(2,\bC)$ group elements as in \cite{Costantino:2011gen}) along each link. This yields the spin network wave-function $\psi_{\{j_{l},\iota_{n}\}}^{\Gamma}\in\cC^{\infty}(\SU(2)^{\times E_{\sigma}})$ where $E_{\sigma}$ counts the number of edges on the bubble boundary (or equivalently the number of links in the boundary graph $\Gamma$). The  spin network evaluation $\cA_{\sigma}[\Gamma, \{j_{l},\iota_{n}\}]$ then truly is the evaluation of the spin network wave-function $\psi_{\{j_{l},I_{n}\}}^{\Gamma}$ on trivial group elements $g_{l}=\id$ in $\SU(2)$, reflecting the fact that physical states in pure 3D (quantum) gravity with vanishing cosmological constant have a flat curvature. The Ponzano-Regge model can indeed be understood as defining the projector onto the moduli space of flat $\SU(2)$ connections $\eg$ \cite{Ooguri:1991ib,Freidel:2005bb}.

\end{itemize}

At the end of the day, the probability amplitude of a 3D cellular complex $\Delta$ dressed with the algebraic data of spins on its 1-cells and intertwiners on its 2-cells is obtained by putting all those elementary building blocks together and straightforwardly computing the product of the probability amplitudes of each 3-cells together with appropriate weights for the edges and faces:
\be
\cA[\Delta, \{j_{e},\iota_{f}\}]
=
\prod_{e}(-1)^{2j_{e}}d_{j_{e}}\,
\prod_{f} (-1)^{\sum_{e\in\pp f}j_{e}}\,
\prod_{\sigma}\cA_{\sigma}[\{j_{e},\iota_{f}\}_{e,f\in\pp\sigma}]\,.
\label{eq:spinfoam_general}
\ee
The edge weight only depends on the spin carried by the edge and is simply the dimension of the corresponding representation up to a sign. The face weight is a mere parity factor. The bubble weight carries the non-trivial dynamical information of the model.
Such a structure with a cellular complex dressed with representations and intertwiner states and a probability amplitude defined as the product of local amplitude for each cell depending solely on the algebraic data it carries is called the {\it local spinfoam ansatz} for a path integral over discrete quantum geometries. It has been shown in \cite{Boulatov:1992vp,Reisenberger:2000fy} that they are Feynman diagram amplitudes of non-local non-commutative field theories, referred to as {\it group field theories} or tensorial group field theories (see \cite{Carrozza:2013oiy} on recent studies of those Feynman diagrams and the renormalization of such field theories).

A first important remark is that the Ponzano-Regge ansatz is topologically invariant. Indeed, let us consider the amplitude for a 3D cellular complex $\Delta$ obtained by summing over all possible algebraic data. More precisely, we allow $\Delta$ to have a 2D boundary, we keep the algebraic data on $\pp\Delta$ fixed while we sum over bulk spins and bulk intertwiners:
\be
\cA[\Delta, \{j_{e},\iota_{f}\}_{e,f\in\pp\Delta}]
=
\sum_{\{j_{e},\iota_{f}\}_{e,f\in\Delta^{\textrm{o}}}}\cA[\Delta,  \{j_{e},\iota_{f}\}_{e,f\in\pp\Delta},\{j_{e},I=\iota_{f}\}_{e,f\in\Delta^{\textrm{o}}}]\,,
\ee
where we have written $\Delta^{\textrm{o}}=\Delta\setminus\pp\Delta$ for the interior or bulk of $\Delta$. Under appropriate gauge-fixing\footnotemark, this amplitude can be shown to depend only on the boundary data and on the topology of $\Delta$ and to never depend on the details of the bulk cellular complex $\Delta^{\textrm{o}}$ \cite{Freidel:2004vi,Freidel:2005bb,Dowdall:2009eg,Goeller:2019zpz}.
\footnotetext{
The sum over bulk spins and bulk intertwiners is usually divergent, just as Feynman diagrams in quantum field theory.  It is possible to render those amplitudes finite by $q$-deforming $\SU(2)$ at root of unity, with $q=\exp(2i\pi/N+2)$ for an integer $N\in\N$. This gives the Turaev-Viro topological invariant \cite{Turaev:1992hq} and is interpreted as switching on a non-vanishing positive cosmological constant $\Lambda >0$. Even without quantum deforming the gauge group, one can still identify the translational gauge symmetry responsible for the divergences and gauge-fix them - typically by fixing the value of the spins on the edges belonging to a maximal tree in $\Delta^{\textrm{o}}$, in which case the gauge-fixed amplitudes never depend neither on the choice of the gauge-fixing tree nor on the bulk cellular complex $\Delta^{\textrm{o}}$ \cite{Freidel:2004vi,Barrett:2008wh,Bonzom:2010ar,Bonzom:2010zh}.
}
More precisely, the amplitude for a 3D cellular complex with the topology of a 3-ball remains the evaluation of its boundary spin network. For 3D cellular complexes with a non-trivial topology, one needs to evaluate the spin network wave-function on the values of the non-trivial holonomies along the non-contractible cycles of the bulk cellular complex or integrate over possible values \cite{Freidel:2005bb,Dowdall:2009eg,Goeller:2019zpz}.

The second important remark is that the Ponzano-Regge ansatz is locally holographic, in the sense that the probability amplitudes associated to bounded regions of 3D space entirely depend on their boundary data (and their topology).  The amplitude for each elementary 3-cell $\sigma$ depends by definition solely on its boundary data: the spins $j_{e}$ and intertwiners $\iota_{f}$ carried by the edges and faces on its boundary $e,f\in\pp\sigma$. It is important to stress that there is no new algebraic data associated to the 3-cells (no maps between intertwiner states as one could imagine).
The probability amplitude of the 3-cell is the evaluation of its spin network $\cA_{\sigma}[\{j_{e},\iota_{f}\}_{e,f\in\pp\sigma}]=\cE_{\Gamma}[\{j_{l},\iota_{n}\}_{l,n\in\Gamma}]$.
Moreover, this property is true also for non-elementary 3-cells, that is for every bounded 3D region. Indeed, as long as one considers a bounded 3D region $\cR$ with the topology of a 3-ball, the topological invariance property of the Ponzano-Regge amplitude, detailed above, implies that the amplitude $\cA[\cR, \{j_{e},\iota_{f}\}_{e,f\in\pp\cR}]$ is simply the evaluation of the boundary spin network on $\pp\cR$. This means that any bounded 3D region with the topology of a 3-ball behaves exactly as an elementary 3-cell and its probability amplitude always only depends on its boundary data and never on the details of its bulk cellular decomposition.
In this setting, the local holography principle is deeply interlaced with the topological invariance of the theory\footnotemark.
\footnotetext{
A non-topological invariant model could still be locally holographic if one introduces the possibility of a non-trivial renormalization flow under coarse-graining, meaning that the probability amplitude of a bounded region would still be an evaluation of the boundary spin network wave function but that evaluation would now depend on extra parameters reflecting the size of the region (and perhaps other basic coarse-grained observables of the bulk geometry). The key would be that there would be only a finite number of extra parameters and that this number would be the same for all regions.
}

\bigskip

In the present work, we propose to revisit the Ponzano-Regge model and write it in terms of coherent boundary states for each 3-cell instead of pure spin networks sharply peaked on lengths. Those coherent states will be peaked on both intrinsic geometries - the edge lengths - and extrinsic geometry - the dihedral angles between faces, which define a discrete measure of extrinsic curvature. This reformulation  has two important features:
\begin{itemize}

\item Using a coherent superposition of boundary spin networks, defined as an infinite series over the spins controlled by couplings dual to the spins, actually amounts to considering a generating function for the spin network evaluations. This is the same logic as for a simple quantum harmonic oscillator, in which matrix elements of an operator in the coherent state basis. That is 
\be
\la z|\hat{\cO}| \tz\ra=\sum_{n,m\in\N}\f{\bar{z}^{n}\tz^{m}}{\sqrt{n!m!}} \la n|\hat{\cO}| m\ra
\ee
can be understood as a generating function for the matrix elements $\la n|\hat{\cO}| m\ra$ controlled by the complex couplings $z$ and $\tz$.
Here, we focus on 3-valent boundary graphs, for which we don't need intertwiner labels, so that spin network evaluations $\cE_{\Gamma}[\{j_{l},\iota_{n}\}_{l,n\in\Gamma}]$ simply depend on the spins on the boundary graph links. We will simply write $\cE_{\Gamma}[\{j_{l}\}_{l\in\Gamma}]$. Then we define coherent spin network evaluations similarly as for the harmonic oscillator as:
\be
E_{\Gamma}[\{Y_{l}\}_{l\in\Gamma}]
 =
\sum_{\{j_{l}\in\f\N2\}}  Y_{l}^{2j_{l}} \,\cW[\{j_{l}\}_{l\in\Gamma}]\,\cE_{\Gamma}[\{j_{l}\}_{l\in\Gamma}]\,,
\ee
with the couplings $Y_{l}\in\bC$ and weights $\cW[\{j_{l}\}_{l\in\Gamma}]$ possibly involving factorials of the spins \cite{Costantino:2011gen,Bonzom:2012bn,Bonzom:2015ova}.
Generating functions is a powerful mathematical tool. For instance, they typically map the asymptotic behaviour, here at large spins i.e. the semi-classical regime for length scales very large compared to the Planck length, onto poles of the generating function.

\item The coherent spin superpositions, or equivalently the generating functions, that we consider here allow for exact analytical resummation of the Ponzano-Regge amplitudes as rational functions in the couplings. They are actually the generalization of Schwinger's generating function for the $\{6j\}$-symbols \cite{Schwinger:1965an,Bargmann:1962zz}, they were introduced as coherent spin network states in \cite{Freidel:2010tt,Dupuis:2010iq,Dupuis:2011dh} and showed to lead to exact closed formula for spinfoam models in \cite{Varshalovich:1988qu,Freidel:2012ji}. At the end, the evaluations $E_{\Gamma}[\{Y_{l}\}_{l\in\Gamma}]$, for specific well-chosen weights  $\cW[\{j_{l}\}_{l\in\Gamma}]$, were shown to be given by the inverse squared  partition function of the 2D Ising model with inhomogeneous couplings $\tan^{-1} Y_{l}$ on the boundary graph \cite{Bonzom:2015ova}.
For instance, the generating function for the $\{6j\}$-symbols corresponds to the inverse of the square of the 2D Ising model on the tetrahedron with six variables $Y_{1,..,6}$ living on the edges and representing the strength of the coupling between the four triangles \cite{Bonzom:2011nv}.
Then, in general for arbitrary 3-cells and their boundary graphs, this provides formulas for the Ponzano-Regge amplitude as holomorphic functions of couplings living on the boundary of the 3-cells.
These formulas are at the heart of the proposed holographic duality between 3D quantum gravity defined by the Ponzano-Regge path integral and the 2D Ising model \cite{Bonzom:2015ova}.

\end{itemize}

Building on those previous works, we show how to glue the holomorphic amplitudes associated to the 3-cells - or in short, holomorphic blocks- defined as the evaluation of the coherent spin network superpositions on their 2d boundary. This gluing is done in a topologically invariant way, that is so that overall amplitudes of a 3D region do not depend on the chosen bulk cellular decomposition, and ultimately reproduces the sum over spins of the original Ponzano-Regge formulation.

This reformulation offers a new twist to the story of the Ponzano-Regge path integral. We indeed formulate it as a topological net of 2D Ising partition functions glued together: each 3-cell defines a 2D Ising model on its boundary, then those 3-cells, and thus those 2D Ising models, are glued together in a topologically-invariant fashion. We refer to this construction as a {\it topological Ising net}. It would be enlightening to investigate in the future how general such topological Ising nets can be, whether they can be defined in any dimension, using arbitrary powers of the Ising partition function, if they can be generalized to other condensed matter models and whether we can depart from topological invariance in a controlled way with a non-trivial, yet integrable, renormalization flow encoding the fusion of the 3-cell algebraic structure and amplitudes.

Moreover, one can look at this construction from the perspective of (finite distance) holographic dualities  {\it \`a la} AdS/CFT correspondence. As the 2D Ising model becomes a conformal field theory (CFT) in its critical regime, the exact equivalence of the present formulation between the 3D quantum gravity and the 2D Ising partition function, which holds for every value of Ising couplings, can be understood as a non-critical version of the gauge-gravity holography. Interpreting the Ising partition function for non-critical couplings as a non-critical version of conformal blocks, the 3D Ponzano-Regge path integral is realized as gluing such 2D non-critical blocks. This version of holography holds for discrete quantized geometries and not only at the level of field theories in the continuum limit (see $\eg$ \cite{Witten:1998qj,Witten:1998wy} for holographic duality at the field theory level).
More recent work following this line of thought and investigating the holographic behaviour of the Ponzano-Regge path integral can be found in \cite{Dittrich:2017hnl,Dittrich:2017rvb,Dittrich:2018xuk}.
A hope is that this reformulation will lead to new developments in the investigation of the phase diagram of 3D quantum gravity and the implementation of quasi-local holography in spinfoams and loop-gravity-inspired path integrals for quantum gravity.
 
\bigskip

In section \ref{sec:PR}, we review the standard formulation of the Ponzano-Regge state-sum as a path integral for discretized 3D gravity in its first-order formulation in terms of vierbein-connection variables. We show that the Ponzano-Regge amplitude for a 3D region is the spin network evaluation on the 2D boundary of the region, and that gluing neighbouring 3D regions is implemented by a fusion of those spin network evaluations done in a topologically-invariant way, which leads to a locally holographic formulation of 3D quantum gravity.

In section \ref{sec:SGF}, we introduce the generating function for spin network evaluations and compute it as a rational holomorphic function. The $\{6j\}$-symbol for the tetrahedron becomes a holomorphic $\{12z^{\times 2}\}$-symbol, equal to the inverse squared partition function of the 2D Ising model on the tetrahedron. We write the Ponzano-Regge model in terms of those holomorphic blocks and show the topological invariance of this new formulation.

Section \ref{sec:conformal} analyzes the pole structure of the $\{12z^{\times 2}\}$-symbol and of the tetrahedron gluing factors and investigates their geometrical interpretations in terms of gluing tetrahedra up to scale factors.


\section{The Ponzano-Regge spinfoam for 3D quantum gravity}
\label{sec:PR}

\subsection{3D gravity with boundary}
\label{sec:3d_gravity}
Throughout this paper, we consider gravity on a three-dimensional Euclidean manifold 
$\cM$ with two-dimensional boundary $\partial \cM$, thus the Einstein-Hilbert action with the GHY boundary term reads
\be
S[g_{\mu\nu}]
=\frac{1}{16\pi \G}\left[ \int_\cM \rd^3x \sqrt{g} R - 2 \int_{\partial \cM}\rd^2 x\, \sqrt{h} K \right]\,,
\label{eq:EH_1}
\ee
where $\G$ is the 3D gravitational constant proportional to the Planck length $\lp$, 
$h_{ab}$ is the induced metric on the boundary and $K$ the extrinsic scalar curvature.
In this paper, we put the length dimensions into the metric so that the coordinates carry no dimensions. 
 We have used the natural unit $\hbar=c=1$. 
 
We introduce a co-triad $\e$ field, which is an $\su(2)$-valued one-form on $\cM$, and a connection $\A$ on a principle $\SU(2)$-bundle over $\cM$, which can be viewed as an $\su(2)$-valued one-form on $\cM$. 
$\e$ is an orthonormal local frame of $\cM$, while
$\A$ is uniquely defined as the solution to the torsion-free condition, $\rd_\A \e=\rd \e+[\A,\e]=0$, of the co-triad. In components, they read
\be
\e^i_\mu \e^j_\nu \delta_{ij}=g_{\mu\nu}\,,\quad
\A_{\mu}^i=
-\epsilon^{ijk} \e_k^{\nu}\nabla_{\mu}\e_{j\nu} 
=\epsilon^{ijk}(\Gamma_{\mu\nu}^{\rho}\e_{j\rho}-\partial_\mu \e_{j\nu})\e_k^\nu\,,
\label{eq:A_of_e}
\ee
where $\Gamma_{\mu\nu}^\rho$ is the Christoffel symbol. 
We have used Greek letters to denote the coordinates on the manifold and Latin letters to denote the Lie algebra indices. 
The action \eqref{eq:EH_1} thus can be written as
\be
S[\e,\A]=\frac{1}{16\pi \G}\left[ \int_{\cM} \tr \left( \e \wedge  F(\A) \right)
+\int_{\partial \cM} \tr(\e\wedge \A )
\right]\,,
\label{eq:BF}
\ee
where $F(\A):=\rd \A + \f12[\A,\A]$ is the curvature two-form of $\A$, and the trace is over the Lie algebra indices. 

It is well-known that the bulk term recovers the Einstein-Hilbert action. \hmm{We now show} that the boundary term in \eqref{eq:BF} is also consistent with the GHY term imposing the torsion-free condition of the connection \eqref{eq:A_of_e}. 
We choose a gauge for the co-triad $\e^0_{\mu}=n_{\mu}$ on the boundary $\partial \cM$, where $\vec{n}$ is the vector normal to $\partial \cM$ that links the metric and the induced metric on the boundary by
\be
g_{\mu\nu}=h_{\mu\nu}+n_\mu n_\nu\,.
\label{eq:g_h_n}
\ee
The torsion-free connection on the boundary can be written as $\A^i_{\mu}= -\epsilon^{ij0}n^\nu \nabla_\mu \e^j_\nu - \epsilon^{i0k}\e_k^\nu \nabla_\mu n_\nu=2\epsilon^{ij}\e_j^\nu \nabla_\mu n_\nu$. Therefore, the boundary term reads
\be
\int_{\partial \cM}\epsilon^{\mu\nu\rho} \e^i_{\mu} \A_{i\nu} \hat{n}_{\rho}
= 2 \int_{\partial \cM}\epsilon^{\mu\nu\rho} \epsilon^{ij}\e_{i\mu} \e_j^\lambda \hat{n}_{\rho} \nabla_\nu n_\lambda 
= 2 \int_{\partial \cM} \det \e \, \nabla_{\mu} n^{\mu} 
= - 2 \int_{\partial \cM} \rd^2 x \, \sqrt{h}K\,,
\label{eq:bdrY_lqual_GHY}
\ee
which recovers the GHY boundary term in \eqref{eq:EH_1}. $\vec{\hat{n}}$ is a dimensionless normal direction vector to $\partial \cM$ that is proportional to $\vec{n}$, say $\vec{n}=N\vec{\hat{n}}$ \footnotemark{}. 
On the other hand, the variation is well-post upon the Dirichlet boundary condition $\delta \e^i_{\mu} |_{\partial \cM}=0$, which is consistent with the boundary condition $\delta g|_{\partial \cM}=0$ of the action \eqref{eq:EH_1}. 
\footnotetext{
It is important to note the difference between $n_\mu$ and $\hat{n}_\mu$.  
As the metric carries square of length dimensions, $n_\mu$ is with dimension of length according to \eqref{eq:g_h_n} while $\hat{n}_\mu$ is dimensionless so that the dimension of the boundary action \eqref{eq:bdrY_lqual_GHY} is correct. 
}

In the first-order formalism, $\e$ and $\A$ are treated independently, thus the action is of the BF type. 
The theory possesses two kinds of gauge symmetries, namely the Lorentz transformation, performed by an $\SU(2)$ group element $g$ (infinitesimally parametrized by an $\su(2)$-valued function $\lambda$), and the translation, performed by a scalar field $\Phi$ (infinitesimally parametrized by an $\su(2)$-valued function $\phi$). The finite and infinitesimal transformation laws are given by \footnotemark{}
\be\ba{ll}
\left|\ba{lll}
\e & \mapsto & g\e g^{-1}\\[0.15cm]
\A &\mapsto & g \A g^{-1}+ g\,\rd g^{-1}
\ea\right.\,,\quad
&
\left|\ba{lll}
\e & \mapsto &  \e + \rd_{\A} \Phi\\[0.15cm]
\A & \mapsto &\A
\ea\right.\,.\\[0.5cm]
\left|\ba{l}
\delta^\text{L}_\lambda \,\e = [\lambda, \e]\\[0.15cm]
\delta^\text{L}_\lambda \, \A = \rd_\A \lambda
\ea\right. \,,\quad 
&
\left|\ba{l}
\delta^\text{t}_\phi \e = \rd_\A \phi \\[0.15cm]
\delta^\text{t}_\phi \A =0
\ea\right. \,.
\ea
\label{eq:gauge}
\ee
\footnotetext{The subscript ``L'' denote Lorentz transformation and ``t'' denotes translation.}
The variation of the BF action with boundary term \eqref{eq:BF} is
\be
\delta S[\e,\A]= \int_{\cM} \tr \left( \delta\e \wedge F + \delta \A \wedge \rd_{\A}\e \right) +\int_{\partial \cM}\tr \left( \delta\e\wedge\A \right)\,.
\ee
The equations of motion recover the torsion-free condition and, in addition, enforce the flatness of the manifold,
\be
\rd_{\A} \e =0\,, \quad
F(\A) =0\,,
\label{eq:BF_eom}
\ee
subject to the boundary condition $\delta\e=0$ \footnotemark{}.
\footnotetext{
The boundary condition $\delta\A=0$ requires no boundary term to give a well-post variation that leads to the equations of motion \eqref{eq:BF_eom}. Different boundary conditions lead to different quantization results. See \cite{OLoughlin:2000yww} for a discussion.
}

It has been well-known that \eqref{eq:BF} can be rewritten as a Chern-Simons action with the Poincar\'e gauge group $\ISU(2)$, 
\be
S_{\text{CS}}[	\cA]=\frac{k}{4\pi}\int_\cM \tr\left[ \cA \wedge \rd \cA + \f23 \cA \wedge \cA \wedge \cA \right]\,,\quad \cA_\mu = \lp^{-1}\,\e_\mu^i P_i +  \A_\mu^i J_i\,,
\label{eq:CS}
\ee
where $J_i$ and $P_i$ are the generators of Lorentz transformation and translation respectively \cite{Achucarro:1989ch,Witten:1988hc} \footnotemark{}.
\footnotetext{
According to our agreement of dimension assignment, $\e$ is with dimension of length and $\A$, thus $F$, is dimensionless.
}  
Explicitly,
\be
S_{\text{CS}}[\e, \A]= \frac{k}{4\pi\lp}\left[\int_{\cM} \tr \left[ \e \wedge F(\A) \right]
+  \int_{\partial \cM} \tr\left[  \e\wedge \A \right]\right]\,.
\label{eq:CS_ISU}
\ee
We have used the Killing form $\langle J_i,P_j\rangle=\delta_{ij}$ and others vanish. It is remarkable that both the bulk term and the boundary term of \eqref{eq:BF} are recovered
with $k=\frac{\lp}{4\text{G}}$.

Upon quantization, the manifold $\cM$ is discretized into local blocks packed compactly. 
The structure of these local blocks depends on that of the discrete phase space and the quantization process. 
We shall see in Section \ref{sec:conformal} that, by considering the conformal class of the boundary action, natural discretization and the quantization process results in local blocks with scale-invariant nature, $\ie$ it represents the conformal geometry of the discretized manifold, and are glued locally in a scale-invariant way.

\subsection{3D quantum gravity as a discretized gauge field path integral}
In this subsection, we review the spinfoam model for 3D quantum gravity, which is the discretized path integral of the BF action.
We set $\frac{1}{16\pi \text{G}}=1$ and will continue to apply this simplification in the rest of the article since we no longer need to deal with $\lp$. 
Consider the gravity on a three-manifold $\cM$ with boundary $\partial \cM$. 
The path integral is an integration over the gauge equivalent class of the bulk configuration subject to an admissible boundary condition $\partial$,
\be
Z[\cM,\partial]=\int \cD \A_{\text{B}}\int \cD \e_{\text{B}}  \,\exp\left(i S[\e,\A]  
 \right)
=Z[\partial \cM,\partial] \int \cD \A_{\text{B}} \,\delta(F(\A_{\text{B}})) \,,
\label{eq:PI_cont_boundary}
\ee
where $Z[\partial \cM,\partial]$ depends on the boundary condition and the subscript ``$\B$'' denotes the configuration in the bulk.  
The idea of spinfoam is to write the path integral \eqref{eq:PI_cont_boundary},
which reduces to the flatness condition in the bulk, in a discrete fashion that encodes the local geometrical information.

The quantization program proceeds firstly with the cellular decomposition of the spacetime, then with the discretization of variables 
so that they are concentrated on the cells of a particular dimension.
In particular, when the cellular decomposition is the triangulation, denoted as ${\bf T}$, a 3-simplex is a tetrahedron, denoted as $T$, and a 2-simplex is a triangle, denoted as $t$.

Apart from working on the cellular decomposition $\triangle$ of the 3D spacetime manifold $\cM$ and the graph $\Gamma=(\partial \triangle)^*_1$ on the boundary of $\triangle$, we will also work on the dual cellular decomposition $\triangle^*$. 
In the dual picture, a dual face $f^*$ is dual to one edge $e$, an oriented dual edge $e^*$ is dual to one face $f$, and a dual vertex $v^*$ is dual to one 3-cell $\sigma$.
When $e^*$ is on the boundary of $f^*$, we denote $e^*\in \partial f^*$. 
An example is illustrated in fig.\ref{fig:spinnetwork}, in which $l_i$ is dual to $e_i$, and the orientation of $\Gamma$ is chosen to be the Kasteleyn orientation, for the reason related to the duality between spin network and 2D Ising model \cite{Bonzom:2015ova}. A Kasteleyn orientation is such that each plaquette has an odd number of links oriented opposite the orientation induced by the plaquette \footnotemark{}.
\footnotetext{
It was shown in \cite{Cimasoni2007di} that there exists a Kasteleyn orientation if and only if the number of nodes in a graph $\Gamma$ embedded in a surface $\Sigma$ is even, and there are $2^{2g}$ equivalence classes of Kasteleyn orientations of $\Gamma$, where $g$ is the genus of $\Sigma$. Two Kasteleyn orientations are called equivalent when one can be obtained from the other by a sequence of moves, in which the orientation of all the links incident to a node are reversed. Thus there is a unique equivalence class of Kasteleyn orientation for the graph $\Gamma=(\partial T)^*_1$ dual to the boundary of a tetrahedron.
}
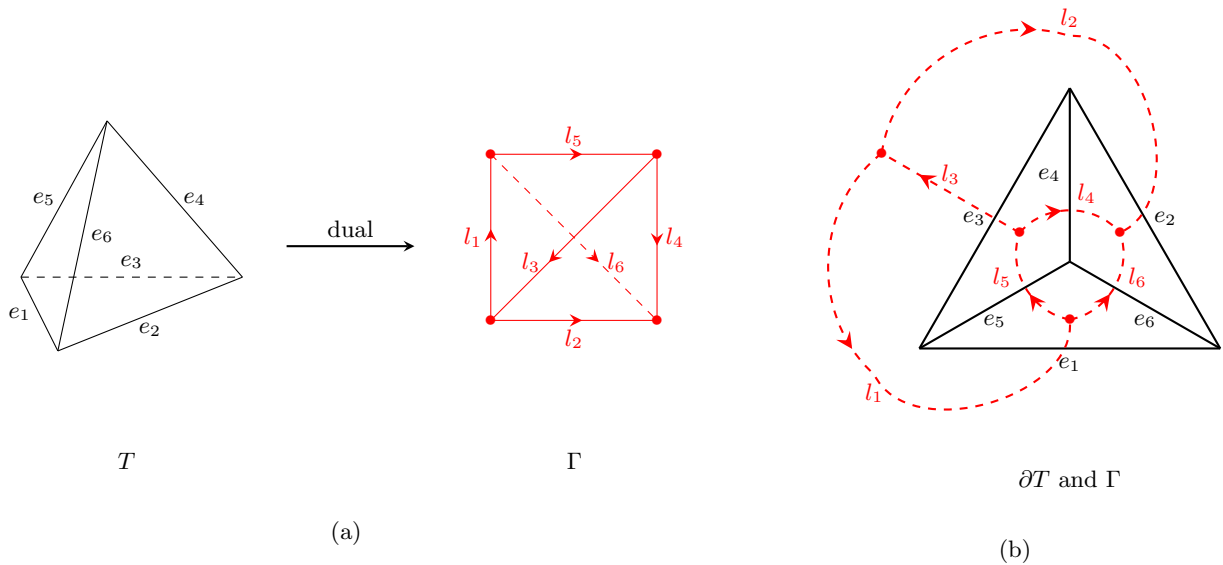
\begin{figure}[h!]
\centering
\begin{minipage}{0.6\textwidth}
	\centering
\begin{tikzpicture}[scale=1.7]
\coordinate (O1) at (0,0,0);

\coordinate (A1) at (0,1.061,0);
\coordinate (B1) at (0,-0.354,1);
\coordinate (C1) at (-0.866,-0.354,-0.5);
\coordinate (D1) at (0.866,-0.354,-0.5);
\coordinate (aa) at (0.25,-1.5, 0.25);
\coordinate (bb) at (4.25-0.5,-1.5, 0.25);

\draw (aa) node{$T$} ;
\draw (bb) node{$\Gamma$};

\draw (A1) -- (B1) node[midway,right]{$e_{6}$};
\draw (A1) -- (C1)node[midway,left]{$e_{5}$};
\draw (A1) -- (D1)node[midway,right]{$e_{4}$};
\draw (B1) -- (C1)node[midway,left]{$e_{1}$};
\draw[dashed] (C1) -- (D1)node[midway,above]{$e_{3}$};
\draw (D1) -- (B1)node[midway,below]{$e_{2}$};

\draw[->,>=stealth, thick] (1.5,0.177,0.25) to node[midway,above]{dual} (3-0.5,0.177,0.25);

\coordinate (A2) at (3.5-0.5,0.8);
\coordinate (B2) at (3.5-0.5,-0.5);
\coordinate (C2) at (4.8-0.5,0.8);
\coordinate (D2) at (4.8-0.5,-0.5);
\draw[red] (A2) node{$\bullet$};
\draw[red] (B2) node{$\bullet$};
\draw[red] (C2) node{$\bullet$};
\draw[red] (D2) node{$\bullet$};

\draw (0,1.8,0) node{};

\draw[red,postaction={decorate},decoration={markings,mark={at position 0.55 with {\arrow[scale=1.5,>=stealth]{>}}}}] (B2) -- (A2) node[midway,left]{$l_{1}$} ;
\draw[red,postaction={decorate},decoration={markings,mark={at position 0.55 with {\arrow[scale=1.5,>=stealth]{>}}}}] (A2) -- (C2) node[midway,above]{$l_{5}$};
\draw[dashed,red,postaction={decorate},decoration={markings,mark={at position 0.65 with {\arrow[scale=1.5,>=stealth]{>}}}}] (A2) -- (D2) node[pos=0.65,right]{$l_{6}$};
\draw[red,postaction={decorate},decoration={markings,mark={at position 0.65 with {\arrow[scale=1.5,>=stealth]{>}}}}] (C2) -- (B2) node[pos=0.65,left]{$l_{3}$};
\draw[red,postaction={decorate},decoration={markings,mark={at position 0.55 with {\arrow[scale=1.5,>=stealth]{>}}}}] (C2) -- (D2) node[midway,right]{$l_{4}$};
\draw[red,postaction={decorate},decoration={markings,mark={at position 0.55 with {\arrow[scale=1.5,>=stealth]{>}}}}] (B2) -- (D2) node[midway,below]{$l_{2}$};
\end{tikzpicture}
\subcaption{}
\label{fig:spinnetwork_a}
\end{minipage}
\begin{minipage}{0.38\textwidth}
	\centering
\begin{tikzpicture}
\coordinate (A) at (0,0);
\coordinate (B) at (4*4/4,0);
\coordinate (C) at (2*4/4,3.464*4/4);

\coordinate (c) at ($ (A)!.5!(B) $);
\coordinate (b) at ($ (A)!.5!(C) $);
\coordinate (a) at ($ (B)!.5!(C) $);

\path[name path = Aa] (A) -- (a);
\path[name path = Bb] (B) -- (b);
\path[name path = Cc] (C) -- (c);
\path [name intersections = {of = Aa and Bb,by=O}];

\draw[thick] (A) -- node[midway,below]{$e_1$} (B);
\draw[thick] (B) -- node[midway,right]{$e_2$} (C);
\draw[thick] (C) -- node[midway,left]{$e_3$} (A);

\draw[thick] (A) -- node[midway,below]{$e_5$}(O);
\draw[thick] (B) -- node[midway,below]{$e_6$}(O);
\draw[thick] (C) -- node[midway,left]{$e_4$}(O);

\coordinate (oa) at ($(A)!.5!(O)$);
\coordinate (ob) at ($(B)!.5!(O)$);
\coordinate (oc) at ($(C)!.5!(O)$);

\path[name path = Coa] (C) -- (oa);
\path[name path = Cob] (C) -- (ob);
\path[name path = Aob] (A) -- (ob);

\path[name intersections = {of = Coa and Bb, by=ACO}];
\path[name intersections = {of = Cob and Aa, by=BCO}];
\path[name intersections = {of = Aob and Cc, by=ABO}];

\coordinate (out) at ($1.5*(b)-0.5*(B)$);
\coordinate (aboveC) at ($1.2*(C)-0.2*(c)$);
\coordinate (belowA) at ($1.2*(A)-0.2*(a)$);

\draw[red] (ACO) node{$\bullet$};
\draw[red] (BCO) node{$\bullet$};
\draw[red] (ABO) node{$\bullet$};
\draw[red] (out) node{$\bullet$};

\draw[dashed,red,thick,postaction={decorate},decoration={markings,mark={at position 0.43 with {\arrow[scale=1.5,>=stealth]{>}}}}] (out) to[bend left=50,swap] (aboveC) to[bend left=80,swap] (BCO);
\draw[dashed,red,thick,postaction={decorate},decoration={markings,mark={at position 0.43 with {\arrow[scale=1.5,>=stealth]{>}}}}] (out) to[bend right=50,swap] (belowA) to[bend right=80,swap] (ABO);
\draw[dashed,red,thick,postaction={decorate},decoration={markings,mark={at position 0.43 with {\arrow[scale=1.5,>=stealth]{>}}}}] (ACO) to[bend left=50,swap] node[pos=0.65,above]  {$l_4$} (BCO);
\draw[dashed,red,thick,postaction={decorate},decoration={markings,mark={at position 0.43 with {\arrow[scale=1.5,>=stealth]{>}}}}] (ABO) to[bend right=50,swap] node[pos=0.6 ,right]  {$l_6$} (BCO);
\draw[dashed,red,thick,postaction={decorate},decoration={markings,mark={at position 0.43 with {\arrow[scale=1.5,>=stealth]{>}}}}] (ABO) to[bend left=50,swap] node[pos=0.6,left]  {$l_5$} (ACO);
\draw[dashed,red,thick,,postaction={decorate},decoration={markings,mark={at position 0.73 with {\arrow[scale=1.5,>=stealth]{>}}}}] (ACO) --node[red,midway,above]{$l_3$} (out);

\draw[red] (aboveC) node[above]{$l_2$};
\draw[red] (belowA) node[below]{$l_1$};

\coordinate (note) at ($1.5*(c) - 0.5*(C)$);
\draw (note) node{$\partial T$ and $\Gamma$};
\end{tikzpicture} 
\subcaption{}
\label{fig:spinnetwork_b}
\end{minipage}
\caption{ 
(a) From the tetrahedron $T$ to the dual boundary spin network graph $\Gamma\equiv(\partial T)^*_1$: each triangle $t \subset T$ is replaced by a node $n \subset \Gamma$ and each edge $e_i \subset T$ by a link $l_i \subset \Gamma$. The graph $\Gamma$ also has the combinatorics of a tetrahedron. (b) Combination of $\partial T$ ({\it in black}), which is the top view projection of the left panel, and its dual graph $\Gamma$ ({\it in red}).}
\label{fig:spinnetwork}
\end{figure}

As \eqref{eq:PI_cont_boundary} integrates out the triads and reduces to a function of the connection in the bulk, it is natural to define the discrete variables related to the connection. 
With the cellular decomposition of $\cM$ and its dual skeleton at hand, we assign an $\SU(2)$ group element $g_{e^*}$ called holonomy to each dual edge $e^*$, which encodes the discrete information of the connection. The reverse of the dual edge orientation maps $g_{e^*}$ to its inverse $g^{-1}_{e^*}$. Then the curvature is naturally defined by the path-ordered product of holonomies for a dual face with a randomly selected starting dual vertex. The $\e_\partial$ field on the boundary, by applying the same technique as in LQG, is discretized to be the flux variable $X_{e}$, which is an $\su(2)$ Lie algebra object, each assigned to an edge $e\in \partial \triangle$ on the discretized boundary. This leads to the discrete version of \eqref{eq:PI_cont_boundary},
\be
Z[\triangle,\partial\triangle]= C[\partial \triangle,\partial ] \int_{\SU(2)} \prod_{e^*\notin (\partial\triangle)^*}\rd g_{e^*} \,\prod_{f^*}\delta(
\prod\limits_{e^*\in \partial f^*}^{\longrightarrow}
g_{e^*})\,,
\label{eq:PI_disc}
\ee 
where the measure $\rd g$ is the Haar measure of $\SU(2)$, and the delta distribution on the $\SU(2)$ group imposes the group element in the argument to be identity. The discrete flatness is thus understood as the trivial holonomy associated to each dual face. 
$C[\partial \triangle,\partial]$ is a term that depends on the boundary cellular decomposition and the discrete boundary condition. Upon quantization, these boundary condition becomes boundary states, thus the quantization of \eqref{eq:PI_cont_boundary} will depend on the boundary states $\psi_{\Gamma}$ on the graph $\Gamma$, which is the dual of the boundary discretization.

The machinery of the spinfoam to achieve localization is to express the delta distribution in \eqref{eq:PI_disc} as a plane wave of $\SU(2)$ in a certain representation, which is then able to be decomposed into the product of plane waves localized in different cells. 
In other words, it is to construct the spinfoam path integral, also understood as the total amplitude, with a product of local amplitudes associated to dual vertices (or tetrahedra), dual edges (or faces) and dual faces (or edges), which capture the (admissible) local representations and local intertwiners, and then sum over all possible local configurations. 
 It is called the local spinfoam ansatz, which postulates that one can formally decompose the total amplitude into
\footnote{
We change from now on the notations and terminology for amplitudes based on the dual cellular decomposition $\triangle^*$ compared to \eqref{eq:spinfoam_general} to be consistent with most of the literature in spinfoams. The two notations are in one-to-one correspondence: an edge amplitude $\cA_{e^*}$ is equivalent to a face weight $\cA_f$, a face amplitude $\cA_{f^*}$ is equivalent to an edge weight $\cA_e$ and a vertex amplitude $\cA_{v^*}$ is equivalent to a bubble weight $\cA_{\sigma}$.} 
\be
Z[\triangle,\psi^{\rho}_{\Gamma}]=\sum_{\rho_{\text{B}}, \iota_{\text{B}}} \,\prod_{f^*} \cA_{f^*} \, \prod_{e^*} \cA_{e^*}\, \prod_{v^*} \cA_{v^*}\,,
\label{eq:amplitudes_boundary}
\ee
where $\rho$ and $\iota$ denote the representation and the intertwiner respectively, and the sum of representations is only over those associated to the bulk (denoted with the subscript ``B''). The boundary state $\psi^{\rho}_{\Gamma}$ encodes the representation $\rho$ associated to the boundary graph $\Gamma$, which is left in the expression of the total amplitude. 
The summation symbol was used as we assumed the representations $\rho$ are discrete, which will no longer be the case when we consider the spinfoam model with spinor representation in Section \ref{sec:SGF}. In the latter case, the summation symbol is changed to the integration symbol.

The edge amplitude $\cA_{e^*}$ and the face amplitude $\cA_{f^*}$ are both kinematical. The former describes the gluing of adjacent 3-cells, while the latter is there in order to compensate the factors to recover the delta distribution in \eqref{eq:PI_disc}, whose contribution dominates only in the quantum regime.  In contrast, the vertex amplitude $\cA_{v^*}$ contains the dynamical information of the spinfoam, thus deserves a deeper investigation.

\subsection{Fourier transform and spin network evaluations}

The Ponzano-Regge model \cite{Ponzano:1968se} is a realization of \eqref{eq:amplitudes_boundary} based on the triangulation of the manifold, where $\rho$ is given by the $\SU(2)$ irreducible representation labelled by spin $j\in \N/2$. $j$ is interpreted into the edge length suggested by the LQG framework. The vertex amplitude of the Ponzano-Regge state-sum geometrically describes the 3D geometry of the tetrahedron it is associated to, and an edge amplitude describes how two neighbouring tetrahedra are glued together. 

We start from the discrete path integral \eqref{eq:PI_disc} for a general cell decomposition and work on the dual picture. $\delta(g_{f^*})$ can be decomposed over the $\SU(2)$ spin representation using the Peter-Weyl theorem,
\be
\delta(g_{f^*})=\sum_{j_{f^*}}d_{j_{f^*}}\, \chi^{j_{f^*}}(g_{f^*})\,,
\label{eq:delta_g}
\ee
where $d_{j}\equiv 2j+1$ is the dimension of the spin $j$ representation space $\cV^{j}$ and $\chi^{j}(g)=\tr D^{j}(g)$ the character of $g$ in the spin $j$ representation, formulated as the trace of the Wigner matrix $D^j(g)$ of $g$ in the $j$ representation. 
Thus equivalently in the cellular decomposition picture, each edge $e$ is dressed with a spin $j_{e}$.
When decomposing $\delta(g_{f^*})$ into spin representation, one has an $\SU(2)$ group integration for each dual edge $e^*$, which is dual to an $N$-gon, of $N$ copies of the Wigner matrix $D^{j_i}(g)$ $(i=1 \cdots N)$.
 This is the projector, also called intertwiner $\id_{\cH}=\sum_{a}|\cI_a\rangle\langle \cI_a|$, of the kinematical Hilbert space onto the $\SU(2)$-invariant Hilbert space $ \text{Inv}_{\SU(2)}\left(\cV^{j_1}\otimes \cdots \otimes \cV^{j_N}\right)\bigotimes \text{Inv}_{\SU(2)}\left( \cV^{* j_1}\otimes \cdots \otimes \cV^{* j_N}\right)$. In components, the integration takes the form
\be
\int_{\SU(2)}\rd g_{e^*}\,
 D^{j_1}_{m_1m'_1}(g_{e^*})\cdots D^{j_n}_{m_nm'_n}(g_{e^*})
=\sum_a \langle \{j_i,m_i\} | \cI_a \rangle \overline{\langle \{j_i,m'_i\} | \cI_a \rangle }\,,
\label{eq:SU2_integral}
 \ee
where $a$ runs from 1 to the dimension of this $\SU(2)$-invariant space. Decomposing the intertwiner into a particular $a$ basis $\sum_a|\cI_a\rangle \langle \cI_a|$ corresponds to separating an $N$-valent node ($N\geq 3$) into $N-2$ three-valent nodes in a tree way. In the cellular decomposition picture, equally, it corresponds to adding edges in the $N$-gon on the boundary of a 3-cell so that the boundary is made of gluing triangles. 
One example of separating a 6-valent node into four 3-valent nodes and its correspondence of separating a hexagon into four triangles is given in fig.\ref{fig:hexagon}.
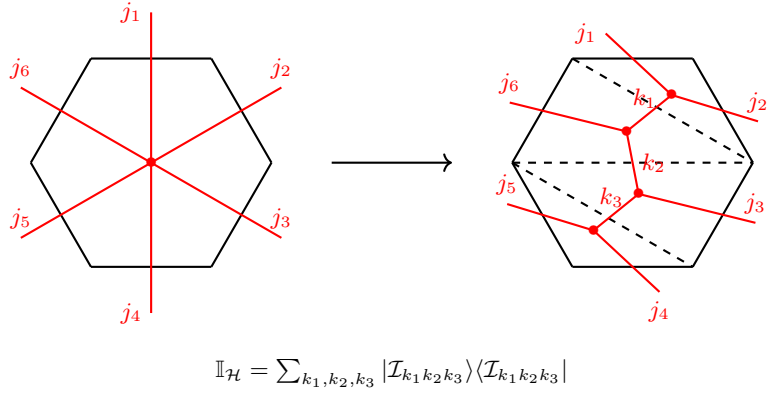
\begin{figure}
\centering
	\begin{tikzpicture}[scale=0.8,extended line/.default=1cm]
\coordinate (A) at (0,0);
\coordinate (B) at ([shift=(0:2)]A);
\coordinate (C) at ([shift=(-60:2)]B);
\coordinate (D) at ([shift=(-120:2)]C);
\coordinate (E) at ([shift=(180:2)]D);
\coordinate (F) at ([shift=(120:2)]E);

\draw[thick] (A) -- coordinate[midway](mab)(B);
\draw[thick] (B) -- coordinate[midway](mbc)(C);
\draw[thick] (C) -- coordinate[midway](mcd)(D);
\draw[thick] (D) -- coordinate[midway](mde)(E);
\draw[thick] (E) -- coordinate[midway](mef)(F);
\draw[thick] (F) -- coordinate[midway](mfa) (A);

\draw[thick,dashed,name path = ac] (A) -- (C);
\draw[thick,dashed] (F) -- (C);
\draw[thick,dashed,name path = fd] (F) -- (D);

\path [name path = be] (B) --(E);
\path [name intersections = {of = ac and be, by=int1}];
\path [name intersections = {of = fd and be, by=int2}];
\path  (B) --coordinate[pos=0.7](mb1) (int1);
\path  (E) --coordinate[pos=0.7](mb4) (int2);
\path  (C) --coordinate[pos=0.6](mb2) (mfa);
\path  (F) --coordinate[pos=0.6](mb3) (mcd);
\draw [thick,red] (mb1) -- node[at end,left]{$j_1$}($(mb1)!1.5cm!(mab)$);
\draw [thick,red] (mb1) -- node[at end,above]{$j_2$}($(mb1)!1.5cm!(mbc)$);
\draw [thick,red] (mb1) -- node[pos=.6,above]{$k_1$}(mb2);
\draw [thick,red] (mb2) -- node[at end,above]{$j_6$}($(mb2)!2cm!(mfa)$);
\draw [thick,red] (mb2) -- node[midway,right]{$k_2$}(mb3);
\draw [thick,red] (mb3) -- node[at end,above]{$j_3$}($(mb3)!2cm!(mcd)$);
\draw [thick,red] (mb3) -- node[pos=.6,above]{$k_3$}(mb4);
\draw [thick,red] (mb4) -- node[at end,below]{$j_4$}($(mb4)!1.5cm!(mde)$);
\draw [thick,red] (mb4) -- node[at end,above]{$j_5$}($(mb4)!1.5cm!(mef)$);

\draw[red] (mb1) node{$\bullet$};
\draw[red] (mb2) node{$\bullet$};
\draw[red] (mb3) node{$\bullet$};
\draw[red] (mb4) node{$\bullet$};

\coordinate (l) at ([shift=(180:7)]C);
\coordinate (r) at ([shift=(180:5)]C);
\draw[thick,->] (l) -- coordinate[midway](id)(r);
\coordinate (cI) at ([shift=(-90:3.5)]id);

\coordinate (A2) at ([shift=(180:8)]A);
\coordinate (B2) at ([shift=(180:8)]B);
\coordinate (C2) at ([shift=(180:8)]C);
\coordinate (D2) at ([shift=(180:8)]D);
\coordinate (E2) at ([shift=(180:8)]E);
\coordinate (F2) at ([shift=(180:8)]F);
\draw[thick] (A2) -- coordinate[midway](mab2)(B2);
\draw[thick] (B2) -- coordinate[midway](mbc2)(C2);
\draw[thick] (C2) -- coordinate[midway](mcd2)(D2);
\draw[thick] (D2) -- coordinate[midway](mde2)(E2);
\draw[thick] (E2) -- coordinate[midway](mef2)(F2);
\draw[thick] (F2) -- coordinate[midway](mfa2)(A2);

\path  (A2) --coordinate[midway](cen) (D2);
\draw [thick,red] (cen) -- node[at end,left]{$j_1$}($(cen)!2.5cm!(mab2)$);
\draw [thick,red] (cen) -- node[at end,above]{$j_2$}($(cen)!2.5cm!(mbc2)$);
\draw [thick,red] (cen) -- node[at end,above]{$j_3$}($(cen)!2.5cm!(mcd2)$);
\draw [thick,red] (cen) -- node[at end,left]{$j_4$}($(cen)!2.5cm!(mde2)$);
\draw [thick,red] (cen) -- node[at end,above]{$j_5$}($(cen)!2.5cm!(mef2)$);
\draw [thick,red] (cen) -- node[at end,above]{$j_6$}($(cen)!2.5cm!(mfa2)$);
\draw[red] (cen) node{$\bullet$};

\node at (cI) {$ \id_{\cH}=\sum_{k_1,k_2,k_3}| \cI_{k_1k_2k_3}\rangle\langle\cI_{k_1k_2k_3}|$};
	\end{tikzpicture}
\caption{Splitting a 6-valent node into four 3-valent nodes {\it (in red)} and its corresponding change on the boundary 2-cell {\it (in black)}, which in this case is to split a hexagon into four triangles by adding three internal edges {\it (dashed lines)}. The splitting is not unique and each way of splitting corresponds to choosing one set of basis $\{a\}$ as in \eqref{eq:SU2_integral}. Here the basis is labelled by the three internal spins $k_1,k_2,k_3$ whose coupling with the $j_1,\cdots,j_6$ is as shown in the right.}

\label{fig:hexagon}
\end{figure}
Performing the group integration to all the dual edges, one ends up with a $\{3nj\}$-symbol (up to a sign) for each 3-cell given a basis of the intertwiner $|\cI^{e^*}_a\rangle \langle \cI^{e^*}_a|$ for each boundary face. This $\{3nj\}$-symbol defines a {\it spin network evaluation} (we will see why it is called so in the next subsection), which is the contraction of the intertwiner basis $\cI^{e^*}_a$ on all the dual edges incident to $v^*$:
\be
\{3nj\}_{v^*,a}= \tr_{\otimes_{f^*}\cV^{j_{f^*}}} \left[ \bigotimes_{e^*}\cI^{e^*}_a \otimes \bigotimes_{f^*}\id_{\cV^{j_{f^*}}} \right]\,.
\ee
 Therefore, the discrete partition function can be written as the gluing of $\{3nj\}$-symbols with edge amplitudes simply given by a sign. Symbolically, 
\be
Z[\triangle,\psi_{\Gamma}^j]= \sum_{\{j_{f^*}\}}\, \prod_{f^*}(-1)^{2j_{f^*}}d_{j_{f^*}}\,\sum_{\{a\}}\prod_{e^*}(-1)^{\sum\limits_{f^*|e^*\in f^*}j_{f^*} + 2J_a}\, \prod_{v^*} \{3nj\}_{v^*,a}\,,
\label{eq:partition_cell}
\ee
where the power of the edge amplitude is the sum of spins over all the dual faces whose boundary contains a given dual edge $e^*$ and twice the spins corresponding to the $a$-basis of the intertwiner on the dual edge (for the example in fig.\ref{fig:hexagon}, $J_a=k_1+k_2+k_3$). The double counting for $J_a$ is due to the fact that when separating the boundary polygon into triangles by adding internal edges, each internal edge is on the boundary of two triangles. 
 
In the case of triangulation, one encounters intertwiners for 3-valent nodes which are one-dimensional. The integration \eqref{eq:SU2_integral} simply becomes
\be
\int_{\SU(2)}\rd g_{e^*}\,
 D^{j_1}_{m_1n_1}(g_{e^*})D^{j_2}_{m_2n_2}(g_{e^*})D^{j_3}_{m_3n_3}(g_{e^*})
=\mat{ccc}{j_1 & j_2 & j_3 \\ m_1 & m_2 & m_3}
 \mat{ccc}{j_1 & j_2 & j_3 \\ n_1 & n_2 & n_3}\,,
\ee
which is the product of two normalized Clebsh-Gordan coefficients, or equivalently the $\{3j\}$-symbols. The re-coupling of the $\{3j\}$-symbols ends up with $\{6j\}$-symbols, each associated to a tetrahedron $T$, or equivalently a dual vertex $v^*$. For a tetrahedron with the notation in fig.\ref{fig:spinnetwork} and each edge $e_i$ dressed with a spin $j_i$, the $\{6j\}$-symbol is given by
\be
\Mat{ccc}{j_1 & j_2 & j_3 \\ j_4 & j_5 & j_6} =
\sum_{m_i} 
(-1)^{\sum_{i=1}^6 (j_i-m_i)}
\mat{ccc}{j_1 & j_2 & j_3 \\ m_1 & m_2 & -m_3} 
\mat{ccc}{j_1 & j_5 & j_6 \\ -m_1 & m_5 & m_6}
\mat{ccc}{j_4 & j_2 & j_6 \\ -m_4 & -m_2 & -m_6}
\mat{ccc}{j_4 & j_5 & j_3 \\ m_4 & -m_5 & m_3}   \,.
\ee
It ends up with a state-sum formulation of the discrete partition function \eqref{eq:PI_disc}, $\ie$ the Ponzano-Regge model, 
\be
Z[{\bf T},\psi_{\Gamma}^j]= \sum_{\{j_{f^*}\}}\, \prod_{f^*} (-1)^{2j_{f^*}}d_{j_{f^*}}\,\prod_{e^*}(-1)^{\sum_{i=1}^3 j_i}\, \prod_{v^*}  
\Mat{ccc}{j_1 & j_2 & j_3 \\ j_4 & j_5 & j_6}_{v^*}\,.
\label{eq:partition_6j}
\ee
It easily reads that the vertex amplitude $\cA_{v^*}$ is the $\{6j\}$-symbol associated to the tetrahedron dual to $v^*$, and the edge amplitude $\cA_{e^*}$ is a sign given by the spins on the sides of the triangle dual to $e^*$, 
and the face amplitude $\cA_{f^*}$ is the dimension $d_{j_{f^*}}$ of the spin representation space associated to the face $f^*$. See $\eg$ \cite{Barrett:2008wh} for detailed explanation of the sign factors. 
When there's no boundary, the edge amplitude term can be absorbed in the vertex amplitude \cite{Barrett:2008wh}, then the state-sum can be written as
\be
Z[{\bf T}]= \sum_{\{j_{f^*}\}}\, \prod_{f^*} (-1)^{2j_{f^*}}d_{j_{f^*}}\, \prod_{v^*}  (-1)^{\sum_{i=1}^6 2 j_i}
\Mat{ccc}{j_1 & j_2 & j_3 \\ j_4 & j_5 & j_6}_{v^*}\,.
\ee
Geometrically, the vertex amplitude describes a tetrahedron with edge lengths specified by the spins in the $\{6j\}$-symbol. The edge amplitude determines that the gluing of two adjacent tetrahedra is performed by matching the side lengths of the triangles, thus the full 2D geometrical information of the triangles. This trivial way of gluing can be viewed as resulting from the flatness of the manifold, imposed by $\delta(g_{e})$ for all the edges of the triangulation. The face amplitude is simply a weight factor, which is important only in the quantum regime.

\subsection{Topological invariance and local holography}
\label{sec:local_holography_SN}
The Ponzano-Regge state-sum formula
\begin{align}
Z[{\bf T},\psi_{\Gamma}^j]
&=  C[\partial {\bf T} ] \int_{\SU(2)} \prod_{e^*\notin (\partial\cM)^*}\rd g_{e^*} \,\prod_{f^*}\delta(\overrightarrow{\prod}_{e^*\in \partial f^*}g_{e^*})  \label{eq:group_form}\\
& =\sum_{\{j_{f^*}\}}\, \prod_{f^*}(-1)^{2j_{f^*}} d_{j_{f^*}}\, \prod_{e^*}(-1)^{\sum_{i=1}^3 j_i}\,\prod_{v^*} 
\Mat{ccc}{j_1 & j_2 & j_3 \\ j_4 & j_5 & j_6}_{v^*} \label{eq:spin_form}
\end{align}
under the Pachner moves in the bulk, which is the discrete version of the bulk diffeomorphism, but only depends on the topology of the manifold $\cM$ \cite{Ooguri:1991ni} and the boundary states. 
This reflects the holographic nature of the spinfoam model for 3D quantum gravity. 
Consider the smallest 3D space-time block with trivial topology -- a tetrahedron, 
one can define the boundary state to be the spin network state (defined in \ref{subsubsec:local_amplitude}) coming from the loop quantum gravity. 
In this way, instead of splitting a delta distribution into local amplitudes as above, one can start from the local amplitude ansatz \eqref{eq:amplitudes_boundary} and construct first the vertex amplitudes then recover total amplitude \eqref{eq:PI_disc} by the chosen edge and face amplitudes.

In this subsection, we first review the Pachner moves for the Ponzano-Regge state-sum formula \eqref{eq:group_form} or \eqref{eq:spin_form} to show the topological invariance of the total amplitude then reconstruct the Ponzano-Regge amplitude from the local amplitude ansatz. 

\subsubsection{Topological invariance from Pachner moves}
In 3D, there are two types of Pachner moves, namely the $2-3$ moves and the $1-4$ moves as shown in fig.\ref{fig:Pachner}. To prove the topological invariance, one can either start from the group formulation \eqref{eq:group_form} and apply the change of variable method  
or start from the spin formulation \eqref{eq:spin_form} and apply the recursion relation of $\{6j\}$-symbols.  
\begin{figure}[h!]
\centering
\begin{minipage}{0.35\textwidth}
	\centering
\begin{tikzpicture}[scale=1.7]
\coordinate (O1) at (0,0,0);

\coordinate (A1) at (0,1.061,0);
\coordinate (B1) at (0.15,-0.354,1);
\coordinate (C1) at (-0.866,-0.354,-0.5);
\coordinate (D1) at (0.866,-0.354,-0.5);
\coordinate (E1) at (0,-2*0.354-1.061,0);

\coordinate (A2) at ([shift=(0:2)]A1);
\coordinate (B2) at ([shift=(0:2)]B1);
\coordinate (C2) at ([shift=(0:2)]C1);
\coordinate (D2) at ([shift=(0:2)]D1);
\coordinate (E2) at ([shift=(0:2)]E1);

\coordinate (aa) at (0.25,-2, 0.25);
\coordinate (bb) at ([shift=(0:2)]aa);

\draw (aa) node{${\bf T}_2$} ;
\draw (bb) node{${\bf T}_3$};

\draw (A1) -- node[pos=.4,left]{$j_2$} (B1) ;
\draw (A1) -- node[midway,left]{$j_1$}(C1);
\draw (A1) -- node[midway,left]{$j_6$}(D1);
\draw (B1) -- node[midway,right]{$j_3$}(C1);
\draw[dashed] (C1) -- node[pos=0.6,above]{$j_5$}(D1);
\draw (D1) -- node[midway,below]{$j_4$}(B1);
\draw (E1) -- node[midway]{$j_8$}(B1) ;
\draw (E1) -- node[midway,left]{$j_7$}(C1);
\draw (E1) -- node[midway,right]{$j_9$}(D1);

\draw (A2) --node[pos=.4,left]{$j_2$} (B2) ;
\draw (A2) -- node[midway,left]{$j_1$}(C2);
\draw (A2) --node[midway,left]{$j_6$} (D2);
\draw (B2) -- node[midway,right]{$j_3$}(C2);
\draw[dashed] (C2) -- node[pos=0.6,above]{$j_5$}(D2);
\draw[red] (A2) -- node[midway,right]{$J$}(E2);
\draw (D2) -- node[midway,below]{$j_4$}(B2);
\draw (E2) -- node[midway]{$j_8$}(B2) ;
\draw (E2) -- node[midway,left]{$j_7$}(C2);
\draw (E2) -- node[midway,right]{$j_9$}(D2);
\end{tikzpicture}
\subcaption{}
\label{fig:Pachner32}
\end{minipage}
\quad
\begin{minipage}{0.55\textwidth}
	\centering
\begin{tikzpicture}[scale=2]
\coordinate (O1) at (0,0,0);

\coordinate (A1) at (0,1.061,0);
\coordinate (B1) at (0.15,-0.354,1);
\coordinate (C1) at (-0.866,-0.354,-0.5);
\coordinate (D1) at (0.866,-0.354,-0.5);
\coordinate (E1) at (0,0.254,0);

\coordinate (A2) at ([shift=(180:2)]A1);
\coordinate (B2) at ([shift=(180:2)]B1);
\coordinate (C2) at ([shift=(180:2)]C1);
\coordinate (D2) at ([shift=(180:2)]D1);

\coordinate (aa) at (0.25,-1, 0.25);
\coordinate (bb) at ([shift=(180:2)]aa);

\draw (aa) node{${\bf T}_4$} ;
\draw (bb) node{${\bf T}_1$};

\draw (A1) -- node[pos=.4,left]{$j_2$} (B1) ;
\draw (A1) -- node[midway,left]{$j_1$}(C1);
\draw (A1) -- node[midway,left]{$j_6$}(D1);
\draw (B1) -- node[midway,right]{$j_3$}(C1);
\draw[dashed] (C1) -- node[midway]{$j_5$}(D1);
\draw (D1) -- node[midway,below]{$j_4$}(B1);
\draw[red] (E1) -- node[midway,right]{$l_3$}(A1) ;
\draw[red] (E1) -- node[midway,right]{$l_1$}(B1) ;
\draw[red] (E1) -- node[midway]{$l_2$}(C1);
\draw[red] (E1) -- node[midway,above]{$J$}(D1);

\draw (A2) --node[pos=.4,left]{$j_2$} (B2) ;
\draw (A2) -- node[midway,left]{$j_1$}(C2);
\draw (A2) --node[midway,left]{$j_6$} (D2);
\draw (B2) -- node[midway,right]{$j_3$}(C2);
\draw[dashed] (C2) -- node[midway]{$j_5$}(D2);
\draw (D2) -- node[midway,below]{$j_4$}(B2);

\draw[red] (E1) node{$\bullet$};

\end{tikzpicture}
\subcaption{}
\label{fig:Pachner41}
\end{minipage}
\caption{(a) 2-3 Pachner move: adding an internal edge. (b) 1-4 Pachner move: adding an internal vertex.}
\label{fig:Pachner}
\end{figure}
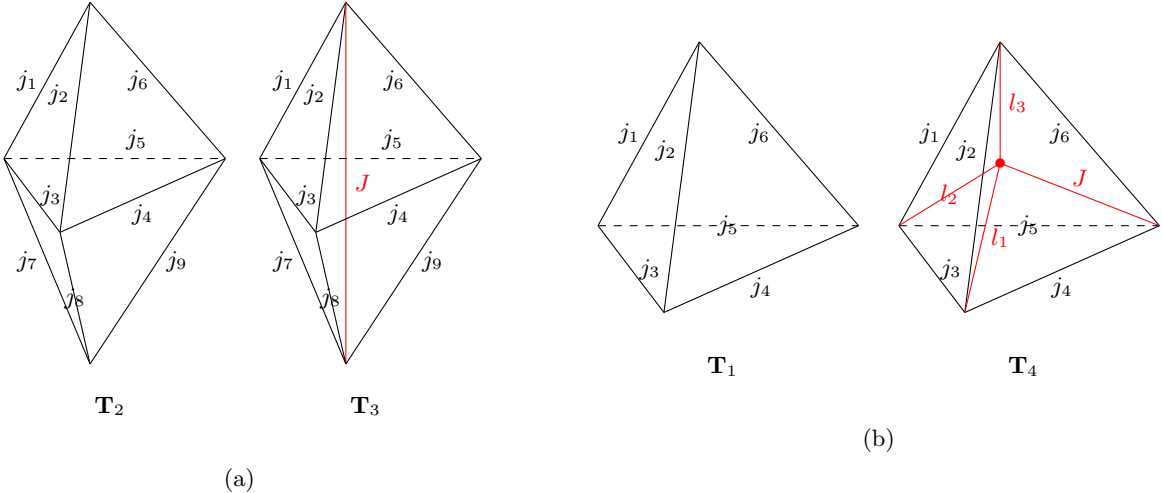

To apply the change of variable method, one simply uses the invariance property of the $\SU(2)$ group measure under (left and right) $\SU(2)$ transformation. 
In the case of $2-3$ moves, the dual triangulation ${\bf T}^*_2$ for two tetrahedra ${\bf T}_2$ and the dual triangulation ${\bf T}^*_3$ for three tetrahedra ${\bf T}_3$ are given in the left panel of fig.\ref{fig:trll_2} and fig.\ref{fig:trll_3} respectively (projected onto a plane).
For both ${\bf T}^*_2$ and ${\bf T}^*_3$, one can choose a base dual vertex and redefine the holonomies starting from the base dual vertex. Using the notation in fig.\ref{fig:trll_2}, one can write the partition function \eqref{eq:group_form} for ${\bf T}_2$ by transforming $h_i \rightarrow H_i = h_i k^{-1}$ as
\be\begin{split}
Z[{\bf T}_2,\psi_\Gamma^j ]
=&C[\partial {\bf T}_2 ]
\int_{\SU(2)}\left(\prod_{e'^*}\rd g_{e'^*} \right) \prod_{f'^*}\delta(\overrightarrow{\prod} g_{e'^*})
\int_{\SU(2)}\left(\prod_{i=1}^3 \rd g_i\rd h_i\right)\rd k\, \delta(g_1k h_1^{-1}\cdots)\delta(g_2kh_2^{-1}\cdots) \\
&\delta(g_3kh_3^{-1}\cdots)\delta(g_1g_2^{-1}\cdots)\delta(g_2g_3^{-1}\cdots)\delta(g_3g_1^{-1}\cdots)\delta(h_1h_2^{-1}\cdots)\delta(h_2h_3^{-1}\cdots)\delta(h_3h_1^{-1}\cdots)\\
=& C[\partial {\bf T}_2 ]
\int_{\SU(2)}\left(\prod_{e'^*}\rd g_{e'^*} \right) \prod_{f'^*}\delta(\overrightarrow{\prod} g_{e'^*})
\int_{\SU(2)}\left(\prod_{i=1}^3 \rd g_i\rd H_i\right) \delta(g_1 H_1^{-1}\cdots)\delta(g_2H_2^{-1}\cdots) \\
&\delta(g_3H_3^{-1}\cdots)\delta(g_1g_2^{-1}\cdots)\delta(g_2g_3^{-1}\cdots)\delta(g_3g_1^{-1}\cdots)\delta(H_1H_2^{-1}\cdots)\delta(H_2H_3^{-1}\cdots)\delta(H_3H_1^{-1}\cdots)\,,
\end{split}
\label{eq:Pachner_T2}
\ee
where we have denoted the irrelevant part of the integration with primes and holonomies not in ${\bf T}_2$ with $\cdots$.
 
For ${\bf T}_3$, one can transform 
\be\left|\ba{l}
g_1 \rightarrow G_1 = g_1k_1^{-1}\,,\quad \,\,\,\,h_1\rightarrow H_1=h_1k_1^{-1} \\
g_2 \rightarrow G_2 = g_2k_2\,,\quad h_2\rightarrow H_2=h_2k_2\\
k_3\rightarrow K_3=k_1 k_3 k_2
\ea\right.
\nn\ee
thus the partition function for ${\bf T}_3$ can be rewritten as (we again denote the irrelevant part of the integration with primes.)
\be\begin{split}
Z[{\bf T}_3,\psi_\Gamma^j ]
=&C[\partial {\bf T}_3 ]
\int_{\SU(2)}\left(\prod_{e'^*}\rd g_{e'^*} \right) \prod_{f'^*}\delta(\overrightarrow{\prod} g_{e'^*})
\int_{\SU(2)}\left(\prod_{i=1}^3 \rd g_i\rd h_i \rd k_i\right)
\delta(g_1k_3 g_2^{-1}\cdots)\delta(h_1k_3h_2^{-1}\cdots)\\
&\delta(g_2k_2g_3^{-1}\cdots)\delta(h_2k_2h_3^{-1}\cdots)
\delta(g_3k_1g_1^{-1}\cdots)\delta(h_3k_1h_1^{-1}\cdots)
\delta(g_1h_1^{-1}\cdots)\delta(g_2h_2^{-1}\cdots)\delta(g_3h_3^{-1}\cdots)
\delta(k_1 k_3 k_2) \\
=& C[\partial {\bf T}_3 ]
\int_{\SU(2)}\left(\prod_{e'^*}\rd g_{e'^*} \right) \prod_{f'^*}\delta(\overrightarrow{\prod} g_{e'^*})
\int_{\SU(2)} \left(\prod_{i=1}^2\rd G_i\rd H_i\right)\rd g_3\rd h_3 \rd K_3 
\delta(G_1K_3 G_2^{-1}\cdots)\delta(H_1K_3H_2^{-1}\cdots)\\
&\delta(G_2g_3^{-1}\cdots)\delta(H_2h_3^{-1}\cdots)
\delta(g_3G_1^{-1}\cdots)\delta(h_3H_1^{-1}\cdots)
\delta(G_1H_1^{-1}\cdots)\delta(G_2H_2^{-1}\cdots)\delta(g_3h_3^{-1}\cdots)
\delta(K_3)\\
=& C[\partial {\bf T}_3 ]
\int_{\SU(2)}\left(\prod_{e'^*}\rd g_{e'^*} \right) \prod_{f'^*}\delta(\overrightarrow{\prod} g_{e'^*})
\int_{\SU(2)} \left(\prod_{i=1}^2\rd G_i\rd H_i\right)\rd g_3\rd h_3
\delta(G_1G_2^{-1}\cdots)\delta(H_1H_2^{-1}\cdots)\\
&\delta(G_2g_3^{-1}\cdots)\delta(H_2h_3^{-1}\cdots)
\delta(g_3G_1^{-1}\cdots)\delta(h_3H_1^{-1}\cdots)
\delta(G_1H_1^{-1}\cdots)\delta(G_2H_2^{-1}\cdots)\delta(g_3h_3^{-1}\cdots)\,.
\end{split}
\label{eq:Pachner_T3}
\ee
to arrive at the last equation, we have used $\delta(K_3)$ to eliminate $K_3$ in the expression. 
The boundary term $C[\partial {\bf T}_2 ]$ and $C[\partial {\bf T}_3 ]$ are the same since ${\bf T}_2$ and ${\bf T}_3$ possess the same boundary, so as the irrelevant parts. Therefore, \eqref{eq:Pachner_T2} and \eqref{eq:Pachner_T3} are exactly the same (although they are written with different notations).
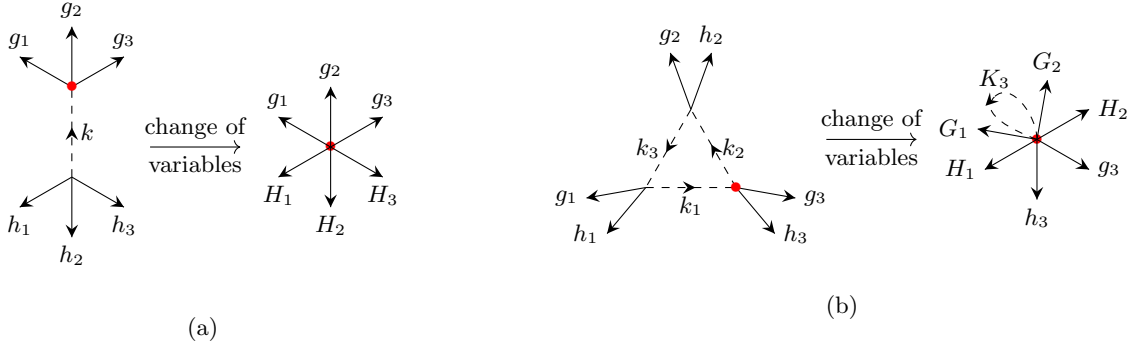
\begin{figure}[h!]
\centering
\begin{minipage}{0.45\textwidth}
	\centering
\begin{tikzpicture}[scale=0.8]
\coordinate (A) at (0,0);
\coordinate (B) at (0,-1.5);

\coordinate (A1) at ([shift=(30:1)]A);
\coordinate (A2) at ([shift=(90:1)]A);
\coordinate (A3) at ([shift=(150:1)]A);
\coordinate (B1) at ([shift=(-30:1)]B);
\coordinate (B2) at ([shift=(-90:1)]B);
\coordinate (B3) at ([shift=(-150:1)]B);

\coordinate (a) at (4.3,-1);
\coordinate (a1) at ([shift=(30:1)]a);
\coordinate (a2) at ([shift=(90:1)]a);
\coordinate (a3) at ([shift=(150:1)]a);
\coordinate (b1) at ([shift=(-30:1)]a);
\coordinate (b2) at ([shift=(-90:1)]a);
\coordinate (b3) at ([shift=(-150:1)]a);

\draw[postaction={decorate},decoration={markings,mark={at position 1 with {\arrow[scale=1.5,>=stealth]{>}}}}] (A) -- (A3) node[at end,above]{$g_{1}$} ;
\draw[postaction={decorate},decoration={markings,mark={at position 1 with {\arrow[scale=1.5,>=stealth]{>}}}}] (A) -- (A2) node[at end,above]{$g_{2}$} ;
\draw[postaction={decorate},decoration={markings,mark={at position 1 with {\arrow[scale=1.5,>=stealth]{>}}}}] (A) -- (A1) node[at end,above]{$g_{3}$} ;
\draw[postaction={decorate},decoration={markings,mark={at position 1 with {\arrow[scale=1.5,>=stealth]{>}}}}] (B) -- (B3) node[at end,below]{$h_{1}$} ;
\draw[postaction={decorate},decoration={markings,mark={at position 1 with {\arrow[scale=1.5,>=stealth]{>}}}}] (B) -- (B2) node[at end,below]{$h_{2}$} ;
\draw[postaction={decorate},decoration={markings,mark={at position 1 with {\arrow[scale=1.5,>=stealth]{>}}}}] (B) -- (B1) node[at end,below]{$h_{3}$} ;

\draw[dashed, postaction={decorate},decoration={markings,mark={at position 0.55 with {\arrow[scale=1.5,>=stealth]{>}}}}] (B) -- (A) node[midway,right]{$k$} ;

\draw[red] (A) node{$\bullet$};\draw[red] (a) node{$\bullet$};

\draw[postaction={decorate},decoration={markings,mark={at position 1 with {\arrow[scale=1.5,>=stealth]{>}}}}] (a) -- (a3) node[at end,above]{$g_{1}$} ;
\draw[postaction={decorate},decoration={markings,mark={at position 1 with {\arrow[scale=1.5,>=stealth]{>}}}}] (a) -- (a2) node[at end,above]{$g_{2}$} ;
\draw[postaction={decorate},decoration={markings,mark={at position 1 with {\arrow[scale=1.5,>=stealth]{>}}}}] (a) -- (a1) node[at end,above]{$g_{3}$} ;
\draw[postaction={decorate},decoration={markings,mark={at position 1 with {\arrow[scale=1.5,>=stealth]{>}}}}] (a) -- (b3) node[at end,below]{$H_{1}$} ;
\draw[postaction={decorate},decoration={markings,mark={at position 1 with {\arrow[scale=1.5,>=stealth]{>}}}}] (a) -- (b2) node[at end,below]{$H_{2}$} ;
\draw[postaction={decorate},decoration={markings,mark={at position 1 with {\arrow[scale=1.5,>=stealth]{>}}}}] (a) -- (b1) node[at end,below]{$H_{3}$} ;

\coordinate (l) at (1.3,-1); \coordinate (r) at (2.8,-1);
\draw[->] (l) -- node[midway, above]{change of}node[midway,below]{variables} (r);
\end{tikzpicture}
\subcaption{}
\label{fig:trll_2}
\end{minipage}
\quad
\begin{minipage}{0.45\textwidth}
\centering
\begin{tikzpicture}[scale=0.8]
\coordinate (A) at (0,0);
\coordinate (B) at ([shift=(120:1.5)]A);
\coordinate (C) at ([shift=(180:1.5)]A);

\coordinate (A1) at ([shift=(-10:1)]A);
\coordinate (A2) at ([shift=(-50:1)]A);
\coordinate (B1) at ([shift=(70:1)]B);
\coordinate (B2) at ([shift=(110:1)]B);
\coordinate (C1) at ([shift=(190:1)]C);
\coordinate (C2) at ([shift=(230:1)]C);

\coordinate (a) at (5,0.8);
\coordinate (a1) at ([shift=(30:1)]a);
\coordinate (a2) at ([shift=(80:1)]a);
\coordinate (a3) at ([shift=(170:1)]a);
\coordinate (b1) at ([shift=(-30:1)]a);
\coordinate (b2) at ([shift=(-90:1)]a);
\coordinate (b3) at ([shift=(-150:1)]a);
\coordinate (c) at ([shift=(135:1)]a);

\draw[postaction={decorate},decoration={markings,mark={at position 1 with {\arrow[scale=1.5,>=stealth]{>}}}}] (A) -- (A1) node[at end,right]{$g_{3}$} ;
\draw[postaction={decorate},decoration={markings,mark={at position 1 with {\arrow[scale=1.5,>=stealth]{>}}}}] (A) -- (A2) node[at end,right]{$h_{3}$} ;
\draw[postaction={decorate},decoration={markings,mark={at position 1 with {\arrow[scale=1.5,>=stealth]{>}}}}] (B) -- (B1) node[at end,above]{$h_{2}$} ;
\draw[postaction={decorate},decoration={markings,mark={at position 1 with {\arrow[scale=1.5,>=stealth]{>}}}}] (B) -- (B2) node[at end,above]{$g_{2}$} ;
\draw[postaction={decorate},decoration={markings,mark={at position 1 with {\arrow[scale=1.5,>=stealth]{>}}}}] (C) -- (C1) node[at end,left]{$g_{1}$} ;
\draw[postaction={decorate},decoration={markings,mark={at position 1 with {\arrow[scale=1.5,>=stealth]{>}}}}] (C) -- (C2) node[at end,left]{$h_{1}$} ;

\draw[dashed, postaction={decorate},decoration={markings,mark={at position 0.55 with {\arrow[scale=1.5,>=stealth]{>}}}}] (A) -- (B) node[midway,right]{$k_2$} ;
\draw[dashed, postaction={decorate},decoration={markings,mark={at position 0.55 with {\arrow[scale=1.5,>=stealth]{>}}}}] (B) -- (C) node[midway,left]{$k_3$} ;
\draw[dashed, postaction={decorate},decoration={markings,mark={at position 0.55 with {\arrow[scale=1.5,>=stealth]{>}}}}] (C) -- (A) node[midway,below]{$k_1$} ;

\draw[red] (A) node{$\bullet$};\draw[red] (a) node{$\bullet$};

\draw[postaction={decorate},decoration={markings,mark={at position 1 with {\arrow[scale=1.5,>=stealth]{>}}}}] (a) -- (a3) node[at end,left]{$G_1$} ;
\draw[postaction={decorate},decoration={markings,mark={at position 1 with {\arrow[scale=1.5,>=stealth]{>}}}}] (a) -- (a2) node[at end,above]{$G_2$} ;
\draw[postaction={decorate},decoration={markings,mark={at position 1 with {\arrow[scale=1.5,>=stealth]{>}}}}] (a) -- (a1) node[at end,right]{$H_2$} ;
\draw[postaction={decorate},decoration={markings,mark={at position 1 with {\arrow[scale=1.5,>=stealth]{>}}}}] (a) -- (b3) node[at end,left]{$H_{1}$} ;
\draw[postaction={decorate},decoration={markings,mark={at position 1 with {\arrow[scale=1.5,>=stealth]{>}}}}] (a) -- (b2) node[at end,below]{$h_3$} ;
\draw[postaction={decorate},decoration={markings,mark={at position 1 with {\arrow[scale=1.5,>=stealth]{>}}}}] (a) -- (b1) node[at end,right]{$g_3$} ;

\coordinate (l) at (1.5,0.8); \coordinate (r) at (3,0.8);
\draw[->] (l) -- node[midway, above]{change of}node[midway,below]{variables} (r);

\draw[dashed, postaction={decorate},decoration={markings,mark={at position 1 with {\arrow[scale=1.5,>=stealth]{<}}}}]  (a) to [out=160,in=225] node[at end,above]{$K_3$}  (c);
\draw[dashed]  (a) to [out=90,in=45] (c);
\end{tikzpicture}
\subcaption{}
\label{fig:trll_3}	
\end{minipage}
\caption{(a) Dual triangulation ${\bf T}^*_2$ and the effective graph after acting gauge transformation on $h_i$ and redefine them from the base dual vertex {\it (in red)} as $h_i\rightarrow H_i=h_ik^{-1},i=1,\cdots,3$.
 (b) Dual triangulation ${\bf T}^*_3$ and the effective graph after acting gauge transformation on $g_{1,2},h_{1,2},k_3$ and redefine them starting from the base dual vertex as $(g_1,h_1,g_2,h_2,k_3)\rightarrow (G_1=g_1k_1^{-1},H_1=h_1k_1^{-1},G_2=g_2k_2,H_2=h_2k_2, k_3\rightarrow K_3 =k_1k_3k_2)$. The extra loop {\it (dashed)} for $\delta(K_3)$ can be eliminated by the integration over the delta function in the spinfoam amplitude: $\int_{\SU(2)}\rd K_3\delta(K_3)=1$. }
\label{fig:tree}
\end{figure}

In the case of $1-4$ moves, one can apply the same method on ${\bf T}_4$, resulting in a divergent term $\delta_{\SU(2)}(\id)$. 
Indeed, the holonomies surrounding the four edges incident to the internal vertex are on four loops patched together to form a 2-sphere. Thus only three of the four delta functions of these holonomies are independent, which gives one extra delta function evaluated on identity.
This divergence corresponds to the translational symmetry of the internal vertex of ${\bf T}_4$, which is the infinite gauge volume of the Lie algebra $\su(2)$ \cite{Freidel:2002dw,Freidel:2004nb}.
To remove the divergence, one can go through the partial gauge fixing method \cite{Freidel:2002xb} (or equivalently the Fadeev-Popov gauge fixing procedure illustrated in \cite{Freidel:2002dw,Freidel:2004vi}).

On the other hand, starting from the spin formulation \eqref{eq:spin_form}, one can apply the {\it Biedenharn-Elliott Identity}
\be
\Mat{ccc}{j_1 & j_2 & j_3 \\ j_4 & j_5 & j_6}
\Mat{ccc}{j_7 & j_8 & j_3 \\ j_4 & j_5 & j_9}
=\sum_{j}(-1)^{J+\sum_{i=1}^9 j_i} d_J
\Mat{ccc}{j_1 & j_6 & j_6 \\ j_9 & j_7 & J}
\Mat{ccc}{j_2 & j_6 & j_4 \\ j_9 & j_8 & J}
\Mat{ccc}{j_1 & j_2 & j_3 \\ j_8 & j_7 & J}
\label{eq:EI}
\ee
and directly show that the partition function after $2-3$ moves is unchanged. For the $1-4$ move case, the corresponding identity for $\{6j\}$-symbols is \cite{Bonzom:2009zd}
\be
d_{J} \Mat{ccc}{j_1 & j_2 & j_3 \\ j_4 & j_5 & j_6}
=\sum_{l_i} (-1)^{\sum_{i=1}^6 j_i +\sum_{i=1}^3 l_i + J } d_{l_1}d_{l_2}d_{l_3}
  \Mat{ccc}{j_1 & j_2 & j_3 \\ l_1 & l_2 & l_3}
 \Mat{ccc}{j_6 & j_5 & j_1 \\ l_2 & l_3 & J}
 \Mat{ccc}{j_4 & j_2 & j_6 \\ l_3 & J & l_1}
 \Mat{ccc}{j_3 & j_4 & j_5 \\ J & l_1 & l_2}\,,
 \label{eq:41_identity}
\ee
which is true for any admissible $J$ corresponding to the length of one of the internal edges as shown in fig.\ref{fig:Pachner41}. \eqref{eq:41_identity} is the partial gauge fixing version of the apparent but divergent result \cite{Bonzom:2009zd}
\be
\sum_J d_{J}^2 \Mat{ccc}{j_1 & j_2 & j_3 \\ j_4 & j_5 & j_6}
=\sum_{l_i,J} (-1)^{\sum_{i=1}^6 j_i +\sum_{i=1}^3 l_i + J } d_{l_1}d_{l_2}d_{l_3}d_{J}
  \Mat{ccc}{j_1 & j_2 & j_3 \\ l_1 & l_2 & l_3}
 \Mat{ccc}{j_6 & j_5 & j_1 \\ l_2 & l_3 & J}
 \Mat{ccc}{j_4 & j_2 & j_6 \\ l_3 & J & l_1}
 \Mat{ccc}{j_3 & j_4 & j_5 \\ J & l_1 & l_2}\,,
 \label{eq:41_identity_diverge}
\ee
which is obtained by writing the amplitude for the four tetrahedra and then applying the Biedenharn-Elliott identity and the orthogonal relations of $\{6j\}$-symbols. 
The divergence comes from $\sum_{J}d_{J}^2=\delta_{\SU(2)}(\id)$ whose degree is related to the topology of $\cM$ (see \cite{Barrett:2008wh,Bonzom:2010ar,Bonzom:2010zh} for discussion). Regularization was originally performed by introducing a cut-off on spin $J$ \cite{Ponzano:1968se}. It was then realized that the divergence is correspondent to the $\su(2)$ gauge that generates the translational symmetry of the internal vertex thus the regularization can be performed by a partial gauge fixing procedure \cite{Freidel:2002dw,Freidel:2004vi}. 

These two methods can be straightforwardly extended to arbitrary cellular decomposition with $\{3nj\}$-symbols, from which one can show that the Ponzano-Regge state-sum model is topological invariant and only depends on the boundary states. 

\medskip
\subsubsection{Locally-holographic amplitude}
\label{subsubsec:local_amplitude}

Due to the topological invariance of the Ponzano-Regge model, the total amplitude is determined by the boundary state for a manifold $\cM$ of fixed topology. Such a boundary state can be provided by the spin network state which is defined in the loop quantum gravity framework. 
We first work on an arbitrary cellular decomposition of the manifold in order to give a general construction of the boundary spin network state. The results from the triangulation naturally follow, which will be the building blocks of the original Ponzano-Regge model. 

Consider a cellular decomposition of $\cM$ whose boundary is a union of 2-cells. We construct the oriented graph $\Gamma$ dual to this boundary cellular decomposition made up with $|L|$ links $l$'s, $|N|$ nodes $n$'s and $|P|$ plaquettes $p$'s. 
On $\Gamma$, we associate a spin $j_l$ to each oriented link $l$ and an intertwiner $\iota^n$ to each node $n$. This specifies a covariant phase space on $\Gamma$ \footnotemark{}. 
\footnotetext{
The phase space is called covariant because the boundary $\partial \cM$ of the $\cM$ is not specified to be a space-like slice $\Sigma$ of $\cM$ but a general 2D hypersurface. The spin network state on such a graph $\Gamma$ is also called the projected spin network state, introduced in \cite{Livine:2002ak}, which is used in the construction of the covariant LQG, see $\eg$ \cite{Alexandrov:2002br,Livine:2006ix}.
}
An intertwiner is an $\SU(2)$-invariant map from the tensor product of the spin representation spaces (or the dual spin representation space), associated to the links incident to the same node, to the trivial space:
\be
\iota^n_{(j_l)}:\left(\bigotimes_{l|s(l)=n} \cV^{j_l}\right)\otimes \left(\bigotimes_{l|t(l)=n} \cV^{j_l^*}\right)\rightarrow 0 \,,
\label{eq:intertwiner}
\ee
where $s(l)$ and $t(l)$ denote respectively the source and target node of the link $l$. 
Or equivalently, the basis of the intertwiner can be written as the tensor product of the magnetic basis followed by a group averaging,
\be
\iota^n_{(j_l)}(m_l)|0\rangle = \int_{\SU(2)} \rd h_n 
\left( \bigotimes_{l|t(l)=n} \langle j_l,n_l| \,h_n^{-1}\right) \otimes 
\left( \bigotimes_{l|s(l)=n} h_n\, | j_l,m_l\rangle\right)\,,
\label{eq:intertwiner_jm}
\ee
which is indeed $\SU(2)$-invariant.  
Finally, the spin network state $s_\Gamma^{(j_l,\iota_n)} $ on $\Gamma$ is simply defined as the tensor product of the intertwiners. Conventionally the spin network state is evaluated on the group elements $\{g_l\}\in \SU(2)^{|L|}$ associated to the links, thus
\be
s_\Gamma^{\{j_l,\iota_n\}}(g_l)=\sum_{m_l,n_l}\prod_l \langle j_l, m_l|g_l|j_l,n_l\rangle \,
\prod_n \langle \otimes_{l|t(l)=n}\,j_l,n_l|\iota^n_{(j_l)}|\otimes_{l|s(l)=n}\, j_l,m_l\rangle\,.
\label{eq:SN}
\ee
Its evaluation on identity $s_\Gamma^{(j_l,\iota_n)}(\id)$ plays the role of the vertex amplitude $\cA_{v^*}(j_l,\iota_n)$ of the spinfoam partition function \eqref{eq:amplitudes_boundary} and it describes the 2D boundary quantum geometry of an elementary 3-cell, which can be taken to be a polyhedron with no lose of generality. 

In short, the spinfoam can be viewed as a gluing, under certain gluing conditions, of ``bubbles'' which are homogeneously two-spheres dressed with spin network evaluation. When working on the spin network states, the gluing condition is to identify the shape of the glued boundaries, $\ie$ they have the same number of sides as polygons, and the spins are assigned on the glued links. The spinfoam can thus be written as
\be
Z[\triangle,\psi_{\Gamma=( \partial \triangle)^*_1}^j]= \sum_{\{j_{f^*}\}}\, \prod_{f^*}(-1)^{2j_{f^*}}d_{j_{f^*}}
\,\prod_{e^*}\text{Sign}({\iota_{e^*}})\, \prod_{v^*}  
s_{v^*}^{\{j_l,\iota_n\}}(\id) \,,
\label{eq:partition_cell}
\ee
where $s_{v^*}^{\{j_l,\iota_n\}}(\id)$ is the spin network evaluation on each bubble, and the edge amplitude $\text{Sign}({\iota_{e^*}})$ is a sign depending on the spins of the intertwiner on (the node dual to) the shared face dual to the dual edge $e^*$. The exact value of this sign depends on the choice of basis of the intertwiner.

When the cellular decomposition is specified to be a triangulation, each boundary is made up of a triangle and thus $\Gamma$ is identically three-valent, in which case the intertwiner is one-dimensional and thus uniquely defined under a chosen basis. 
In this case, the spin network evaluation is simply a $\{3nj\}$-symbol.
For a cellular decomposition whose boundary is not a 2-complex but a general 2-cell, $\Gamma$ can be higher-valent. In this case, one can add $s-3$ virtual links to separate an $s$-valent node into three-valent nodes (in a tree way) and obtain a ``fattened node'' in the same spirit as in fig.\ref{fig:hexagon}. One then assigns all admissible spins to each virtual link. For each admissible assignment, say $\{k_{\alpha}\},\alpha=1,\cdots,s-3$, the intertwiner for this fattened node corresponds to one basis $|\cI_a\rangle \langle \cI_a |$ of the intertwiner for the original $s$-valent node. This intertwiner basis combinatorial gives a $\{3nj\}$-symbol associated to the boundary ``fattened graph'' given by the original links plus the added virtual links. Thus we reproduce the result of the Ponzano-Regge amplitude by separating the delta distribution on the $\SU(2)$ group described in the last subsection, $\ie$
\be
Z[\triangle,\psi_{\Gamma}^j]= \sum_{\{j_{f^*}\}}\, \prod_{f^*}(-1)^{2j_{f^*}}d_{j_{f^*}}\,\sum_{\{a\}} \prod_{e^*}(-1)^{\sum_{i=1}^s j_i + 2\sum_{\alpha}^{s-3}k_{\alpha}}\, \prod_{v^*} \{3nj\}_{v^*,a}\,,
\nn\ee
where the power in the edge amplitude is the sum of spins on the links dual to edges surrounding the $s$-gon plus twice of the spin on the links on the virtual links. These spins need to be added twice since each added link is dual to a virtual edge that is on the boundary of two triangles (see black dashed line in fig.\ref{fig:hexagon}).
 
The smallest three-valent graph embedded on a two-sphere is indeed a tetrahedron graph $\Gamma=(\partial T)^*_1$, as illustrated in fig.\ref{fig:spinnetwork}. The spin network state trivially evaluated on a tetrahedron graph gives a $\{6j\}$-symbol \footnotemark{},
\footnotetext{
We have ignored the notation $\iota^n$ for the intertwiners on the left-hand side for simplicity. The intertwiners are implicit in the definition of the $\{6j\}$-symbol.
}
\be
s_{\tet}^{\{j_l\}}(\id)=\Mat{ccc}{j_1 & j_2 & j_3 \\ j_4 & j_5 & j_6}\,,
\label{eq:6j}
\ee
which is exactly the vertex amplitude we obtained through decomposing the delta distribution on $\SU(2)$. 
This is the Ponzano-Regge state-sum for the simplest triangulation of a 3-ball (with no summation at all), describing the boundary quantum geometry of a tetrahedron. In the semi-classical limit, seen from scaling $j$ to $\lambda j $ and taking $\lambda \rightarrow \infty$, the $\{6j\}$-symbol is given by the Hartle-Sorkin action \cite{Hartle:1981cf} in Regge calculus for a tetrahedron in terms of the edge lengths and dihedral angles:
\be
\Mat{ccc}{\lambda j_1 & \lambda j_2 &\lambda j_3 \\ \lambda j_4 & \lambda j_5 & \lambda j_6}\xrightarrow{\lambda\rightarrow \infty} \f{1}{\sqrt{12\pi V}}\cos\left( S_{\text{HS}}(\{\lambda j_l+\f12\}) +\f{\pi}{4} \right)\,,\quad\text{ with }\quad
S_{\text{HS}}(\{\ell_l\})=\sum_{l=1}^6 \ell_l \Theta_l\,,
\label{eq:6j-asymptotic}
\ee
where $\Theta_l$ is the dihedral angle around the edge $e$ (dual to link $l$) and $V$ is the volume of the tetrahedron with edge length $\ell_l=\lambda j_l+\f12\,,l=1,\cdots, 6$ (see also Section \ref{sec:geo_SGF} for a more detailed analysis).
 This asymptotic was postulated by Ponzano and Regge \cite{Ponzano:1968se} and proven in different methods \cite{Roberts:1998zka,Schulten:1975sem,Freidel:2002mj,Barrett:1993db}. The Hartle-Sorkin action \eqref{eq:6j-asymptotic} is the discrete version the GHY boundary term \cite{Hartle:1981cf}:
\be
\int_{\partial \cM} \rd^2 \sqrt{h}K\xrightarrow{\text{discretize}} \sum_{l\in (\partial \triangle)^*_1} \ell_l\Theta_l\,.
\ee   
For a general $\{3nj\}$-symbol, the semi-classical limit also encodes the geometry of the boundary 2-cell it describes \cite{Dowdall:2009eg}.  
The reproduction of the vertex amplitude with the boundary states exposes the fact that the vertex amplitude is a local-holographic amplitude, encoding only the boundary data of the 3-cell it is associated to. 
The gluing process for two adjacent 3-cells can be understood as smearing the data on the shared boundaries. 
In this way, the degrees of freedom on the shared boundaries become gauge through gluing, and the only physical degrees of freedom are on the union of the un-glued boundaries, $\ie$ the cellular decomposition of $\partial \cM$. 
This is exactly why the bulk part of the amplitude is independent of the cellular decomposition and we can refine it up to the quantum limit.

In the Ponzano-Regge state-sum \eqref{eq:spin_form}, the boundary data are the lengths of the one-skeleton encoded in the $\{6j\}$-symbols. Therefore, the path integral constructed as such depends on the boundary metric and thus the length scale. Apparently, the scale invariance of the bulk is merely obtained by the summation of the spin labels in the bulk, $\ie$ the smearing of the length scale. 

In the next section, we will study a scale-invariant path integral even if with boundary configuration. To this end, it is natural to choose a scale-invariant boundary state, then one can define a scale-invariant vertex amplitude as the evaluation of this new quantum state on the boundary of an elementary 3-cell.   
These states should encode the conformal geometry, $\ie$ angles, on the boundary surface. 
Technically, to define such a conformal boundary quantum state means to find an ``alternative" representation of $\SU(2)$ that can be geometrically interpreted as angles. 
Given a different formulation of the partition function from \eqref{eq:spin_form}, the recursion relation of $\{6j\}$-symbols would be replaced by other identities in order to reproduce the topological invariance, which will be the case illustrated in Section \ref{sec:Topo_inv_coherence}.

\section{A new coherent holomorphic state-integral}
\label{sec:SGF}

The ``alternative'' representation we are going to apply is the spinor representation of $\SU(2)$ \cite{Livine:2011gp}, developed from the $\bU(N)$ formulation of LQG \cite{Dupuis:2010iq,Freidel:2010tt,Borja:2010rc}. 
The quantum states under this representation are called the {\it coherent spin network states} (see below) \cite{Bonzom:2011nv}. 
After a concise review of its general construction, we will specialize in the coherent spin network state of a tetrahedron graph, whose evaluation on identity after changing the weight gives the SGF \cite{Schwinger:1965an,Bargmann:1962zz}. 
We will construct a new Ponzano-Regge model with spinor variables, where the SGF serves as the vertex amplitude. Similar to the original one, it can also be seen as built with local amplitudes associated to the elementary bubbles.

\subsection{Generating function for spin networks}
\label{sec:spinor}

Let us introduce the $\SU(2)$ spin coherent state (or the $\SU(2)$ coherent state for short) {\it à la} Perelomov, denoted as $|j,z\rangle$ with a fixed spin $j\in\N/2$ and a spinor $|z\rangle:=\mat{c}{z^0\\z^1}\in \bC^2$. It is a superposition state of a pair of harmonic oscillators $|n^0,n^1\rangle_{\HO}$, identified with a magnetic number basis $|j,m\rangle\in \cV^j$ by the relation $j=\f12 (n^0 + n^1), \,m=\f12 (n^0 - n^1)$, and reads
\be
|j,z\rangle:=\frac{(z^A a^{A\dagger})^{2j}}{\sqrt{(2j)!}}|0\rangle \equiv
\sum_{m=-j}^{j}\sqrt{\frac{(2j)!}{(j+m)!(j-m)!}}(z^0)^{j+m}(z^1)^{j-m}|j,m\rangle\,,
\label{eq:jz_basis}
\ee
where $a^{A\dagger}$ is the creation operator acting on the oscillator $n^A$, $A=0,1$. This is the state that admits the generalized minimal uncertainty given by the dispersion of the $\SU(2)$ Casimir \cite{Livine:2007vk}. 
The norm is easily computed $\langle j,z|j,z\rangle = \langle z|z\rangle^{2j}$.
For a fixed spin $j$, $|j,z\rangle$ serves as an alternative orthonormal basis spanning the representation space $\cV^j$. We refer to \cite{Livine:2011gp} for more details. 

It is easy to check that the creation and annihilation operators act on the $\SU(2)$ coherent state as
\be
a^A|j,z\rangle =\sqrt{2j}\, z^A \,|j-\f12, z\rangle \,,\quad
a^{A\dagger} |j,z\rangle =\frac{1}{\sqrt{2j+1}}\frac{\partial}{\partial z^A}\,|j+\f12,z\rangle \,,
\label{eq:a_on_jz}
\ee
thus $a^A$ decreases the spin by $1/2$ as well as multiplying the state by $z^A$, while $a^{A\dagger}$ increases the spin by $1/2$ and derives the state by $z^A$. 

We also introduce a dual $\SU(2)$ coherent state $[j,z|\equiv \langle j, \varsigma z|$ in terms of a dual spinor $[z|=\langle \varsigma z|:= \mat{cc}{-z^1 ,& z^0}$, which is also holomorphic $z^A$, living in the dual representation space $\cV^{j*}$ (see \eqref{eq:dual_spinor}) \footnotemark{}. $a^A$ acts on $[j,z|$ as the creation operator while $a^{A\dagger}$ acts as the annihilation operator,
\be
[j,z|\,a^A 
=-\epsilon^{AB}\frac{1}{\sqrt{2j+1}}\frac{\partial}{\partial z^B}\, [j+\f12,z|\,,\quad
[j,z|\,a^{A\dagger}
=-\epsilon^{AB}\sqrt{2j}\,z^B \,[j-\f12,z|\,.
\label{eq:a_on_jz_dual}
\ee
\footnotetext{
We remind the readers that $[j,z|$ is called the dual spinor due to the fact that it is orthogonal to the regular spinor $|j, z\rangle$, $\ie$ $[j,z|j,z\rangle=0$, but not because it lives in the dual representation $\cV^{j*}$. $|j,z]=|j,\varsigma z\rangle$ is the dual spinor living in $\cV^j$ and anti-holomorphic on $z^A$.
}
Spinors have been used to construct the spinorial phase space of loop gravity, which is equivalent to the holonomy-flux phase space \cite{Livine:2011gp}. A brief summary is given in Appendix \ref{app:spinorial}. 

Consider an oriented graph $\Gamma$ with $|L|$ oriented links $l$'s, $|N|$ nodes $n$'s and $|P|$ plaquettes $p$'s. 
We dress each link $l$ with a spin $j_l$, and associate a spinor $|z_l\rangle$ to the source $s(l)$ of $l$ and another spinor $|\zt_l\ra$ to the target $t(l)$. 
For an $N$-valent node $n$, we can construct an intertwiner living in the tensor space $\left(\bigotimes_{l|s(l)=n}\cV^{j_l}\right)\otimes\left( \bigotimes_{l|t(l)=n}\cV^{j_l^*}\right)$ by $\SU(2)$-group averaging the tensor product of the $\SU(2)$ holomorphic coherent states, $\ie$
\be
\iota^n_{(j_l)}(z_l)\,|0\rangle:=\int_{\SU(2)}\rd g_n \left( \bigotimes_{l|t(l)=n}[j_l,\zt_l|\,g_n^{-1}\right) \otimes \left(\bigotimes_{l|s(l)=n}g_n\,|j_l,z_l\rangle\right)\,.
\label{eq:LS_intertwiner}
\ee
This is called the {\it LS coherent intertwiner}, first introduced in \cite{Livine:2007vk} (see also \cite{Freidel:2010tt}) and used to define the EPRL-FK spinfoam models \cite{Engle:2007wy,Livine:2007ya}. It is also closely related to the $\bU(N)$ coherent states which are by definition covariant under the $\bU(N)$ action \cite{Freidel:2010tt}. 

\eqref{eq:intertwiner_jm} and \eqref{eq:LS_intertwiner} are simply projections of a general $\SU(2)$ intertwiner on different bases, the former on the magnetic number basis while the latter on the coherent state basis.
Equipped with the intertwiners \eqref{eq:LS_intertwiner}, one can define a spin network state evaluated on $\SU(2)$ group elements $\{g_l\}$ as a holomorphic function of the spinors:
\be
s_\Gamma^{\{j_l,z_l,\zt_l\}}(g_l)= \int_{\SU(2)^{|N|}}\prod_{n}\rd h_n\prod_{l} \,
[ j_l,\zt_l |\, h_{t(l)}^{-1}\, g_l \,h_{s(l)}\,| j_l,z_l\rangle\,.
\label{eq:SN_spinor}
\ee

It is also possible to eliminate the spins and define an intertwiner associated to a node $n$ with only spinor labels, which can be viewed as a generating function of the LS coherent intertwiners \eqref{eq:LS_intertwiner} with a chosen series of weight \cite{Bonzom:2012bn}. 
The simplest weight is $\frac{1}{\prod_{l\in n} \sqrt{(2j_l)!}}$, which defines the {\it coherent intertwiner} \footnotemark{},
\be
\iota^n(z_l) = \sum_{\{j_l\}} \frac{1}{\prod_{l\in n} \sqrt{(2j_l)!}} \iota^n_{(j_l)}(z_l)\,.
\label{eq:coherent_intertwiner}
\ee
It is indeed an $\SU(2)$ invariant state in $\bigoplus_{\{j_l\}}\left(\bigotimes_{l|s(l)=n}\cV^{j_l}\right)\otimes\left( \bigotimes_{l|t(l)=n}\cV^{j_l^*}\right)$.
We associate a coherent intertwiner to each node and glue them along links associated with $\SU(2)$ holonomies. The gluing is performed in the standard way, by taking different irreducible representations orthogonal. The result defines the coherent spin network state \cite{Bonzom:2012bn}
\be
s_\Gamma^{\cohe}(g_l)=
\sum_{\{j_l\}} \int_{\SU(2)^{|N|}}\prod_{n}\rd h_n\prod_{l} \,
\f{1}{(2j_l)!}[ j_l,\zt_l |\, h_{t(l)}^{-1}\, g_l \,h_{s(l)}\,| j_l,z_l\rangle 
=\int_{\SU(2)^{|N|}}\prod_{n}\rd h_n\prod_{l} \, e^{[\zt_l| h_{t(l)}^{-1}\, g_l \,h_{s(l)}|z_l\rangle}\,.
\label{eq:coherent_SN}
\ee
\footnotetext{
The term ``coherent intertwiner'' was used to denote the LS coherent intertwiner \eqref{eq:LS_intertwiner} for short in some literature. We remind the readers that these two terms have distinct definitions in this paper, following \cite{Bonzom:2012bn}.
}

A slightly different choice of weight, which is what we will focus on in this paper, is $\frac{(J_n+1)!}{\prod_{l\in n} \sqrt{(2j_l)!}}$ \footnotemark{}, 
\footnotetext{
Other choices of weight lead to different generating functions, which would be useful for different interests. See $\eg$ \cite{Bonzom:2012bn} for a discussion and the application of other generating functions with alternative weights.
}
where $J_n\equiv \sum_{l\in n}j_l$ is the sum of spins on the links incident to $n$. It defines the {\it scaleless intertwiner} (as we will see in the next section that it encodes the scale-invariant geometry of a 3-cell boundary):
\be
\iota_n^{\sl}(z_l) = \sum_{\{j_l\}} \frac{(J_n+1)!}{\prod_{l\in n} \sqrt{(2j_l)!}} \iota^n_{(j_l)}(z_l)\,.
\label{eq:scaleless_intertwiner}
\ee

Again, we associate a scaleless intertwiner to each node and glue them in the same way as the coherent intertwiners. We define it as the {\it scaleless spin network state} :  
\be
s^{\sl}_{\Gamma}(g_l)= \sum_{\{j_l\}}\prod_n \frac{(J_n+1)!}{\prod_{l\in n}(2j_l)!} \int_{\SU(2)^{|N|}}\prod_{n}\rd h_n\prod_{l} \,
[ j_l,\zt_l |\, h_{t(l)}^{-1}\, g_l \,h_{s(l)}\,| j_l,z_l\rangle\,.
\label{eq:SpinorN}
\ee
It can be viewed as a generating function of the spin network state \eqref{eq:SN}. 
The use of spinors shifts the view of building blocks of quantum geometries from the links to the nodes, which is also the spirit behind the construction of the $\bU(N)$ coherent states \cite{Freidel:2010tt,Borja:2010rc,Dupuis:2010iq}. 

Above we have defined the coherent intertwiners, coherent spin network states, scaleless intertwiners and the scaleless spin network states for a general graph. The goal is to use these notions to define a new vertex amplitude in terms of spinors in the Ponzano-Regge model. To this end, we will work on a three-valent graph in the next subsection. More specifically, we will study the tetrahedron graph as shown in fig.\ref{fig:spinnetwork} and study the evaluation of the scaleless spin network state on this simple graph.

\subsection{The holomorphic $\{12z^{\times 2}\}$ symbol and 2D Ising on a tetrahedron}

In this subsection and the next, we will fix the cellular decomposition of $\cM$  to be a triangulation unless specified and intensively work on the tetrahedron graph $\Gamma=(\partial T)^*_1$ that is a 2D dual graph of the boundary two-skeleton of a tetrahedron $T$. The simplicity it brings helps to quantify the scaleless intertwiners \eqref{eq:scaleless_intertwiner} and the scaleless spin network functions \eqref{eq:SpinorN}. 
On the other hand, it turns out that the scaleless spin network state for a tetrahedron graph, when evaluated on the identity, possesses a closed form known as the SGF.

For a three-valent graph, 
the intertwiner \eqref{eq:LS_intertwiner} for each node is one-dimensional, thus it must be proportional to the $\{3j\}$-symbol. The exact relation is well-known \cite{Varshalovich:1988qu}, with the proportionality coefficient given by a holomorphic polynomial of degree $J_n$. Consider three outgoing links $(l_1,l_2,l_3)$ meeting at the node $n$, we denote the total spin as $J_{123}=j_1+j_2+j_3$. Then one gets
\be
\iota^{n}_{j_1j_2j_3}(z_1,z_2,z_3)=P_{j_1j_2j_3}(z_1,z_2,z_3)\, \iota^{n}_{j_1j_2j_3}\,,
\label{eq:intertwiner_123}
\ee
with
\be
P_{j_1j_2j_3}(z_1,z_2,z_3) 
 = \frac{\triangle(j_1j_2j_3)}{(J_{123}+1)!}
\left(\prod_{l=1}^3 \sqrt{(2j_l)!}\right)
[z_1|z_2\rangle^{J_{123}-2j_3}\,[z_2|z_3\rangle^{J_{123}-2j_1}\,[z_3|z_1\rangle^{J_{123}-2j_2}\,,
\label{eq:P_jz}
\ee
where $\triangle(j_1j_2j_3)$ is the quantum triangle coefficient defined as
\be
\triangle(j_1j_2j_3)=\sqrt{\frac{(J_{123}+1)!}{(j_1+j_2-j_3)!(j_1+j_3-j_2)!(j_2+j_3-j_1)!}}\,.
\ee
The scaleless intertwiner \eqref{eq:scaleless_intertwiner} for the node with incident links $(l_1,l_2,l_3)$ outgoing reads
\be
\iota_n^{\sl} (z_1,z_2,z_3) = \sum_{j_1,j_2,j_3} \frac{(J_{123}+1)!}{\prod_{l=1}^3 \sqrt{(2j_l)!}}
\iota^n_{j_1j_2j_3}(z_1,z_2,z_3)\,,
\ee
which is an invariant vector on $\otimes_{l=1}^3(\oplus_{j_l}\cV^{j_l} )$, and is also a generating function of the $\{3j\}$-symbol. 

The proportionality coefficient \eqref{eq:P_jz} remains unchanged when some links are incoming (except that, in this case, we denote the spinor on an incoming link with a tilde).
For instance, for a node with links $e_1$, $e_2$ incoming and link $e_3$ outgoing, the relation reads
\be
\iota^n_{j_1^*j_2^*j_3}(\zt_1,\zt_2,z_3)=P_{j_1j_2j_3}(\zt_1,\zt_2,z_3)\, \iota^n_{j_1^*j_2^*j_3}\,.
\ee

The scaleless spin network state for the tetrahedron graph naturally follows except that there is a sign ambiguity as the graph is odd-valent. 
The reason is that the sign of the intertwiner \eqref{eq:intertwiner_123} would be changed under exchanging any pair of spinors in the argument, thus it is necessary to fix the ordering of links incident to a node in order to specify (the sign of) the definition of the scaleless spin network function. 
In practice, it is enough to fix a cyclic order $\prec$ for each node. To do this, we first embed the graph $\Gamma$ on a 2D oriented surface $\Sigma$. 
Consider three links $(l,l',l'')$ incident to one node $n$, we call $l'$ is of higher order than $l$, denoted as $l\prec l'$(or $l'\succ l$), when the sweeping from $l$ to $l'$ (in the direction without touching $l''$) induces a surface with the same orientation as $\Sigma$. Otherwise, $l'$ is of lower order than $l$, denoted as $l \succ l'$ (or $l' \prec l$). For an ordering $(l \prec l', l' \prec, l'', l'' \prec l )$ as in fig.\ref{fig:order_node}, we fix an ordered holomorphic polynomial $P^{\prec}_{j_lj_{l'}j_{l''}}(z_l,z_{l'},z_{l''}) $ to be
\be
P^{\prec}_{j_ej_{l'}j_{l''}}(z_l,z_{l'},z_{l''}) 
 = \frac{\triangle(j_{l}j_{l'}j_{l''})}{(J_n+1)!}
\left(\prod_{l\in n} \sqrt{(2j_l)!}\right)\,
[z_l|z_{l'}\rangle^{J_n-2j_{l''}}\,[z_{l'}|z_{l''}\rangle^{J_n-2j_{l}}\,[z_{l''}|z_l\rangle^{J_n-2j_{l'}}\,.
\ee
We have ignored the tilde of spinors for incoming links for simplicity and unification. 
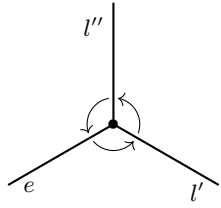
\begin{figure}[h!]
	\centering
\begin{tikzpicture}[scale=0.7]
\coordinate (A) at (0,0);
\coordinate (B) at (4*4/4,0);
\coordinate (C) at (2*4/4,3.464*4/4);

\coordinate (c) at ($ (A)!.5!(B) $);
\coordinate (b) at ($ (A)!.5!(C) $);
\coordinate (a) at ($ (B)!.5!(C) $);

\path[name path = Aa] (A) -- (a);
\path[name path = Bb] (B) -- (b);
\path[name path = Cc] (C) -- (c);
\path [name intersections = {of = Aa and Bb,by=O}];

\draw[thick] (A) -- node[pos=0.2,below]{$e$}(O);
\draw[thick] (B) -- node[pos=0.2,below]{$l'$}(O);
\draw[thick] (C) -- node[pos=0.2,left]{$l''$}(O);

\draw[->] ([shift=(-20:0.5cm)]O) arc (-20:80:0.5cm);
\draw[->] ([shift=(100:0.5cm)]O) arc (100:200:0.5cm);
\draw[->] ([shift=(220:0.5cm)]O) arc (220:320:0.5cm);

\draw (O) node{$\bullet$};

\end{tikzpicture}
\caption{The cyclic order $(l \prec l', l' \prec l'', l'' \prec l )$ of three links $l,l',l''$ incident to the same node.}	
\label{fig:order_node}
\end{figure}

Consider again a tetrahedron graph as in fig.\ref{fig:spinnetwork}, it can be naturally embedded in a 2-sphere, which generates the cyclic order for all the nodes at once. The scaleless spin network function is then uniquely defined as
\be
\cS_{\tet}^{\{z_l,\zt_l\}}(g_l)=\sum_{j_1,\cdots,j_6}\frac{\prod_{n=1}^4 (J_n+1)!}{\prod_{l=1}^6 (2j_l)!}s_{\tet}^{\{j_l,z_l,\zt_l\}}(g_l)\,,
\label{eq:spinor_from_spin_network}
\ee
with
\be\begin{split}
s_{\tet}^{\{j_l,z_l,\zt_l\}}(g_l)
&=\int_{\SU(2)^{4}}\prod_{n=1}^4 \rd h_n \prod_{l=1}^6 
 [j_l,\zt_l|\,h_{t(l)}^{-1}\, g_l \, h_{s(l)}\, |j_l,z_l\rangle\\
&=P^{\prec}_{j_1j_2j_3}(z_1,z_2,z_3)P^{\prec}_{j_1j_5j_6}(\zt_1,\zt_5,z_6)P^{\prec}_{j_3j_4j_5}(\zt_3,\zt_4,z_5)P^{\prec}_{j_2j_4j_6}(\zt_2,\zt_6,z_4)\,
s_{\tet}^{\{j_l,\iota_n\}}(g_l)
\end{split}
\label{eq:spinor_network_j}
\ee
being the special case of \eqref{eq:SN_spinor} for a tetrahedron graph, which can be factorized, as shown in the second line of \eqref{eq:spinor_network_j}, into the standard spin network function independent of the spinors, and a holomorphic polynomial independent of the arguments $\{g_l\}$. 
In particular, its evaluation on identity gives a holomorphic ``$\{12z^{\times 2}\}$ symbol'', known as the Schwinger's generating function (SGF) of the $\{6j\}$-symbols, which is a function of 12 spinors thus 24 complex variables,
\be
\cS^{\sl}_{\tet}(\{z_l,\zt_l\})
=\cS_{\tet}^{\{z_l,\zt_l\}}(\id)=
\sum_{j_1\cdots j_6} \left[ \prod_{n=1}^4\sqrt{\frac{(J_n+1)!}{\prod_{l\in n}(J_n-2j_l)!}} \right]
\left\{ 
\ba{ccc}
j_1 & j_2 & j_3 \\
 j_4 & j_5 & j_6
\ea\right\}
\prod_{n=1}^4\prod_{\ba{c}l,l',l''\in n, \\ l\prec l' \ea} \left[z_{l}|z_{l'}\right>^{J_n-2j_{l''}}\,.
\label{eq:SGF_1}
\ee
It was first found to be of the closed form by Schwinger \cite{Schwinger:1965an,Bargmann:1962zz,Bonzom:2011nv}. (See also \cite{Bonzom:2015ova} for the deduction from the duality between the 2D Ising model and the Ponzano-Regge model). It is in a form of a scaleless function: 
\be\begin{split}
\cS(\{z_l,\zt_l\})\equiv\cS^{\sl}_{\tet}(\{z_l,\zt_l\})=G(\{z_l,\zt_l\})^{-2}\,,\quad
G(\{z_l,\zt_l\})
&=1+\sum_{\cL}\prod_{v\subset \cL \,,\, l \prec l' } [z_l|z_{l'}\rangle\\
&=1+\sum_{(3c)}^4
\prod_{\ba{c}v\subset (3c) \\ l \prec l' \ea}\left[ z_{l}|z_{l'} \right>
+\sum_{(4c)}^3
\prod_{\ba{c}v\subset (4c) \\ l \prec l' \ea}\left[ z_{l}|z_{l'} \right>\,,
\label{eq:evaluate_SGF}
\end{split}\ee
where $\cL$'s denote the loops in the tetrahedron graph, including three-cycles $(3c)$'s and four-cycles $(4c)$'s. 
Use the notation in fig.\ref{fig:spinnetwork_b}, 
the cycle sums are explicitly
\begin{align}
\sum_{(3c)}^4\prod_{\ba{c} v\subset (3c) \\ l \prec l' \ea}\left[ z_{l}|z_{l'} \right>
&=\left[ z_1|z_2 \right>\! \left[\zt_2|\zt_6 \right> \!\left[z_6|\zt_1 \right> 
+\left[ \zt_3|z_1 \right>\! \left[ \zt_1|z_5 \right> \!\left[ \zt_5|z_3 \right> 
+\left[ z_2|\zt_3 \right>\! \left[ z_3|z_4 \right>\! \left[ \zt_4|\zt_2 \right> 
+\left[ \zt_6|\zt_4 \right>\! \left[ z_4|\zt_5 \right>\! \left[ z_5|z_6 \right> \,,\\
\sum_{(4c)}^3
\prod_{\ba{c} v\subset (4c) \\ l \prec l' \ea}\left[ z_{l}|z_{l'} \right>
&=\left[ \zt_3|z_1 \right>\!\left[ \zt_1|z_6 \right>\!\left[ \zt_6|\zt_4 \right>\!\left[ z_4|z_3 \right>
+\left[ \zt_5|z_3 \right>\!\left[ \zt_3|z_2 \right>\!\left[ \zt_2|\zt_6 \right>\!\left[ z_6|z_5 \right>
+\left[ \zt_1|z_5 \right>\!\left[ \zt_5|z_4 \right>\!\left[ \zt_4|\zt_2 \right>\!\left[ z_2|z_1 \right>\,.
\end{align}
Examples of the 3-cycles and 4-cycles for the tetrahedron graph are given in fig.\ref{fig:3_4_cycles}. 
Note that the SGF \eqref{eq:evaluate_SGF} does not depend on the orientation of edges but only the order of the half edges incident to each vertex, which is already fixed to be consistent with the orientation of the surface they are embedded on. 
Each inner product of a pair of spinors $[z_l|z_{l'}\rangle$ contained in a cycle is associated to an angle, thus we can understand the cycles as the result of gluing those angles in cyclic order. 
\begin{figure}[h!]
\centering
\begin{tikzpicture}[scale=1.7]
\coordinate (O1) at (0,0);

\coordinate (A1) at (0,0.8);
\coordinate (B1) at (0,-0.5);
\coordinate (C1) at (1.3,0.8);
\coordinate (D1) at (1.3,-0.5);
\coordinate (aa) at (0.65,-1);
\coordinate (bb) at (4.15,-1);
\draw (aa) node{3-cycle $(3c)$ $(126)$} ;
\draw (bb) node{4-cycle $(4c)$ $(1643)$} ;

\draw[red] (A1) node{$\bullet$};
\draw[red] (B1) node{$\bullet$};
\draw[blue] (C1) node{$\bullet$};
\draw[red] (D1) node{$\bullet$};

\draw[thick, red,postaction={decorate},decoration={markings,mark={at position 0.55 with {\arrow[scale=1.5,>=stealth]{>}}}}] (B1) -- (A1) node[midway,left]{$l_{1}$} ;
\draw[blue,postaction={decorate},decoration={markings,mark={at position 0.55 with {\arrow[scale=1.5,>=stealth]{>}}}}] (A1) -- (C1) node[midway,above]{$l_{5}$};
\draw[dashed,thick, red,postaction={decorate},decoration={markings,mark={at position 0.65 with {\arrow[scale=1.5,>=stealth]{>}}}}] (A1) -- (D1) node[pos=0.65,right]{$l_{6}$};
\draw[blue,postaction={decorate},decoration={markings,mark={at position 0.65 with {\arrow[scale=1.5,>=stealth]{>}}}}] (C1) -- (B1)  node[pos=0.65,left]{$l_{3}$};
\draw[blue,postaction={decorate},decoration={markings,mark={at position 0.55 with {\arrow[scale=1.5,>=stealth]{>}}}}] (C1) -- (D1)  node[midway,right]{$l_{4}$};
\draw[thick, red,postaction={decorate},decoration={markings,mark={at position 0.55 with {\arrow[scale=1.5,>=stealth]{>}}}}] (B1) -- (D1) node[midway,below]{$l_{2}$};

\coordinate (A2) at (3.5,0.8);
\coordinate (B2) at (3.5,-0.5);
\coordinate (C2) at (4.8,0.8);
\coordinate (D2) at (4.8,-0.5);
\draw[red] (A2) node{$\bullet$};
\draw[red] (B2) node{$\bullet$};
\draw[red] (C2) node{$\bullet$};
\draw[red] (D2) node{$\bullet$};
\draw[thick, red,postaction={decorate},decoration={markings,mark={at position 0.55 with {\arrow[scale=1.5,>=stealth]{>}}}}] (B2) -- (A2)  node[midway,left]{$l_{1}$} ;
\draw[blue,postaction={decorate},decoration={markings,mark={at position 0.55 with {\arrow[scale=1.5,>=stealth]{>}}}}] (A2) -- (C2) node[midway,above]{$l_{5}$};
\draw[dashed,thick, red,postaction={decorate},decoration={markings,mark={at position 0.65 with {\arrow[scale=1.5,>=stealth]{>}}}}] (A2) -- (D2) node[pos=0.65,right]{$l_{6}$};
\draw[thick,red,postaction={decorate},decoration={markings,mark={at position 0.65 with {\arrow[scale=1.5,>=stealth]{>}}}}] (C2) -- (B2)  node[pos=0.65,left]{$l_{3}$};
\draw[thick, red,postaction={decorate},decoration={markings,mark={at position 0.55 with {\arrow[scale=1.5,>=stealth]{>}}}}] (C2) -- (D2)  node[midway,right]{$l_{4}$};
\draw[blue,postaction={decorate},decoration={markings,mark={at position 0.55 with {\arrow[scale=1.5,>=stealth]{>}}}}] (B2) -- (D2) node[midway,below]{$l_{2}$};

\end{tikzpicture}
\caption{({\it in red}) Cycle $(126)$ as an example of 3-cycles and cycle $(1643)$ as an example of 4-cycles. The numbers in the bracket denote the links on the cycle. }
\label{fig:3_4_cycles}
\end{figure}
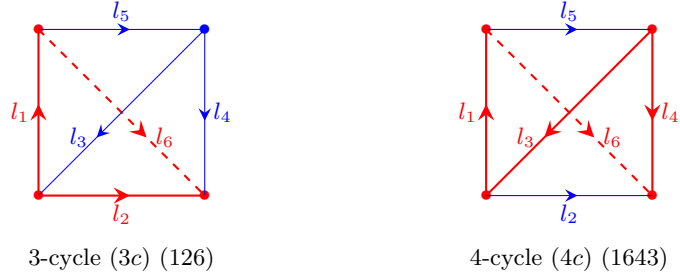

\medskip
\noindent {\bf Duality between SGF and the 2D Ising on a tetrahedron. } 
\medskip

It was shown in \cite{Bonzom:2015ova} that the SGF is related to the high-temperature expansion of the Ising model defined on the same tetrahedron graph. In fact, such a relation can be 
generalized to an arbitrary three-valent planar graph (with the Kasteleyn orientation).  
Explicitly, the scaleless spin network evaluation, as a function of the spinors, is inversely proportional to the square of the amplitude $Z^{\text{Ising}}_{\Gamma}$, as a function of the edge couplings $\{Y_l\}$ in terms of the spinors $\{z_l,\tz_l\}$, of the Ising model, which is a polynomial of the couplings. Precisely,
\be
\cS^{\sl}_\Gamma(\{z_l,\tz_l\}) =\left( 2^{2|N|} \prod_{e\in\Gamma} \cosh^2 y_l \right) \f{1}{\left[Z^{\text{Ising}}_\Gamma(\{y_l\})\right]^2}\,, \quad \text{ with }\quad
Z^{\text{Ising}}_\Gamma(\{y_l\}) = \sum_{\cL \in \Gamma} \prod_{l\in \cL} \tanh y_l\,,
\ee
where the sum is over all the disjoint loops $\cL$'s in $\Gamma$. 
The pole of the SGF thus corresponds to the zeros of the 2D Ising amplitude for the same graph. 
The dependence of the scaleless spin network states on the spinors can also be rewritten as the dependence on the link couplings (defined below in \ref{sec:geo_SGF}) as one-to-one correspondence to the edge couplings in the Ising mode. 
Thus this duality directly gives an exact and closed expression of the scaleless spin network for a general three-valent graph $\Gamma$ (see \cite{Bonzom:2015ova}). The exact expression is also proven directly in \cite{Westbury:1998ge,Freidel:2012ji}. 
A similar result was also found for a square 2D lattice \cite{Dittrich:2013jxa}

\medskip

Let us comment that the SGF can be understood in two ways, leading to the application of different interests:
\begin{enumerate}
	\item As indicated by the name, the SGF is a generating function of the $\{6j\}$-symbols, thus it contains all the information of the spin network evaluation of a tetrahedron graph in its expansion. 
That the spins are restricted to be half-integers and admissible in the $\{6j\}$-symbols brings difficulties in numerical analysis, which was often circumvented by looking at the large $j$ limit of the $\{6j\}$-symbol. 
The computation issues are translated to complex analysis when we consider the SGF, which is expected to be easier to deal with as it is a continuous function of 12 independent spinors. The large $j$ limit of the $\{6j\}$-symbols then contributes dominantly at the poles of the SGF \cite{Bonzom:2019dpg}. 

Other than the possible computational benefits, the SGF is attractive because of its remarkable closed form. A similar closed form for the generating function of the $9j$-symbols was also constructed in the original work by Schwinger \cite{Schwinger:1965an}. Later on, Bargmann \cite{Bargmann:1962zz} reproduced Schwinger's results and his method was used to construct the generating function of the $12j$- and $15j$-symbols \cite{Huang:1974gs} and more generally the $\{3nj\}$-symbols \cite{Labarthe:1975yf}, which can also be expressed as loops. 
	\item The SGF is also understood as (the evaluation of) a quantum state - the scaleless spin network state -
that describes the physical state of a tetrahedron graph in the spinor representation. 
The physical state of a graph is most simply represented by the flatness projection written in terms of the holonomies, which is the product of the Dirac deltas on holonomies corresponding to independent cycles in the graph. For a graph embedded on the manifold of trivial topology, the physical state is unique, which is the case for the tetrahedron graph embedded in the 2-sphere.
Using the notation in fig.\ref{fig:spinnetwork}, the physical state for a tetrahedron graph can be written as
\be
\psi_{\phys}(g_l)=\delta(g_6g_1g_2^{-1})\delta(g_1^{-1}g_5^{-1}g_3^{-1})\delta(g_3g_4^{-1}g_2)\,.
\label{eq:physical_state}
\ee
When projected onto the scaleless spin network basis, $\psi_{\phys}(g_l)$ gives exactly the SGF,
\be
\psi_{\phys}(\{z_l,\zt_l\})= \int_{\SU(2)^6}\prod_{l=1}^6 \rd g_l \cS_\tet^{\{z_l,\zt_l\}}(g_l) \psi_\phys(g_l) = \cS(\{z_l,\zt_l\})\,.
\ee
This can be deduced from \eqref{eq:spinor_from_spin_network} and the fact that the projection of the physical state onto the spin network basis \eqref{eq:SN} is the $\{6j\}$-symbol \cite{Bonzom:2011hm}:
\be
\psi_{\phys}(j_l)=\int_{\SU(2)^6}\prod_{l=1}^6 \rd g_l \, s_\tet^{\{j_l,\iota_n \}}(g_l) \psi_\phys (g_l) = \Mat{ccc}{j_1 & j_2 & j_3 \\ j_4 & j_5 & j_6}\,.
\ee
The scaleless spin network states can be used to describe the kinematical information in the Hilbert space of quantum geometry. By the continuous nature of the spinor arguments in a scaleless spin network state, the dynamics given by the Wheeler de-Witt equation would be translated into a differential equation of the SGF \cite{Bonzom:2011nv}, rather than a recursion relation as for the $\{6j\}$-symbols when we consider the spin network states \cite{Bonzom:2011hm}.
The interested reader will find a short review with details of this approach in Appendix \ref{app:WdW_SGF}.

\end{enumerate}

\medskip
\noindent {\bf Symmetries of the SGF. } 
\medskip

The loop structure \eqref{eq:evaluate_SGF} of the SGF brings a large number of degrees of symmetry.
Compared to a $\{6j\}$-symbol which has only 6 real variables, the SGF has 12 independent spinors, thus 24 complex or 48 real variables, in its argument. 
The symmetries would allow us to work on a smaller set of variables. 
Firstly, the SGF \eqref{eq:SGF_1} is explicitly written only in terms of the inner product $[z_l|z_{l'}\rangle$ of spinors, which is by definition invariant under the $\SL(2,\bC)$ action that acts covariantly on the spinors
\be
g_n \triangleright |z_l'\rangle = g_n|z_l'\rangle\,,\quad
g_n \triangleright [z_l| =  [z_l|g_n^{-1}\,,\quad
l, l' \in n\,,\quad
g_n \in \SL(2,\bC)\,.
\ee
Thus the spinors form a set with 12 complex or 24 real variables. 
Secondly, notice that the exact evaluation \eqref{eq:evaluate_SGF} takes the form as cycles, a (complex) rescaling of the spinor $z_l \rightarrow \alpha z_l$ $(\alpha\in \bC\backslash \{0\})$ on one end of each link and an ``anti-rescaling'' of the spinor $\tz_l \rightarrow  \tz_l/\alpha$ on the other end leaves the SGF unchanged.
Therefore, the symmetries of the SGF can be expressed as 
\be
\cS(\{z^n_l,\tz^{n'}_l\}) = \cS(\{g_n\, z^n_l, g_{n'}\, \tz^{n'}_l\}) = \cS(\{\alpha_l\, z^n_l, \alpha_l^{-1} \tz^{n'}_l\})\,,\quad
\forall\, g_n \in \SL(2,\bC)\,,
\forall\, \alpha_l \in \bC\backslash \{0\}\,.
\label{eq:symm_SGF}
\ee
In words, the SGF is invariant under the $\SL(2,\bC)$ gauge transformation, one for each node hence a total of 12 complex degrees of freedom, and anti-scale transformation, one for each link hence a total of 6 complex degrees of freedom. 
These symmetries form the full redundant degrees of freedom in the 12 spinors, leaving 6 complex degrees of freedom in the SGF. 
The same redundancy appears when we use the spinor representation to describe the scaleless spin network states for a general graph. 
It is thus possible to change the argument variables of the SGF to a smaller set.  We will see in Subsection \ref{sec:geo_SGF} that a change of variables will make it clear to see the geometrical information given by this scaleless spin network state.

\subsection{A new Ponzano-Regge state-integral formula}
We now turn our attention back to the Ponzano-Regge model. We aim to decompose the discrete path integral 
\be
Z[{\bf T},\partial{\bf T}]= C[\partial {\bf T} ] \int_{\SU(2)} \prod_{e^*\notin \partial{\bf T}}\rd g_{e^*} \,\prod_{f^*}\delta(\overrightarrow{\prod}_{e^*\in \partial f^*}g_{e^*})\nn
\ee 
into the product of vertex, edge and face amplitudes encoding local quantum geometries through the spinor representation labels, so that the global quantum geometry can be understood as gluing elementary blocks with local geometry information stored in these spinor variables. 
The total amplitude should be written alternatively in the form as
\be
\cA[{\bf T},\psi_\Gamma^{\sl} ]=\int \left[\rd \mu(z)\right] \prod_{f^*} \cA_{f^*}[z_{f^*}] \prod_{e^*}\cA_{e^*}[z_{e^*}]\prod_{v^*}\cA_{v^*}[z_{v^*}]\,,
\label{eq:state_integral}
\ee
where $\psi^{\sl}_\Gamma$ is the boundary scaleless spin network state and $\rd\mu(z):=\frac{1}{\pi^2}e^{-\left<z|z\right>}\rd z^0 \rd z^1$ is the Haar measure of spinors (see also Appendix \ref{app:spinorial}). With the use of spinors, therefore, the amplitude would be written as a ``state-integral'' instead of a state-sum.

Let us first introduce the notations we will use in the state-integral formulas both in Proposition \ref{prop:state_integral_2} (also in \eqref{eq:PR_cohe}, which is proved in Appendix \ref{app:coherent_PR}). For the vertex amplitude, we denote the spinors (or dual spinors) on the source $s(l)$ and target $t(l)$ of a link $l$ as $z_l$ and $\tz_l$ respectively. For the edge amplitude of a dual edge $e^*$, we consider the node $n$ in a triangle shared by two adjacent tetrahedra. Therefore, spinors on $v$ have two independent copies, one from the tetrahedron dual to $s(e^*)$ and the other from the tetrahedron dual to $t(e^*)$. For a link $l\in n$, we denote the two spinors from the two tetrahedra respectively as $z_l^{s(e^*)}$ and $z_l^{t(e^*)}$. For the face amplitude, we consider a dual face $f^*$ whose boundary loop connects $M(\geq 3)$ tetrahedra. One needs to choose randomly a node on the boundary of one of these tetrahedra, say $T_1$. We denote the spinor on this node that will contribute to this face amplitude as $z^{f^*,T_1}$.
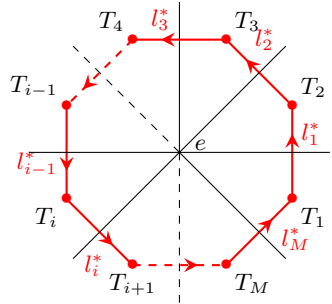
\begin{figure}[h!]
	\centering
\begin{tikzpicture}[one end extended/.style={shorten >=-#1},one end extended/.default=0.5cm]
	\coordinate (O) at (0,0);
	\def\rr{1.5};
	\path[name path = l1] (\rr,2) -- (\rr,-2);
	\path[name path = l2,rotate around={45:(O)}] (\rr,2) -- (\rr,-2);
	\path[name path = l3,rotate around={90:(O)}] (\rr,2) -- (\rr,-2);
	\path[name path = l4,rotate around={135:(O)}] (\rr,2) -- (\rr,-2);
	\path[name path = l5,rotate around={180:(O)}] (\rr,2) -- (\rr,-2);
	\path[name path = l6,rotate around={-135:(O)}] (\rr,2) -- (\rr,-2);
	\path[name path = l7,rotate around={-90:(O)}] (\rr,2) -- (\rr,-2);
	\path[name path = l8,rotate around={-45:(O)}] (\rr,2) -- (\rr,-2);
	
	\path [name intersections = {of = l1 and l2,by=T1}];
	\path [name intersections = {of = l2 and l3,by=T2}];
	\path [name intersections = {of = l3 and l4,by=T3}];
	\path [name intersections = {of = l4 and l5,by=T4}];
	\path [name intersections = {of = l5 and l6,by=T5}];
	\path [name intersections = {of = l6 and l7,by=T6}];
	\path [name intersections = {of = l7 and l8,by=T7}];
	\path [name intersections = {of = l8 and l1,by=T8}];
	
	\draw[red,thick,decoration={markings,mark=at position 0.7 with {\arrow[scale=1.3,>=stealth]{>}}},postaction={decorate}] (T8) -- node[pos=0.7,right]{$l_1^*$} (T1);
	\draw[red,thick,decoration={markings,mark=at position 0.7 with {\arrow[scale=1.3,>=stealth]{>}}},postaction={decorate}] (T1) -- node[pos=0.7,above right]{$l_2^*$} (T2);
	\draw[red,thick,decoration={markings,mark=at position 0.7 with {\arrow[scale=1.3,>=stealth]{>}}},postaction={decorate}] (T2) -- node[pos=0.7,above]{$l_3^*$} (T3);
	\draw[red,thick,dashed,decoration={markings,mark=at position 0.75 with {\arrow[scale=1.3,>=stealth]{>}}},postaction={decorate}] (T3) -- (T4);
	\draw[red,thick,decoration={markings,mark=at position 0.7 with {\arrow[scale=1.3,>=stealth]{>}}},postaction={decorate}] (T4) -- node[pos=0.85,above left]{$l_{i-1}^*$} (T5);
	\draw[red,thick,decoration={markings,mark=at position 0.7 with {\arrow[scale=1.3,>=stealth]{>}}},postaction={decorate}] (T5) -- node[pos=0.7,below left]{$l_i^*$} (T6);
	\draw[red,thick,dashed,decoration={markings,mark=at position 0.7 with {\arrow[scale=1.3,>=stealth]{>}}},postaction={decorate}] (T6) -- (T7);
	\draw[red,thick,decoration={markings,mark=at position 0.7 with {\arrow[scale=1.3,>=stealth]{>}}},postaction={decorate}] (T7) -- node[pos=0.7,below right]{$l_M^*$} (T8);
	
	\draw[one end extended] (O) -- ($(T1)!(O)!(T2)$);
	\draw[one end extended] (O) -- ($(T2)!(O)!(T3)$);
	\draw[dashed, one end extended] (O) -- ($(T3)!(O)!(T4)$);
	\draw[one end extended] (O) -- ($(T4)!(O)!(T5)$);
	\draw[one end extended] (O) -- ($(T5)!(O)!(T6)$);
	\draw[dashed, one end extended] (O) -- ($(T6)!(O)!(T7)$);
	\draw[one end extended] (O) -- ($(T7)!(O)!(T8)$);
	\draw[one end extended] (O) -- ($(T8)!(O)!(T1)$);
	\draw[red] (T1) node{$\bullet$};
	\draw[red] (T2) node{$\bullet$};
	\draw[red] (T3) node{$\bullet$};
	\draw[red] (T4) node{$\bullet$};
	\draw[red] (T5) node{$\bullet$};
	\draw[red] (T6) node{$\bullet$};
	\draw[red] (T7) node{$\bullet$};
	\draw[red] (T8) node{$\bullet$};
	\draw (T1) node[above right]{$T_2$};
	\draw (T2) node[above right]{$T_3$};
	\draw (T3) node[above left]{$T_4$};
	\draw (T4) node[above left]{$T_{i-1}$};
	\draw (T5) node[below left]{$T_i$};
	\draw (T6) node[below ]{$T_{i+1}$};
	\draw (T7) node[below right]{$T_M$};
	\draw (T8) node[below right]{$T_1$};
	
	\draw ([shift=(20:0.3)]O) node{$e$};	
\end{tikzpicture}
	\caption{A dual face {\it (in red)} dual to the edge $e$ surrounded by a loop $(e_1^*e_2^*\cdots e_M^*e_1^*)$, with $e_i^*$ dual to triangle $t_i$. Two adjacent tetrahedra $T_i$ and $T_{i+1}$ (identifying $T_{M+1}\equiv T_1$) are glued along the triangle $t_i$.}
	\label{fig:order_tetra}
\end{figure}

Our starting point is the expression of the SGF in terms of the group integral similar to \eqref{eq:coherent_SN} \cite{Bonzom:2011nv}
\be
\cS(\{z_l,\zt_l\})
=\sum_{j_1\cdots j_6} \left(\prod_{n=1}^4(J_n+1)! \int_{\SU(2)}\rd h_n \right)
\prod_{l=1}^6 \frac{1}{(2j_l)!}
\left[ j_l, \zt_l \right| h_{t(l)}^{-1}h_{s(l)}\left|j_l,z_l \right>
\label{eq:SGF}
\ee
and the fact that (the bulk part of) the Ponzano-Regge amplitude \eqref{eq:group_form} can be written in the following way as the collection of local vertex, edge and face amplitudes. 
\be
\cA_{\bf T}[\cM] =\int \left[\rd \mu(z)\right]
\,\prod_{f^*} \left( \langle z^{f^*,T_1}|z^{f^*,T_1}\rangle-1 \right) 
\prod_{e^*} \left( e^{\sum_{l\in n} \langle z_l^{s(e^*)}|z_l^{t(e^*)}]} \right)
\prod_{v^*} \left( \int_{\SU(2)^4}\prod_{n=1}^4 \rd h_n\, e^{\sum_{l=1}^6 [\zt_l|h_{t(l)}^{-1}h_{s(l)}|z_l\rangle} \right)\,,
\label{eq:PR_cohe}
\ee
where each vertex amplitude is given by the coherent state \eqref{eq:coherent_SN} associated to a tetrahedron graph evaluated on identity. The proof of this re-expression is given in Appendix \ref{app:coherent_PR}.
This is a state-integral expression of the Ponzano-Regge model since, instead of summation of spins, it is written in terms of integrals over spinors. A similar state-integral model for 4D BF theory was explicitly constructed in \cite{Dupuis:2011fz}. 
The vertex amplitude in \eqref{eq:PR_cohe} no longer depends on the spins, thus it is expected to be irrelevant to the size of the tetrahedron it is associated to. 
The drawback of this construction is that the $\SU(2)$ group elements are still included in the formula other than spinors, making it hard to single out the geometrical information stored in the spinors themselves. Thus we aim to promote the integral expression \eqref{eq:PR_cohe} in order that only spinors are left in the integral expression. 

To get rid of the $\SU(2)$ group integral in the vertex amplitude, we replace the vertex amplitude by the SGF.
This is promising logically due to the fact that the vertex amplitude is simply the spin network state evaluation on the boundary of a tetrahedron. Therefore, the SGF, as the scaleless spin network evaluation, is a natural replacement of the $\{6j\}$-symbol.   

 Notice that the difference between \eqref{eq:SGF} from \eqref{eq:coherent_SN} is the factor $(J_n+1)!$ for each node. This can be cancelled out by modifying the edge amplitude, and the result is given in the following proposition. The notations are the same as in Proposition \ref{prop:state_integral_1}.
\begin{prop}
The spinfoam model can be expressed as a state-integral
\be
\cA_{\bf T}[\cM,\psi_\Gamma^{\sl} ]=\int \left[\rd \mu(z)\right] \prod_{f^*} \cA_{f^*}[z_{f^*}] \prod_{e^*}\cA_{e^*}[z_{e^*}]\prod_{v^*}\cA_{v^*}[z_{v^*}]
\nn\ee
with the vertex, edge and face amplitude written as
\begin{align}
\cA_{v^*}
&=\cS^{\sl}_{v^*}(\{z_l,\tz_l\})= \frac{1}{(1+\sum_{\cL}\prod_{v\subset \cL \,,\, l \prec l' } [z_l|z_{l'}\rangle)^2}
\label{eq:vertex_amplitude_new}\\
\cA_{e^*}
&=
\sum_{k=0}^{\infty} \f{1}{(k+1)!^2 (2k)!}\left( \sum_{l\in n}\langle z_l^{s(e^*)}|w_l^{t(e^*)}] \right)^{2k} 
=\,{}_0F_3(;2,2,\f12; \f{\left(\sum_{l\in n}  \langle z_l^{s(e^*)}|w_l^{t(e^*)}]\right)^2}{4}  ) 
\label{eq:edge_amplitude_new}\\
\cA_{f^*}
&=\langle z^{f^*,T_1}|z^{f^*,T_1}\rangle-1\,.
\label{eq:face_amplitude_new}
\end{align}	
\label{prop:state_integral_2}
$\psi^{\sl}_\Gamma$ is the boundary scaleless spin network state. 
\end{prop}
\begin{proof}
This is a re-arrangement of the vertex amplitude and edge amplitude compared to \eqref{eq:vertex_amplitude}-\eqref{eq:face_amplitude}. 
The difference between the vertex amplitude in \eqref{eq:PR_cohe} (or \eqref{eq:vertex_amplitude}) and \eqref{eq:vertex_amplitude_new} is a factorial factor. 
We first apply the identity
\be
\frac{1}{(M-1)!}=\frac{1}{2\pi i}\oint\limits_{|s|=r_0} \rd s \frac{e^s}{s^M}
\ee
to rewrite the vertex amplitude \eqref{eq:vertex_amplitude} so that it is related to the SGF \eqref{eq:SGF}:
\be
s_{\tet}^{\cohe}(\id)=
\sum_{j_1\cdots j_6}
\left(\prod_{n=1}^4 \frac{1}{2\pi i}\oint \rd s_n \frac{e^{s_n}}{s_n^{J_n+2}} (J_n+1)! 
 \int_{\SU(2)}\rd h_n \right) 
\prod_{l=1}^6 \frac{1}{(2j_l)!}
[ j_l, z_l | h_{t(l)}^{-1}h_{s(l)}|j_l,\zt_l \rangle \,.
\label{eq:vertex_amplitude_to_SGF}
\ee
Then we expand the exponential of the edge amplitude in \eqref{eq:PR_cohe} 
\be
e^{\sum_{l\in n} \langle z_l^{s(e^*)}|z_l^{t(e^*)}]}= \sum_{k_1\cdots k_3\in \N/2}
\prod_{l\in n} \frac{\langle k_l, z_l^{s(e^*)}|k_l, z_l^{t(e^*)}]}{(2k_l)!}\,.
\label{eq:expand_edge_amplitude}
\ee
The spinor integration in the total amplitude expression will select $k_l\equiv j_l\,,\,\forall l$ from \eqref{eq:vertex_amplitude_to_SGF} and \eqref{eq:expand_edge_amplitude}, thus we are safe to redefine $J_n= \sum_{l\in n}k_l$. This moves the contour integral from the vertex amplitude to the edge amplitude, leaving the vertex amplitude purely given by the SGF. As such, each edge amplitude absorbs two contour integrals, one from the tetrahedron $s(e^*)$ and the other from $t(e^*)$, and is written as
\be\begin{split}
\cA_{e^*}(\{z_l\},\{w_l\})
&=
\frac{1}{(2\pi i)^2}\oint\rd s\oint \rd t \,
\frac{e^{s+t}}{(st)^2} 
\sum_{k_1\cdots k_3\in\N/2} \prod_{l\in n}
 \frac{1}{(st)^{k_l}}\frac{\langle k_l, z_l^{s(e^*)}|k_l, z_l^{t(e^*)}]}{(2k_l)!}\\
&= 
\frac{1}{(2\pi i)^2}\oint\rd s\oint \rd t \,
\frac{1}{(st)^2}e^{s+t+\frac{1}{\sqrt{st}}\sum_{l\in n}\langle z_l^{s(e^*)}|w_l^{t(e^*)}] }\\
&=\frac{1}{(2\pi i)^2}\oint\rd s\oint \rd t\,
\sum_{k,m,u\in \N}\f{1}{k!m!u!} s^{k-2-\f{u}{2}} \, t^{m-2-\f{u}{2}} \left(\sum_{l\in n} \langle z_l^{s(e^*)}|w_l^{t(e^*)}]\right)^{k} \\
&=\sum_{k=0}^{\infty} \f{1}{(k+1)!^2 (2k)!}\left( \sum_{l\in n}\langle z_l^{s(e^*)}|w_l^{t(e^*)}] \right)^{2k}  \\
&=\,{}_0F_3(;2,2,\f12; \f{\left(\sum_{l\in n}  \langle z_l^{s(e^*)}|w_l^{t(e^*)}]\right)^2}{4}  ) 
\end{split}
\label{eq:new_ldge_w_contour}
\ee
thus \eqref{eq:edge_amplitude_new}. We have used the identity $\langle k_l, z_l^{s(e^*)}|k_l, z_l^{t(e^*)}] \equiv 
\langle z_l^{s(e^*)}|w_l^{t(e^*)}]^{2k_l}$ to get the second line. To arrived at the fourth line, we have identified $u$ with $2k-2$ and $m$ with $k$ followed with a change of variable $k\rightarrow k+1$ since $\f{1}{2\pi i}\oint\rd t \,t^u = 1$ for $u=-1$ and zero otherwise.                                                                                                                          
\end{proof}

To summarize, Proposition \ref{prop:state_integral_2} is a re-grouped version of \eqref{eq:PR_cohe}. In order to have a simple vertex amplitude with a nice geometrical interpretation, the edge amplitude becomes more complicated. However, we will see in Subsection \ref{sec:glue} that this new edge amplitude also possesses a nice geometrical interpretation that is compatible with that of the new vertex amplitude \eqref{eq:vertex_amplitude_new}.

\subsection{Topological invariance of the holomorphic blocks}
\label{sec:Topo_inv_coherence}
 The topological invariance of the new Ponzano-Regge state-integral formula can also be proven by performing the $2-3$ and $1-4$ Pachner moves. To do this, one can either express the vertex amplitudes \eqref{eq:vertex_amplitude_new} into the explicit form as the generating function of the $\{6j\}$-symbols \eqref{eq:SGF_1} then use the recursion relation of the $\{6j\}$-symbols \eqref{eq:EI} and \eqref{eq:41_identity}, or work on the spinors variables to prove the invariance of total amplitude under Pachner moves. The former approach is rather straightforward hence here we only illustrate the latter approach. We consider in this section the gluing of general 3-cells through triangles and show that the resulting total amplitude after gluing is given by the boundary scaleless spin network state on the union boundary, which means that the state-integral model is topological invariant. We will show it through three different types of gluing which serve as the elementary steps to glue disjoint 3-cells. The first type is to glue two 3-cells through one triangle with no extra internal edge produced (Proposition \ref{prop:gluing_3_edges}). The second type is to glue two adjacent triangles sharing one edge on the boundary of one 3-cell and produce one extra internal edge (Proposition \ref{prop:gluing_2_edges}). The third type is to glue two adjacent triangles sharing two edges on the boundary of one 3-cell and produce two internal edges and one internal vertex (Proposition \ref{prop:gluing_1_edge}). 
 The topological invariance of the Ponzano-Regge state-integral model, which can be directly derived by analyzing the total amplitude under a combination of these three types of gluing without changing the topological nature of the 3-cell, is a natural conclusion (Corollary \ref{cor:topological}).

The scaleless spin network state for a three-valent graph can be expressed into a Gaussian integral \cite{Bargmann:1962zz,Bonzom:2015ova}. Let us introduce an auxiliary spinor $|\xi_{nl}\rangle =\mat{c}{\xi_{nl}^0\\ \xi_{nl}^1} \in \bC^2$ attached to each half link $l$ incident to node $n$ with the same spinor measure $\rd \mu(\xi_{nl})=\f{1}{\pi^2}\rd \xi_{nl}^0\rd\bxi_{nl}^0\rd \xi_{nl}^1\rd\bxi_{nl}^1 e^{-\langle \xi_{nl}|\xi_{nl}\rangle}$. The dual auxiliary spinor is defined in the same way as $z$'s, $\ie$ $|\xi_{nl}]=\mat{c}{-\bxi_{nl}^1\\ \bxi_{nl}^0}$. 
To unify the notation, we denote the new spinors at the source $s(l)$ and target $t(l)$ of a link $l$ as $\xi_l$ and $\txi_l$. 
Consider a general 3-cell $\triangle$ whose boundary $\partial \triangle$ is made up by jointed triangles whose dual graph $(\partial \triangle)^*_1$ is a (closed) three-valent graph. The correspondent scaleless spin network state is \cite{Bonzom:2015ova} 
\be
\cS^{\sl}_{(\partial \triangle)^*_1}(\{z_l,\tz_l\})=\int\left( \prod_{l} \rd\mu(\xi_{l})\rd\mu(\txi_{l}) \right)
e^{\sum_{l} \langle \xi_{l}|\txi_{l}] + \sum_{\alpha} [z_{s(\alpha)}|z_{t(\alpha)}\rangle [\xi_{s(\alpha)}|\xi_{t(\alpha)}\rangle }\,,
\label{eq:SGF_aux}
\ee
where $\alpha$ is the angle formed by two links $s(\alpha)$ and $t(\alpha)$ incident to a same node with the cyclic order $s(\alpha)\prec t(\alpha)$. We denote $\alpha\in n$ if $s(\alpha),t(\alpha)\in n$. Clearly \eqref{eq:SGF_aux} contains a link-term $e^{\sum_{l} \langle \xi_{l}|\txi_{l}] }$ and an angle-term $e^{\sum_{\alpha} [z_{s(\alpha)}|z_{t(\alpha)}\rangle [\xi_{s(\alpha)}|\xi_{t(\alpha)}\rangle}$.

\begin{lemma}
\label{lemma:gluing_type}
Arbitrary gluing of 3-cells through triangles can be separated into a sequence of the following three types of gluing:
\begin{itemize}
	\item {\bf Type \ri}: identifying three edges of two disjoint triangles, each on the boundary of one 3-cell (fig.\ref{fig:gluing_3_edges_1}), whose result is one 3-cell with no extra internal edges produced (fig.\ref{fig:gluing_3_edges_2});
	\item {\bf Type \rii}: identifying the two remaining edges of two adjacent triangles sharing one edge on the boundary of one 3-cell (fig\ref{fig:gluing_2_edges_1}), whose result is one 3-cell with one extra internal edge produced (fig.\ref{fig:gluing_2_edges_2});
	\item {\bf Type \riii}: identifying the remaining edge of two adjacent triangles sharing two edges on the boundary of one 3-cell (fig.\ref{fig:gluing_1_edge_1}), whose result is one 3-cell with two extra internal edges and one internal vertex produced (fig.\ref{fig:gluing_1_edge_2}).
\end{itemize}
\end{lemma}
We consider separately these three types of gluing and analyze the total amplitude after gluing.
 \begin{prop}
 \label{prop:gluing_3_edges}
 For gluing of Type \ri, 
 the scaleless spin network states $\cS^{\sl}_{(\partial \triangle_1)^*_1}, \cS^{\sl}_{(\partial \triangle_2)^*_1}$ on the boundaries $\partial \triangle_1, \partial \triangle_2$ of two 3-cells $\triangle_1, \triangle_2$ glued with an edge amplitude $ \cA_{e^*}^{\partial \triangle_1 \cap\partial \triangle_2}$ in the form of \eqref{eq:edge_amplitude_new} produces a scaleless spin network state $\cS^{\sl}_{(\partial (\triangle_1 \cup \triangle_2))^*_1}$ on the union boundary $\partial(\triangle_1 \cup \triangle_2)$ after gluing.
 	Using the spinor notation in fig.\ref{fig:gluing_3_edges} (also introduced in the proof), it is symbolically expressed as
\be
\cS^{\sl}_{(\partial (\triangle_1 \cup \triangle_2))^*_1}
=  \int\left(\prod_{l,c\in \partial \triangle_1 \cap\partial \triangle_2} \rd\mu(z_l)\rd\mu(\tw_c)\right)
  	\cS^{\sl}_{(\partial \triangle_1)^*_1}(\{z_l,\tz_l\})\,
 \cA_{e^*}^{\partial \triangle_1 \cap\partial \triangle_2}
\,  	\cS^{\sl}_{(\partial \triangle_2)^*_1}(\{w_c,\tw_c\})\,.
\label{eq:gluing_2_cells}
\ee 
 \end{prop}
 \begin{proof}
\begin{figure}[h!]
\begin{minipage}{0.55\textwidth}
\centering
\includegraphics{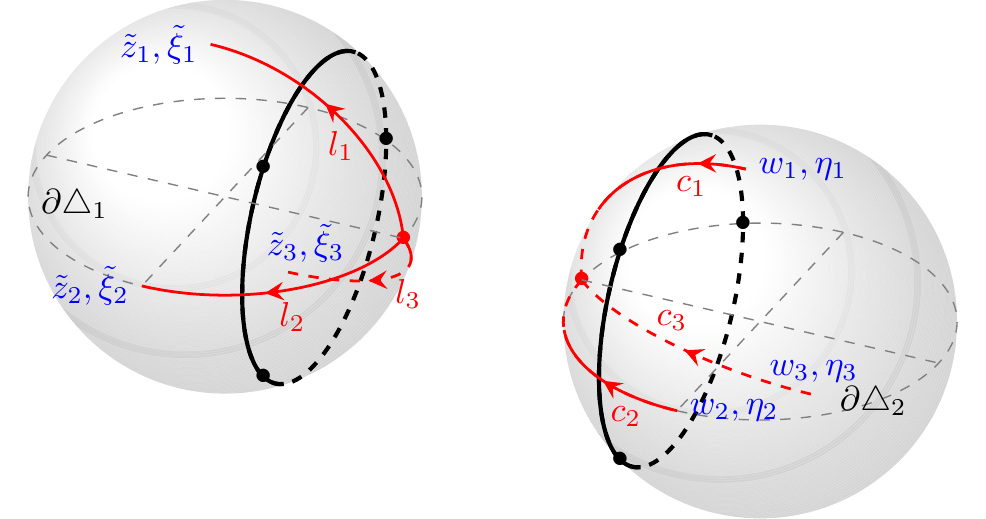}
\subcaption{}
\label{fig:gluing_3_edges_1}
\end{minipage}
\quad
\begin{minipage}{0.3\textwidth}
\centering
\includegraphics{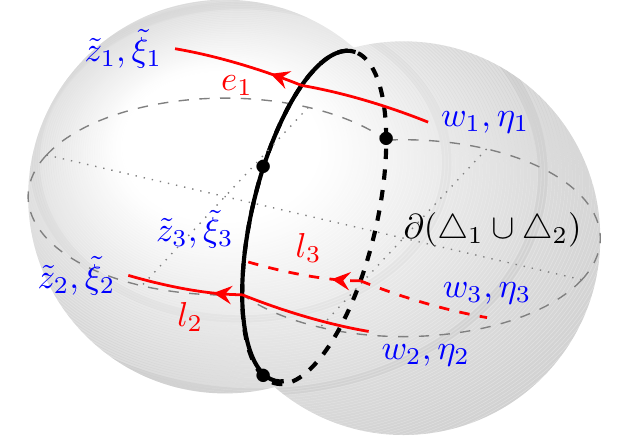}
\subcaption{}
\label{fig:gluing_3_edges_2}
\end{minipage}	
\caption{(a) Before the gluing of two 3-cells $\triangle_1$ and $\triangle_2$ {\it (visualized as 3-balls)} through identifying the left and right triangles {\it (in thick, each visualized as a circle embedding in the corresponding 2-sphere with three vertices on it)}. The links {\it(in red)}, each dual to one edge of a triangle, are assigned spinor information. Fix the orientation of the links $l_{1,2,3}$ on the left triangle to be outgoing, the sources are assigned the spinors $z_{1,2,3}$ and the auxiliary spinors $\xi_{1,2,3}$ (not shown in the figure for clear visualization) and the targets are assigned the spinors $\tz_{1,2,3}$ and the auxiliary spinors $\txi_{1,2,3}$; also fix the orientation of the links $c_{1,2,3}$ on the right triangle to be incoming, the sources are assigned the spinors $w_{1,2,3}$ and the auxiliary spinors $\eta_{1,2,3}$ and the targets are assigned the spinors $\tw_{1,2,3}$ and the auxiliary spinors $\teta_{1,2,3}$ (not shown in the figure for clear visualization). (b) After gluing two 3-cells $\triangle_1$ and $\triangle_2$ {\it (visualized as a double bubble)} . The spinor information in the bulk is integrated out and only the spinors on the union boundary $\partial (\triangle_1 \cup \triangle_2)$ are left.}
\label{fig:gluing_3_edges}
\end{figure}

 Let us write explicitly, using the definition of the scaleless spin network state \eqref{eq:SGF_aux}, the relevant part of the right-hand side of \eqref{eq:gluing_2_cells} in terms of the spinors on the two triangles to be glued (see fig.\ref{fig:gluing_3_edges_1}). 
 Denote the links in the graph $(\partial \triangle_1)^*_1$ on the left 3-cell as $l$'s. For each link $l$, the source $s(l)$ ({\it resp.} target $t(l)$) is assigned a spinor $z_l$ ({\it resp.} $\tz_l$) and an auxiliary spinor $\xi_l$ ({\it resp.} $\txi_l$). Also, denote the links in the graph $(\partial \triangle_2)^*_1$ on the right 3-cell as $c$'s. For each link $c$, the source $s(c)$ ({\it resp.} target $t(c)$) is assigned a spinor $w_c$ ({\it resp.} $\tw_c$) and an auxiliary spinor $\eta_c$ ({\it resp.} $\teta_c$). 
 With no loss of generality, we can fix the orientation of the relevant links.
We fix that the links $l_1,l_2,l_3$ on the left triangle are outgoing from the node dual to the left triangle and $c_1,c_2,c_3$ are incoming towards the node dual to the right triangle. Let $\triangle_1$ be the source of the dual edge dual to the glued triangle and $\triangle_2$ the target, then the edge amplitude for the gluing is
 \be
\cA_{e^*}^{\partial \triangle_1 \cap \partial \triangle_2}
={}_0F_3(;2,2,\f12; \f{\left(\sum_{l=1}^3  \langle z_l|\tw_l]\right)^2}{4}  )
=\sum_{k=0}^{\infty} \f{1}{(k+1)!^2 (2k)!}\left( \sum_{e=1}^3\langle z_l|\tw_l] \right)^{2k} \,.
 \nn\ee
The right hand side of \eqref{eq:gluing_2_cells} reads
\begin{multline}
 \int\left( \prod_{l'} \rd\mu(\xi_{l'})\rd\mu(\txi_{l'}) \prod_{c'} \rd\mu(\eta_{c'})\rd\mu(\teta_{c'}) \right)
e^{\sum_{l'} \langle \xi_{l'}|\txi_{l'}]
+ \sum_{c'} \langle \eta_{c'}|\teta_{c'}]
+\sum_{\alpha'} [z_{s(\alpha')}|z_{t(\alpha')}\rangle [\xi_{s(\alpha')}|\xi_{t(\alpha')}\rangle 
+\sum_{\beta'} [w_{s(\beta')}|w_{t(\beta')}\rangle [\eta_{s(\beta')}|\eta_{t(\beta')}\rangle }
 \\
\int \left(\prod_{l=1}^3\rd \mu(\xi_l)\rd\mu(\teta_l) \right) e^{\sum_{l=1}^3\left(\langle \xi_l|\txi_l]+\langle \eta_l|\teta_l] \right)}
\int\left(\prod_{l=1}^3 \rd\mu(z_l)\rd\mu(\tw_l)\right)
\exp\left[\sum_{\substack{l,l'=1,2,3\\l\prec l'}}[z_l|z_{l'}\rangle[\xi_{l}|\xi_{l'}\rangle 
+[\tw_l|\tw_{l'}\rangle[\teta_l|\teta_{l'}\rangle \right] 
\\
\left(\sum_{q\in\N}\f{1}{(q+1)!^2(2q)!}\left(\sum_{l=1}^3 \langle z_l|\tw_l] \right)^{2q} \right)\,.
\label{eq:T2}
 \end{multline}
 We have denoted the irrelevant part with primes in the first line.  
 The second line is the part of $\cS^{\sl}_{\partial \triangle_1}(\{z_l,\tz_l\})$ and $\cS^{\sl}_{(\partial \triangle_2)^*_1}(\{w_c,\tw_c\})$ relevant to the triangles to be glued, and the third line is the edge amplitude gluing $\triangle_1$ and $\triangle_2$.
We first perform the spinor integration for $z_{1,2,3}$ and $\tw_{1,2,3}$. 

Let us introduce the complex triples, following Bargmann's trick \cite{Bargmann:1962zz},
\be\ba{llll}
a=(z_1^0,z_2^0,z_3^0)\,,\quad &
b=(z_1^1,z_2^1,z_3^1)\,,\quad &
c=(\tw_1^0,\tw_2^0,\tw_3^0)\,,\quad &
d=(\tw_1^1,\tw_2^1,\tw_3^1)\,,\quad \\[0.2cm]
\alpha=(\xi_1^0,\xi_2^0,\xi_3^0)\,,\quad &
\beta=(\xi_1^1,\xi_2^1,\xi_3^1)\,,\quad &
\sigma=(\teta_1^0,\teta_2^0,\teta_3^0)\,,\quad &
\rho=(\teta_1^1,\teta_2^1,\teta_3^1)\,.
\ea
\label{eq:triples}
\ee
$\bar{a} =(\bar{a}_1,\bar{a}_2,\bar{a}_3)^{T}$ is the conjugate of $a$, and  
we denote the measure for the $a$ as $\rd\mu_3(a) = \f{1}{\pi^3}\rd a_1\rd\bar{a}_1\rd a_2\rd\bar{a}_2\rd a_3\rd\bar{a}_3 $. Likewise for $b,c,d,\alpha,\beta,\sigma,\rho$. 
For simplicity, we also denote $\Sigma=\alpha\times\beta = ([\xi_2|\xi_3\rangle, [\xi_3|\xi_1\rangle, [\xi_1|\xi_2\rangle)$ and $\Lambda = \sigma\times\rho=([\teta_2|\teta_3\rangle,[\teta_3|\teta_1\rangle,[\teta_1|\teta_2\rangle $.
We then arrange them in the matrices 
\be
A=\mat{ccc}{\Sigma_1 & \Sigma_2 & \Sigma_3 \\ 
 a_1 & a_2 & a_3 \\ b_1 & b_2 & b_3 }\,,\quad
B=\mat{ccc}{\Lambda_1 & \Lambda_2 & \Lambda_3 \\ 
 c_1 & c_2 & c_3 \\ d_1 & d_2 & d_3}\,,\quad
\Gamma = \mat{ccc}{0 & 0 & 0 \\ 0 & 0 & -1 \\ 0 & 1 & 0}\,.
\ee

One can thus rewrite the integral of $z_{1,2,3}$ and $\tw_{1,2,3}$ in the second and third lines of \eqref{eq:T2} in a compact way:
\be\begin{split}
&\f{1}{(2\pi i)^2}\oint\rd t \oint \rd s \f{e^{s+t}}{(st)^2}\,
\int\left(\prod_{l=1}^3 \rd\mu(z_l)\rd\mu(\tw_l)\right)\,
e^{\det A+\det B + \f{1}{\sqrt{st}}\tr (A^\dagger \Gamma \bar{B}) }\\
=& \f{1}{(2\pi i)^2}\oint\rd t \oint \rd s \f{e^{s+t}}{(st)^2}\,
\int\rd\mu_3 (a)\rd\mu_3 (b)\rd\mu_3 (c) \rd\mu_3 (d)\,
e^{-\bar{a}\cdot a - \bar{b}\cdot b -\bar{c}\cdot c - \bar{d}\cdot d }
e^{(\Sigma\times b)\cdot a + (\Lambda\times d)\cdot c}
e^{\f{1}{\sqrt{st}}(\bar{b}\cdot \bar{c}-\bar{a}\cdot \bar{d} )}\,,
\end{split}
\label{eq:T2_part}
\ee
where we have used the contour integral expression for the inverse Gamma function $\f{1}{(q+1)!}= \f{1}{2\pi i} \oint \rd t \f{e^t}{t^{q+2}}$. One can calculate this Gaussian integral for $a,b,c,d$ one by one. Note that 
given a complex $n$-ple ${\bf v}$ and a complex $n\times n$ matrix $A$ whose Hermitian part is positive definite, one has the Gaussian integral
\be
\int \rd\mu_n({\bf v}) e^{-\bar{{\bf v}}\cdot A {\bf v} + {\bf u} \cdot {\bf v}+ {\bf u}'\cdot \bar{{\bf v}}}
=\det A^{-1} e^{{\bf u} \cdot A^{-1}{\bf u}'}\,,
\label{eq:Gaussian_integral}
\ee
where ${\bf u}$ and ${\bf u}'$ are independent $n$-ples.  

After integrating out $a,b,c$, one can use \eqref{eq:Gaussian_integral} to perform the remaining integral for $d$: 
\be
\int\rd\mu_3 (d)\, e^{-\bar{d}\cdot d}\, e^{\f{1}{st}(\Lambda\times \bar{d})\cdot (\Sigma\times d )}
= \int\rd\mu_3 (d)\, e^{\bar{d}\cdot(\id- \f{{\bf M}}{st}) d}
=\f{1}{\det(\id-\f{{\bf M}}{st})}=\f{1}{(1-\f{\Lambda\cdot \Sigma}{st})^{2}}\,,
\ee
where the matrix ${\bf M}$ has entries
\be
{\bf M}_{ij}=(\Lambda \cdot\Sigma ) \delta_{ij} - \Lambda_i \Sigma_j\,,\quad i,j=1,2,3
\ee
so that $\bar{d}\cdot {\bf M} d = (\Lambda\cdot \Sigma )(\bar{d}\cdot d) - (\Lambda\cdot d)(\Sigma\cdot \bar{d})=(\Lambda\times \bar{d})\cdot (\Sigma\times d )$.
The explicit form of $\Lambda\cdot \Sigma$ is
\be\begin{split}
\Lambda \cdot \Sigma =& (\alpha \times \beta )\cdot (\sigma \times \rho) =
(\alpha\cdot\sigma)(\beta\cdot \rho) - (\alpha\cdot \rho)(\beta\cdot \sigma)\\[0.2cm]
=& [\xi_1|\xi_2\rangle [\teta_1|\teta_2\rangle + [\xi_2|\xi_3\rangle [\teta_2|\teta_3\rangle + [\xi_3|\xi_1\rangle [\teta_3|\teta_1\rangle  \,.
\end{split}
\label{eq:lambda_sigma}
\ee
Now we perform the counter integral to complete the integration in \eqref{eq:T2_part}:
\be\begin{split}
&\f{1}{(2\pi i)^2}\oint\rd t \oint \rd s \f{e^{s+t}}{(st)^2}\,
\f{1}{(1-\f{\Lambda\cdot \Sigma}{st})^2} \\
=&\f{1}{(2\pi i)^2}\oint\rd t \oint \rd s 
\sum_{m,u,k=0}^{\infty}\f{k+1}{m!u!}\f{s^mt^u}{s^2t^2}\left(\f{\Lambda\cdot \Sigma}{st} \right)^k \\
=&\sum_{k=0}^\infty \f{1}{k!(k+1)!}(\Lambda\cdot \Sigma)^k \\
=&C_1(\Lambda\cdot \Sigma )\,,
\end{split}
\label{eq:result_of_trick}
\ee
where $C_1(z)$ is the Bessel-Clifford function of order one. A Bessel-Clifford function of order $m$ expands as $C_m(x):= \sum_{k=0}^{\infty} \frac{x^k}{k!(k+m)!}$. The explicit form \eqref{eq:lambda_sigma} of $\Lambda\cdot \Sigma$ allows us to express $C_1(\Lambda\cdot \Sigma)$ into an $\SU(2)$ integral by the following beautiful identity \cite{Bonzom:2015ova} 
\be
\int_{\SU(2)}\,\rd g\, e^{\sum_{i}[z_i|g|\tz_i\rangle} =
\sum_{k=0}^{\infty}\f{1}{k!(k+1)!} \left( \sum_{i<j}[z_i|z_j\rangle[\zt_i|\zt_j\rangle \right)^k\,,\quad
\forall z_i,\zt_i \in \bC^2
\,.
\label{eq:polynomial_identity}
\ee
Therefore, 
\be
C_1(\Lambda\cdot \Sigma)= \int_{\SU(2)}\rd g\, e^{\sum_{l=1}^3 [\xi_l|g|\teta_l \rangle}\,.
\label{eq:BC_identity}
\ee

We next combine this result with the integral of the auxiliary spinors in the second line of \eqref{eq:T2}. It is straightforward to calculate that
\be\begin{split}
&\int_{\SU(2)}\rd g \int \left(\prod_{l=1}^3\rd \mu(\xi_l)\rd\mu(\teta_l) \right) 
e^{\sum_{l=1}^3\left(\langle \eta_l|\teta_l] +[\teta_l|g|\xi_l\rangle+ \langle \xi_l|\txi_l] \right)}\\
=& \int_{\SU(2)}\rd g\, e^{\sum_{l=1}^3 \langle \eta_l|g|\txi_l]}
\end{split}
\label{eq:T2_result}
\ee
Finally, one performs the $\SU(2)$ transformation with $g$ on all the auxiliary spinors $(\xi_l,\txi_l)\rightarrow (g\xi_l\,, g \txi_l)$ from $\triangle_1$, which preserves the inner products $\langle \xi_{l'}|\txi_{l'}]$ and $[\xi_{s(\alpha')}|\xi_{t(\alpha')}\rangle$ of the irrelevant auxiliary spinors (see the first line of \eqref{eq:T2}). Thanks to the $\SU(2)$-invariant property of the spinor Haar measure, one can rewrite \eqref{eq:T2_result} into $\int_{\SU(2)}\rd g\, e^{\sum_{l=1}^3 \langle \eta_l|\txi_l]}= e^{\sum_{l=1}^3 \langle \eta_l|\txi_l]}$, which is exactly the link term for a general scaleless spin network state (see \eqref{eq:SGF_aux} and fig.\ref{fig:gluing_3_edges_2}). Combining this result with the irrelevant part in the first line of \eqref{eq:T2}, one arrives at a scaleless spin network state $\cS^{\sl}_{(\partial(\triangle_1\cup \triangle_2))^*_1}$ on the union boundary $\partial(\triangle_1\cup \triangle_2)$ after gluing.
 \end{proof}

\begin{prop}
\label{prop:gluing_2_edges}
For gluing of Type \rii, the scaleless spin network state $\cS_{(\partial \triangle)^*_1}^{\sl}$ on the boundary $\partial \triangle$ of a 3-cell $\triangle$ glued with an edge amplitude $\cA_{e^*}^{t \cap t_2=e}$ of the form \eqref{eq:edge_amplitude_new}, which is for two adjacent triangles $t_1,t_2 \in \partial \triangle$ sharing one edge $e$, and a face amplitude $\cA_{f^*}^{e}$, which is for the shared edge $e$, produces a scaleless spin network state $\cS_{(\partial \triangle')^*_1}^{\sl}$ on the resulting 3-cell boundary $\partial \triangle'$. It is symbolically expressed as
\be
\cS_{(\partial \triangle')^*_1}^{\sl} = \int \left(\prod_{e\in t_1\cup t_2} \rd\mu(z_l)\right)
\cS_{(\partial \triangle)^*_1}^{\sl} \, \cA_{e^*}^{t_1 \cap t_2=e}\, \cA_{f^*}^{e}\,.
\label{eq:gluing_2_edges}
\ee
\end{prop}

\begin{proof}
\begin{figure}[h!]
\centering
\begin{minipage}{0.45\textwidth}
\centering
\includegraphics{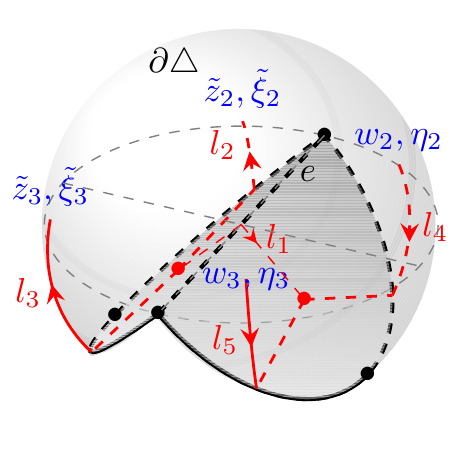}
\subcaption{}
\label{fig:gluing_2_edges_1}
\end{minipage}
\quad
\begin{minipage}{0.45\textwidth}
\includegraphics{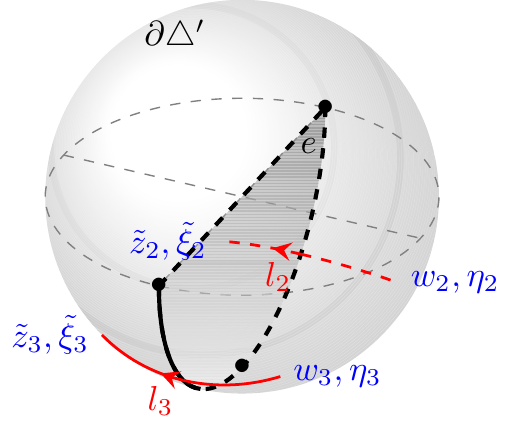}
\subcaption{}
\label{fig:gluing_2_edges_2}
\end{minipage}
\caption{Gluing two adjacent triangles {\it (in thick, each visualized as a semicircle and one edge connecting the two ends)} sharing one edge $e$ on one continuous boundary {\it(visualized as the boundary of a 3-ball removing a lemon slice)}. (a) Before gluing: The left triangle is dual to a node with three links $l_{1,2,3}$ outgoing, while the right triangle is dual to the other node with three links $l_{1,4,5}$ incoming. For $l_1$, the source is assigned $z_1,\xi_1$ and the target is assigned $\tw_1,\teta_1$. For links $l_{2,3}$, the sources are assigned $z_{2,3},\xi_{2,3}$ and the targets are assigned $\tz_{2,3},\txi_{2,3}$. For links $l_{4,5}$, the source are assigned $w_{2,3},\eta_{2,3}$ and the targets are assigned $\tw_{2,3},\teta_{2,3}$. (Only part of the spinors are shown for clear visualization). (b) After gluing: The shared edge $e$ becomes internal and the other two edges from different triangles collapse {\it (in thick)}. Spinor information in the bulk is integrated out and only the spinors on the resulting boundary (as shown) are left.
}
\label{fig:gluing_2_edges}
\end{figure}

Before gluing, the 3-cell boundary $\partial \triangle$ and the dual graph is as shown in fig.\ref{fig:gluing_2_edges_1}. The right-hand side of \eqref{eq:gluing_2_edges} reads (we again denote the irrelevant part with primes.)
\begin{multline}
\int\left( \prod_{l'} \rd\mu(\xi_{l'})\prod_{l'}\rd\mu(\txi_{l'}) \right)
e^{\sum_{l'} \langle \xi_{l'}|\txi_{l'}]
+ \sum_{\alpha'} [z_{s(\alpha')}|z_{t(\alpha')}\rangle [\xi_{s(\alpha')}|\xi_{t(\alpha')}\rangle }\\
\int \left(\prod_{l=1}^3\rd \mu(\xi_l)\rd\mu(\teta_l) \right) 
e^{\langle \xi_1|\teta_1]+\sum_{l=1}^2\left(\langle \xi_l|\txi_l]+ \langle \eta_l|\teta_l] \right)}\\
\int\left(\prod_{l=1}^3 \rd\mu(z_l)\rd\mu(\tw_l)\right)
\exp\left[\sum_{\substack{l,l'=1,2,3\\l\prec l'}}[z_l|z_{l'}\rangle[\xi_{l}|\xi_{l'}\rangle 
+[\tw_l|\tw_{l'}\rangle[\teta_l|\teta_{l'}\rangle \right] \\
\left(\sum_{q\in\N}\f{1}{(q+1)!^2(2q)!}\left(\sum_{l=1}^3\langle z_l|\tw_l]\right)^{2q} \right)
(\langle z_1|z_1\rangle-1)\,.
\label{eq:T3}
 \end{multline}
 
With no loss of generality, we have fixed the half links incident to the left node to be outgoing and denote the spinors and auxiliary spinors on these half links to be $z_{1,2,3}$ and $\xi_{1,2,3}$ respectively, and fixed the half links incident to the right node to be incoming and denote the spinors and auxiliary spinors on these half links to be $\tw_{1,2,3}$ and $\teta_{1,2,3}$ respectively. The difference of the relevant integral in \eqref{eq:T3} (the last three lines) from that in \eqref{eq:T2} are the terms $e^{\langle \xi_1|\teta_1]}$, since the two auxiliary spinors $\xi_1$ and $\teta_1$ are from the same link $l_1$, and $(\langle z_1|z_1\rangle-1)$ which is the face amplitude on edge $e$ given that the left node is chosen to be the base node. 

We now claim that the face amplitude $(\langle z_1|z_1\rangle-1)$ can be replaced by $(\langle \xi_1|\xi_1\rangle-1)$ without changing the amplitude. To prove that, we first rewrite $e^{\langle \xi_1|\teta_1]}$ by introducing two intermediate spinor integrals,
\be
e^{\langle \xi_1|\teta_1]}
=\int\rd\mu(\txi_1)\rd\mu(\eta_1) e^{\langle \xi_1|\txi_1]+ \langle \eta_1|\teta_1]} \,
e^{[\txi_1|\eta_1 \rangle}
\ee
so that the second line in \eqref{eq:T3} can be expressed in a symmetric way:
\be
\int\rd\mu(\txi_1)\rd\mu(\eta_1) e^{[\txi_1|\eta_1 \rangle}
 \int \left(\prod_{l=1}^3\rd \mu(\xi_l)\rd\mu(\teta_l) \right) 
e^{\sum_{l=1}^3\left(\langle \xi_l|\txi_l]+ \langle \eta_l|\teta_l] \right)}\,.
\label{eq:T3_line3}
\ee
We next notice for the edge amplitude that each term of the summation is a homogenous holomorphic polynomial of spinors $z_{1,2,3}$ and $\tw_{1,2,3}$ of order $2q$. Each term of order $2q$ can survive under the spinor integration only by matching with a homogenous anti-holomorphic polynomial of spinors $z_{1,2,3}$ and $\tw_{1,2,3}$ of order $2q$. This means we can safely move the term $\sum_{q\in \N}\f{1}{(q+1)!^2}$ from the last line of \eqref{eq:T3} to the third line, then the last two lines become
\be\begin{split}
&\int\left(\prod_{l=1}^3 \rd\mu(z_l)\rd\mu(\tw_l)\right)
\sum_{k\in \N}\f{1}{k!}
\left(\sum_{\substack{l,l'=1,2,3\\l\prec l'}}[z_l|z_{l'}\rangle[\xi_{l}|\xi_{l'}\rangle 
+[\tw_l|\tw_{l'}\rangle[\teta_l|\teta_{l'}\rangle \right)^{k}\delta_{k,2q} \\
&\left(\sum_{q\in \N}\f{1}{(q+1)!^2(2q)!} \left(\sum_{l=1}^3\langle z_l|\tw_l]\right)^{2q} \right)
(\langle z_1|z_1\rangle-1)\\
=& \int\left(\prod_{l=1}^3 \rd\mu(z_l)\rd\mu(\tw_l)\right)
\sum_{q\in \N}\f{1}{(q+1)!^2(2q)!}
\left(\sum_{\substack{l,l'=1,2,3\\l\prec l'}}[z_l|z_{l'}\rangle[\xi_{l}|\xi_{l'}\rangle 
+[\tw_l|\tw_{l'}\rangle[\teta_l|\teta_{l'}\rangle \right)^{2q} \\
&e^{ \sum_{l=1}^3\langle z_l|\tw_l]} (\langle z_1|z_1\rangle-1)\,.
\end{split}
\label{eq:T3_part}
\ee
Now the third line in \eqref{eq:T3_part} is a summation of homogenous holomorphic polynomial of the auxiliary spinors $\xi_{1,2,3}$ and $\teta_{1,2,3}$, each of order $2q$. For the same reason, one can further move the term $\sum_{q\in \N}\f{1}{(q+1)!^2}$ to the second line of \eqref{eq:T3}. 
We also separate $z_l$ and $\tw_l$ in the last line of \eqref{eq:T3_part} by adding six intermediate spinor integral over $\tz_{1,2,3}$ and $w_{1,2,3}$. 
As a result, the last three lines of \eqref{eq:T3} can be rewritten as
\begin{multline}
\int\rd\mu(\txi_1)\rd\mu(\eta_1) e^{[\txi_1|\eta_1 \rangle}
 \int \left(\prod_{l=1}^3\rd \mu(\xi_l)\rd\mu(\teta_l) \right) 
\sum_{q\in \N}\f{1}{(q+1)!^2(2q)!}\left(\sum_{e=1}^3\left(\langle \xi_l|\txi_l]+ \langle \eta_l|\teta_l] \right)\right)^{2q}\\ 
\int\left(\prod_{l=1}^3 \rd\mu(z_l)\rd\mu(\tw_l)\right)
\exp\left[\sum_{\substack{l,l'=1,2,3\\l\prec l'}}[z_l|z_{l'}\rangle[\xi_{l}|\xi_{l'}\rangle 
+[\tw_l|\tw_{l'}\rangle[\teta_l|\teta_{l'}\rangle \right]\\
\int \left( \prod_{l=1}^3 \rd\mu(\tz_l)\rd\mu(w_l)\right) 
e^{\sum_{l=1}^3 [\tz_l|w_l\rangle }
\,e^{ \sum_{l=1}^3\langle z_l|\tz_l]+\langle w_3|\tw_l]}
(\langle z_1|z_1\rangle-1)\,.
\label{eq:T3_relevant}
\end{multline}
We now realize from comparing \eqref{eq:T3_relevant} and the last three lines of \eqref{eq:T3} that $\xi_{1,2,3}$ and $z_{1,2,3}$ take the exchanged expressions. This means one can simply replace the face amplitude $(\langle z_1|z_1\rangle-1)$ in \eqref{eq:T3} with $(\langle \xi_1|\xi_1\rangle -1)$ without changing the total amplitude.
This allows us to use the same Bargmann's trick as in \eqref{eq:triples}-\eqref{eq:result_of_trick} as well as the identity \eqref{eq:BC_identity} so that the last three lines of \eqref{eq:T3} arrives at
\be\begin{split}
&\int\rd\mu(\txi_1)\rd\mu(\eta_1)\rd\mu(\xi_1)  e^{[\txi_1|\eta_1 \rangle}
\int_{\SU(2)}\rd g\, e^{ \langle \xi_1|\txi_1]+\la \eta_1| g|\xi_1\ra}\,e^{ \langle \eta_2|g|\txi_2]+\langle \eta_3|g|\txi_3]}
(\langle \xi_1|\xi_1\rangle -1)\\
=&\int_{\SU(2)}\rd g \int\rd\mu(\xi_1) (\langle \xi_1|\xi_1\rangle -1) e^{\langle \xi_1|g|\xi_1 \rangle} 
e^{ \langle \eta_2|g|\txi_2]+\langle \eta_3|g|\txi_3]}\\
=&\int_{\SU(2)}\,\rd g\, \delta(g)\, e^{ \langle \eta_2|g|\txi_2]+\langle \eta_3|g|\txi_3]}\\
=& e^{ \langle \eta_2|\txi_2]+\langle \eta_3|\txi_3]}\,,
\end{split}
\label{eq:T3_result}
\ee
which is the link term for a general scaleless spin network state.
To arrive at the third line of \eqref{eq:T3_result}, we have used the identity for delta distribution on $\SU(2)$
\be
\delta(g) =\int\rd\mu(z)(\langle z|z\rangle-1 )e^{\langle z|g|z\rangle}
\,.
\ee
Combining the result of \eqref{eq:T3_result} and the irrelevant part in the first line of \eqref{eq:T3}, one arrives at a scaleless spin network state $\cS_{(\partial \triangle')^*_1}^{\sl}$ on the 3-cell boundary $\partial \triangle'$ after gluing. 
\end{proof}

Define a function of the face amplitude, which is defined with the spinor $z$ sitting on the base node, as
\be
G(\cA_{f^*}(z)):=e^{-(1+\cA_{f^*}(z))}=e^{-\langle z|z\rangle}\,.
\label{eq:G}
\ee
We then have the following proposition.

\begin{prop}
\label{prop:gluing_1_edge}
For gluing of Type \riii, the scaleless spin network state $\cS_{(\partial \triangle)^*_1}^{\sl}$ on the boundary $\partial \triangle$ of a 3-cell $\triangle$ glued with an edge amplitude $\cA_{e^*}^{t_1\cap t_2=\{e_1,e_2\}}$ in the form of \eqref{eq:edge_amplitude_new}, which is for two adjacent triangles $t_1,t_2\in \partial \triangle$ sharing two edges $e_1,e_2$, a face amplitude $\cA_{f^*}^{e_1}$ in the form of \eqref{eq:face_amplitude_new}, which is for one of the shared edge $e_1$, and the function $G(\cA_{f^*}^{e_2})$ of the face amplitude for the other shared edge $e_2$  in the form of \eqref{eq:G}, produces a scaleless spin network state $\cS_{(\partial \triangle')^*_1}^{\sl}$ on the resulting 3-cell boundary $\partial \triangle'$. It is symbolically expressed as
\be
\cS_{(\partial \triangle')^*_1}^{\sl} 
=\int\left( \prod_{l\tilde{\in} t_1\cup t_2}\rd\mu(z_l) \right)
\cS_{(\partial \triangle)^*_1}^{\sl}\, \cA_{e^*}^{t_1\cap t_2=\{e_1,e_2\}}\, \cA_{f^*}^{e_1} \,G(\cA_{f^*}^{e_2})\,,
\label{eq:gluing_1_edge}
\ee
where the $l\tilde{\in} t_1\cup t_2$ denotes that the integral is for all the spinors on (the dual node of) the two triangles $t_1$ and $t_2$ except the one defining the face amplitude $\cA_{f^*}^{e_2}$.
\end{prop}

\begin{proof}
\begin{figure}[h!]
\centering
\begin{minipage}{0.45\textwidth}
\centering
\includegraphics{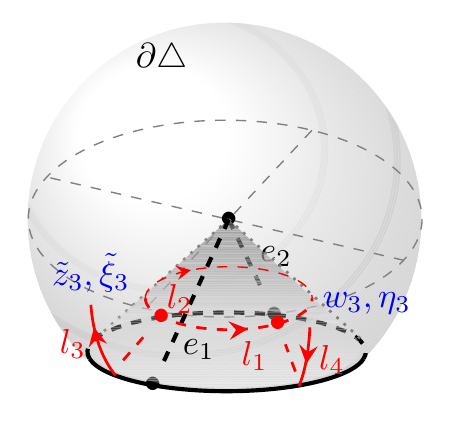}
\subcaption{}
\label{fig:gluing_1_edge_1}
\end{minipage}
\quad
\begin{minipage}{0.45\textwidth}
\centering
\includegraphics{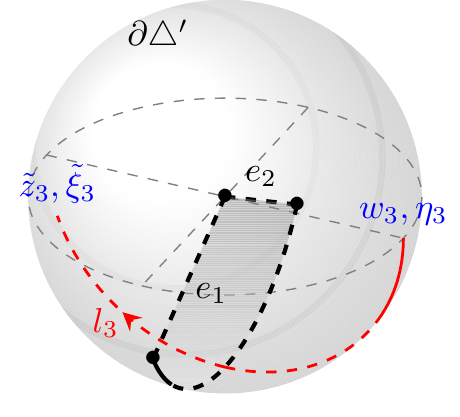}
\subcaption{}
\label{fig:gluing_1_edge_2}
\end{minipage}
\caption{Gluing of the two adjacent triangles {\it (in thick, each visualized as a semicircle and two edges incident to the origin of the 3-ball)} sharing two edges on one continuous boundary {\it(visualized as the boundary of a 3-ball removing a solid cone)}. (a) Before gluing: The left triangle is dual to one node with three links $l_{1,2,3}$ outgoing, while the right triangle is dual to the other node with three links $l_{1,2,4}$ incoming. 
$s(l_1)$ ({\it resp.} $t(l_1)$) is assigned $z_1,\xi_1$ ({\it resp.} $\tw_1,\teta_1$), $s(l_2)$ ({\it resp.} $t(l_2)$) is assigned $z_2,\xi_2$ ({\it resp.} $\tw_2,\teta_2$), $s(l_3)$ ({\it resp.} $t(l_3)$) is assigned $z_3,\xi_3$ ({\it resp.} $\zt_3,\txi_3$), and $s(l_4)$ ({\it resp.} $t(l_4)$) is assigned $\xi_3,\eta_3$ ({\it resp.} $\tw_3,\teta_3$). 
(only part of the spinors are shown for clear visualization). (b) After gluing: The shared edges become internal and one extra internal vertex (the origin of the 3-ball) is created. The remaining edge from different triangles collapses and lives on the resulting boundary. Spinor information in the bulk is integrated out and only the spinors on the resulting boundary, $\ie$ $\tz_3,\txi_3$ and $w_3,\eta_3$, are left.
}

\label{fig:gluing_1_edge}
\end{figure}
 This gluing can be visualized as collapsing a cone into a triangle bounded with two edges connecting the apex and two points on the base circle, as shown in fig.\ref{fig:gluing_1_edge_1}. 
With no loss of generality, we again consider half links associated to one of the nodes (the left one in fig.\ref{fig:gluing_1_edge_1}) are outgoing and are assigned the spinors $z_{1,2,3}$ and half links associated to the other node are incoming and are assigned spinors $\tw_{1,2,3}$. 
With this setting, it is not hard to get a similar expression for the relevant part of the right-hand side of \eqref{eq:gluing_1_edge} as in \eqref{eq:T3_result} after integrating out $z_{1,2,3}$, $\tw_{1,2,3}$ and other auxiliary spinors.
We choose the two (function of) face amplitudes to be
\be
\cA_{f^*}^{e_1}(z_1) =\la z_1|z_1 \ra -1\,,\quad
G(\cA_{f^*}^{e_2}(z_2)) = e^{-\la z_2|z_2 \ra}\,.
\ee
The same analysis as in the proofs of Proposition \ref{prop:gluing_3_edges} and \ref{prop:gluing_2_edges} leads the simple expression of the right-hand side of \eqref{eq:gluing_1_edge}(we use here the spinors $z_{1,2}$ instead of the auxiliary spinors $\xi_{1,2}$ as in \eqref{eq:T3_result} which does not change the final result.):
\be\begin{split}
&\int_{\SU(2)}\rd g \int\rd\mu(z_1)
(\langle z_1|z_1\rangle -1) 
\,e^{\langle z_1|g|z_1 \rangle}\,e^{\langle z_2|g|z_2 \rangle -\langle z_2|z_2\rangle} e^{\langle \eta_4|g|\txi_3]}\\
=&\int_{\SU(2)}\rd g \,\delta(g)e^{\langle z_2|g|z_2 \rangle-\langle z_2|z_2\rangle}  e^{\langle \eta_4|\txi_3]}\\
=&e^{\langle \eta_4|\txi_3]} \,.
\label{eq:T4_result}
\end{split}
\ee
The result is again the link-term for the general scaleless spin network state. Together with the irrelevant part of the right-hand side of \eqref{eq:gluing_1_edge}, we arrive at the scaleless spin network state for the 3-cell boundary after gluing, thus the left-hand side of \eqref{eq:gluing_1_edge}.
\end{proof}

Note that \eqref{eq:gluing_1_edge} is a finite equation because we did not perform the integration over $z_2$. The dependence of the final result on $z_2$ is removed simply by $\delta(g)$. If one replaces $G(\cA_{f^*}^{e_2})$ by $\cA_{f^*}^{e_2}(z_2)$ and further integrate $z_2$ on the right hand side, one gets a delta distribution of $\SU(2)$ group evaluated on the identity on the left hand side, $\ie$
\be
\cS_{(\partial \triangle')^*_1}^{\sl} \,\delta_{\SU(2)}(\id)
=\int\left( \prod_{l\in t_1\cup t_2}\rd\mu(z_l) \right)
\cS_{(\partial \triangle)^*_1}^{\sl}\cdot \cA_{e^*}^{t_1\cap t_2=\{e_1,e_2\}} \cdot 
\cA_{f^*}^{e_1}\cdot \cA_{f^*}^{e_2} \,,
\label{eq:gluing_1_edge_diverge}
\ee
which diverges and the divergence corresponds to the translational symmetry of the extra internal vertex produced after gluing as in the original Ponzano-Regge state-sum model \cite{Freidel:2002dw,Freidel:2004vi}.
In the state-sum model, this divergence can be eliminated by fixing the spin (thus edge length, say $\ell_{l}$,) of one internal edge incident to this internal vertex.
 In a similar spirit,
 \eqref{eq:gluing_1_edge} can be viewed as the gauge-fixing version of \eqref{eq:gluing_1_edge_diverge} which fixes one spinor $z_2$ in the bulk. Another way of gauge fixing is to gauge fix (the norm of) the inner product $|[z_1|z_2\rangle|$ of the two spinors $z_1,z_2$ used to define the face amplitude $\cA_{f^*}^{e_1}$ and $\cA_{f^*}^{e_2}$, which encodes the angle information of the triangle dual to the node that $z_1,z_2$ lives on (see Section \ref{sec:geo_SGF} below). This can be seen by rewriting $\delta_{\SU(2)}(\id)$ in terms of $|[z_1|z_2\rangle|$ as
\be\begin{split}
\delta_{\SU(2)}(\id)=
&\int_{\SU(2)}\rd g\int\rd\mu(z_1)\rd\mu(z_2) \, 
(\langle z_1|z_1 \rangle -1)(\langle z_2|z_2 \rangle -1) e^{\langle z_1|g|z_1 \rangle} e^{\langle z_2|g|z_2\rangle} \\
=&\int\rd\mu(z_1)\rd\mu(z_2) \, (\langle z_1|z_1 \rangle -1)(\langle z_2|z_2 \rangle -1)
e^{|[z_1|z_2\rangle|^2}\,.
\end{split}
\ee
One can straightforwardly conclude from Lemma \ref{lemma:gluing_type} and Proposition \ref{prop:gluing_3_edges}, \ref{prop:gluing_2_edges}, \ref{prop:gluing_1_edge} the following corollary.

\begin{coro}
\label{cor:topological}
The Ponzano-Regge state-integral formula given in Proposition \ref{prop:state_integral_2} is topological invariant. This means the total amplitude is independent of the bulk configuration and is equal to the boundary scaleless spin network state upon gauge fixings, one for each internal vertex.
\end{coro}
In particular, in the $2-3$ Pachner move, the gluing of two tetrahedra is of Type \ri; the gluing of three tetrahedra includes three steps, two of which are of Type \ri\, and the other is of Type \rii. In the $1-4$ Pachner move, the gluing of four tetrahedra includes six steps, three of which are of Type \ri, two of which are of Type \rii\, and the remaining one is of Type \riii.

\section{Geometric interpretation of the state-integral}
\label{sec:conformal}
Spins are geometrically interpreted as lengths thus the geometrical meaning of the original Ponzano-Regge state-sum is rather simple as reviewed above. In contrast, the geometrical interpretation of spinors, or the inner product of spinors which are used in the new Ponzano-Regge state-integral described in Proposition \ref{prop:state_integral_2}, is not as apparent. To understand what the new model describes about the geometry, we analyze in this section separately the geometrical information of the tetrahedron described by the new vertex amplitude, and that of the gluing process described by the new edge amplitude. This will also unravel the geometry described by the boundary state $\psi^{\sl}_\Gamma$ in the dependence of the total amplitude and hint at the correspondence boundary condition in the classical counterpart.

\subsection{Poles of the  $\{12z^{\times 2}\}$ symbol}
\label{sec:geo_SGF}

In this subsection, we come back to the tetrahedron graph and look into the geometrical interpretation of the SGF. To this end, we rewrite the SGF in terms of the new variables, namely the {\it angle couplings} and the {\it link couplings}. They are both invariant under the $\SL(2,\bC)$ gauge transformation of spinors and the latter is further invariant under the anti-scale transformation (see \eqref{eq:symm_SGF}) thus forming a minimum set of variables. 

At first glance, since the SGF encodes the summation of spins, which geometrical represent lengths of the one-skeleton of a tetrahedron, the length information is expected to be washed out and it leaves the rest of the geometrical information, thus angles. Angle information of a tetrahedron includes internal angles of triangles (called internal angles for short) on the boundary of the tetrahedron, and the dihedral angles between each pair of triangles. 
We will find,  by taking the large $j$ limit of the $\{6j\}$-symbols, that these angle information are exactly stored in the norms and phases of the 6 independent link couplings. 

On the other hand, while the angle couplings contain redundant degrees of freedom in the SGF, it will become important when we describe the geometrical interpretation of the edge amplitude \eqref{eq:edge_amplitude_new} in the newly constructed spinfoam amplitude. Therefore, it is necessary to explore what the angle couplings represent geometrically as well. This can be realized from the relation of the angle couplings and the link couplings quantified below. In other words, by solving the angle couplings from the link couplings, one can directly translate the geometrical interpretation of the link couplings to that of the angle couplings. The obstacle is that we have 12 angle couplings at hand but only 6 link couplings, which means we need to choose a ``gauge'' to determine the solution. This is at the same time a benefit since we are free to choose a gauge such that it is the most geometrically reasonable then the angle coupling can be interpreted ``nicely''. 

Before we dig into details, let us first summarize what we will analyze in this subsection. The expression of the SGF \eqref{eq:evaluate_SGF} is purely in terms of the inner products of spinors, which we will call the angle couplings. We first rewrite the SGF into an expression \eqref{eq:SGF_angle_edge_coupling} in terms of the link couplings \eqref{eq:angle_to_edge}, then do a stationary analysis at large spins to find the geometrical interpretation of the link couplings which is illustrated in \eqref{eq:saddle_edge_norm} and \eqref{eq:saddle_edge_phase} \cite{Bonzom:2015ova,Bonzom:2019dpg}. Link couplings encode the conformal geometry of a tetrahedron both in their norms and phases. This geometrical interpretation can be translated back to the angle couplings. However, the map from the angle couplings to link couplings is only surjective but not bijective. We choose a ``geometrical gauge'' \eqref{eq:angle_couple_geometry}, which has the most local sense, to fix the angle coupling definition in terms of the link couplings. The geometrical interpretation of the angle couplings is not relevant to the later analysis in this subsection but will be important for Subsection \ref{sec:glue} where we investigate the geometrical meaning of the edge amplitudes. 
We plug the critical points of the link couplings into the SGF \eqref{eq:SGF_angle_edge_coupling} and write the semi-classical version of the SGF \eqref{eq:laplace_of_SGF} which can be expressed as a Laplace transform of (the exponential of) the Regge action of a tetrahedron. We also give a formal expression of the classical correspondence \eqref{eq:classical_SGF} of the SGF. finally, we conclude this subsection in the diagram \eqref{diag:6j_SGF} with the relations between the $\{6j\}$-symbol and the SGF as well as their classical and semi-classical correspondences. 

\medskip
\noindent {\bf From angle couplings to link couplings. }
\medskip
 
Each pair of links incident to the same node form an angle. We define the angle coupling $X_{ll'}$ by the (holomorphic) inner product of the spinors associated to the pair of links $(l,l')$ with the order $l\prec l'$ as 
\be
X_{ll'} \equiv X_{l'l}:=\left[z_l|z_{l'}\right> = |X_{ll'}|e^{i \Phi_{ll'}}\,,\quad
l\prec l'\,,\quad l,l'\in n\,,\quad \Phi_{ll'}\in [0,\pi) \,.
\ee
We have separated the norm $|X_{ll'}|$ and the phase $\Phi_{ll'}$ of the angle coupling.
These angle couplings can be grouped to form the link couplings $Y_l$'s \cite{Bonzom:2015ova} such that
\be
\prod_{n\in \Gamma}\, \prod_{l,l',l'' \in v} X_{ll'}^{J_n-2j_{l''}}
=\prod_{l\in \Gamma}Y_{l}^{2j_l}\,.
\label{eq:angle_equal_link}
\ee
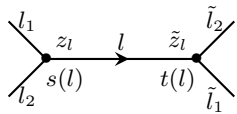
\begin{figure}[h!]
	\begin{tikzpicture}[scale=1]
	\coordinate (X) at (-6,0.7);
	\coordinate (gl) at (-5.5,0.2);
	\coordinate (gr) at (-3.5,0.2);
	\coordinate (Xt) at (-3,0.7);
	\coordinate (Y) at (-6,-0.3);
	\coordinate (Yt) at (-3,-0.3);	
	
	\draw[thick,decoration={markings,mark=at position 0.55 with {\arrow[scale=1.3,>=stealth]{>}}},postaction={decorate}] (gl) --node[midway,above]{$l$} node[very near start,above]{$z_l$} node[very near end,above]{$\tilde{z}_l$} node[very near start,below]{$s(l)$} node[very near end,below]{$t(l)$} (gr);
	\draw[thick](gl)--node[midway,above]{$l_1$}(X);
	\draw[thick](gl)--node[midway,below]{$l_2$}(Y);
	\draw[thick](gr)--node[midway,above]{$\tilde{l}_2$}(Xt);
	\draw[thick](gr)--node[midway,below]{$\tilde{l}_1$}(Yt);	
	
	\draw (gl) node{$\bullet$};
	\draw (gr) node{$\bullet$};
	
	\end{tikzpicture}
\caption{Two three-valent nodes $s(l)$ and $t(l)$ connected by an oriented link $l$. links $l,l_1,l_2$ are incident to $s(l)$, and links $l,\tilde{l}_1,\tilde{l}_2$ are incident to $t(l)$.}
\label{fig:Y_l}
\end{figure}
Consider two three-valent nodes connected with an oriented link $l$, where $l_1, l_2$ are the other two links incident to the source $s(l)$ of the link $l$, and $\tilde{l}_1,\tilde{l}_2$ are the two other links incident to the target $t(l)$, as shown in fig.\ref{fig:Y_l}. 
The link coupling $Y_l$ is expressed in terms of the angle couplings as
\be
Y_l
\equiv |Y_l|e^{i\Psi_l}
=\sqrt{
\frac{[z_l|z_{l_1}\rangle [z_{l_2}|z_l\rangle}{[z_{l_1}|z_{l_2}\rangle}
\frac{[\tz_l|z_{\tilde{l}_1}\rangle [z_{\tilde{l}_2}|\tz_l\rangle}{[z_{\tilde{l}_1}|z_{\tilde{l}_2}\rangle}}
=\sqrt{
\frac{X_{ll_1} X_{ll_2}}{X_{l_1l_2}}
\frac{X_{l\tilde{l}_1} X_{l\tilde{l}_2}}{X_{\tilde{l}_1\tilde{l}_2}}
}
\,,\quad
\Psi_l \in [0,\pi)
\,.
\label{eq:angle_to_edge}
\ee
The norm $|Y_l|$ and the phase $\Psi_l$ of the link coupling read explicitly
\be
|Y_l|=\sqrt{
\frac{|X_{ll_1}||X_{ll_2}|}{|X_{l_1l_2}|}
\frac{|X_{l\tilde{l}_1}||X_{l\tilde{l}_2}|}{|X_{\tilde{l}_1\tilde{l}_2}|}
}\,,\quad
\Psi_l=
\mod\left( \f12 (\Phi_{ll_1} + \Phi_{ll_2} - \Phi_{l_1l_2} 
+ \Phi_{l\tilde{l}_1} + \Phi_{l\tilde{l}_2} - \Phi_{\tilde{l}_1 \tilde{l}_2})
, \pi\right)
\,.
\label{eq:norm_and_phase}
\ee

It will be convenient to introduce the ``shared'' spins for angles on the same node as
\be
k_{ll'}:=J_n-2j_{l''}\,,\quad
k_{e'l''}:=J_n-2j_{l}\,,\quad
k_{l''l}:=J_n-2j_{l'}\,,\quad
\text{where }J_n=j_l+j_{l'}+j_{l''}\,
\text{ and } l,l',l''\in n \,.
\label{eq:j_to_k}
\ee
$k_{ll'}$ can be understood as the number of threads (equal to twice the spin value) shared by the links $l$ and $l'$, as illustrated in fig.\ref{fig:thread}.
\begin{figure}
	\centering
\begin{tikzpicture}
\coordinate (A) at (0,0);
\coordinate (B) at (4*4/4,0);
\coordinate (C) at (2*4/4,3.464*4/4);

\coordinate (c) at ($ (A)!.5!(B) $);
\coordinate (b) at ($ (A)!.5!(C) $);
\coordinate (a) at ($ (B)!.5!(C) $);

\path[name path = Aa] (A) -- (a);
\path[name path = Bb] (B) -- (b);
\path[name path = Cc] (C) -- (c);
\path [name intersections = {of = Aa and Bb,by=O}];

\coordinate (A1) at ([shift=(120:0.2)]A);
\coordinate (A2) at ([shift=(120:0.1)]A);
\coordinate (A3) at ([shift=(-60:0.1)]A);
\coordinate (A4) at ([shift=(-60:0.2)]A);

\coordinate (A11) at ([shift=(30:2)]A1);
\coordinate (A21) at ([shift=(30:2)]A2);
\coordinate (A31) at ([shift=(30:2)]A3);
\coordinate (A41) at ([shift=(30:2)]A4);
\coordinate (A01) at ([shift=(30:2)]A);

\coordinate (B1) at ([shift=(60:0.25)]B);
\coordinate (B2) at ([shift=(60:0.15)]B);
\coordinate (B3) at ([shift=(-120:0.05)]B);
\coordinate (B4) at ([shift=(-120:0.15)]B);
\coordinate (B0) at ([shift=(60:0.05)]B);

\coordinate (B11) at ([shift=(150:2)]B1);
\coordinate (B21) at ([shift=(150:2)]B2);
\coordinate (B31) at ([shift=(150:2)]B3);
\coordinate (B41) at ([shift=(150:2)]B4);
\coordinate (B01) at ([shift=(150:2)]B0);

\coordinate (C1) at ([shift=(0:0.3)]C);
\coordinate (C2) at ([shift=(0:0.2)]C);
\coordinate (C3) at ([shift=(0:0.1)]C);
\coordinate (C4) at ([shift=(-180:0.1)]C);
\coordinate (C5) at ([shift=(-180:0.2)]C);

\coordinate (C11) at ([shift=(-90:2)]C1);
\coordinate (C21) at ([shift=(-90:2)]C2);
\coordinate (C31) at ([shift=(-90:2)]C3);
\coordinate (C41) at ([shift=(-90:2)]C4);
\coordinate (C51) at ([shift=(-90:2)]C5);
\coordinate (C01) at ([shift=(-90:2)]C);

\draw[red,thick](A)-- (A01);
\draw[red,thick](A1)--node[midway,above]{$k_{ll''}$} (A11);
\draw[red,thick](A2)--(A21);
\draw[orange,thick](A3)--(A31);
\draw[orange,thick](A4)--(A41);

\draw[blue,thick](B0)--(B01);
\draw[blue,thick](B1)--(B11);
\draw[blue,thick](B2)--(B21);
\draw[orange,thick](B3)--(B31);
\draw[orange,thick](B4)--node[midway,left]{$k_{ll'}$} (B41);

\draw[red,thick](C)-- (C01);
\draw[blue,thick](C1)-- node[midway,right]{$k_{l'l''}$} (C11);
\draw[blue,thick](C2)--(C21);
\draw[blue,thick](C3)--(C31);
\draw[red,thick](C4)--(C41);
\draw[red,thick](C5)--(C51);

\def\x{0.8};
\def\y{0.5};
\coordinate (rA) at ([shift=(30:0.2)]A);
\coordinate (rAp) at ([shift=(-60:\x/2)]rA);
\coordinate (rA) at ([shift=(120:\x/2)]rA);
\coordinate (rAp) at ([shift=(30:\y)]rAp);

\coordinate (rB) at ([shift=(120:0.2)]B);
\coordinate (rBp) at ([shift=(60:\x/2)]rB);
\coordinate (rB) at ([shift=(-120:\x/2)]rB);
\coordinate (rBp) at ([shift=(150:\y)]rBp);

\coordinate (rC) at ([shift=(-90:0.2)]C);
\coordinate (rCp) at ([shift=(0:\x/2)]rC);
\coordinate (rC) at ([shift=(180:\x/2)]rC);
\coordinate (rCp) at ([shift=(-90:\y)]rCp);

\coordinate (t1) at ([shift=(180:0.6)]C01);
\coordinate (t2) at ([shift=(0:0.6)]C01);
\coordinate (t3) at ([shift=(-90:1)]C01);

\draw[fill=white,rotate around={-60:(rA)}] (rA) rectangle (rAp) node[pos=.5,rotate=-60] {$2j_l$};
\draw[fill=white,rotate around={-30:(rB)}] (rB) rectangle (rBp) node[pos=.5,rotate=60] {$2j_{l'}$};
\draw[fill=white] (rC) rectangle (rCp) node[pos=.5] {$2j_{l''}$};
\draw[fill=white] (t1) -- (t2) -- (t3) --cycle;

\draw (O) node{$2J_n$};

\end{tikzpicture}
\caption{Thicken the links $l,l',l''$ incident to the node $n$ into threads. Each thread carries a spin $\f12$ and the number of threads for each link is twice the spin value it is dressed with. The shared spins $k_{ll'}$ can be viewed as the number of threads shared by the links $l$ and $l'$. Similarly for $k_{ll''}$ and $k_{l'l''}$. }
\label{fig:thread}
\end{figure}
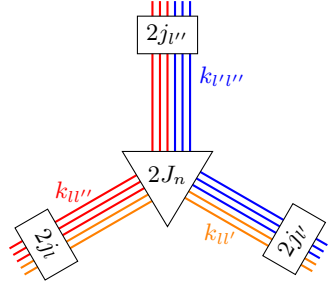
These shared spins are not independent. 
Referring to the relative position of the links $l,l_1,l_2,\tilde{l}_1,\tilde{l}_2$ shown in fig.\ref{fig:Y_l}, 
the constraint for $k_{ll'}$'s is
\be
k_{ll_1}+k_{ll_2}=k_{l\tilde{l}_1}+k_{l\tilde{l}_2}= 2j_l\,,\quad \forall l\,.
\label{eq:k_to_j}
\ee
Therefore, we have six constraints for 12 shared spins $k_{ll'}$'s, which results in 6 independent shared spins as expected. 

Let us recall the SGF \eqref{eq:SGF_1} written with the angle and link couplings,
\be
\cS(\{X_{ll'}\})
=\sum_{j_1\cdots j_6} \left[ \prod_{n=1}^4\sqrt{\frac{(J_n+1)!}{\prod_{l,l'\in n}k_{ll'}!}} \,\right]
\left\{ 
\ba{ccc}
j_1 & j_2 & j_3 \\
 j_4 & j_5 & j_6
\ea\right\}
\prod_{n=1}^4\left( \prod_{l,l'\in n} X_{ll'}^{k_{ll'}} \right)\,,
\ee
or
\be
\cS(\{Y_l\})
=\sum_{j_1\cdots j_6} \left[ \prod_{n=1}^4\sqrt{\frac{(J_n+1)!}{\prod_{l,l'\in n}k_{ll'}!}} \,\right]
\left\{ 
\ba{ccc}
j_1 & j_2 & j_3 \\
 j_4 & j_5 & j_6
\ea\right\}
\prod_{l=1}^6 Y_l^{2j_l}\,.
\label{eq:SGF_y}
\ee
For any loop $\cL$, the following equality holds
\be
\prod_{n\in \cL} X_{ll'}^n = \prod_{l\in \cL} Y_l\,,
\ee
thus \eqref{eq:evaluate_SGF} can also be written in two ways
\be
\cS(\{z_l,\zt_l\})=G(\{z_l,\zt_l\})^{-2}\,,\quad
G(\{z_l,\zt_l\})
=1+\sum_{\cL}
\prod_{l,l' \in n \subset \cL} X_{ll'}
=1+\sum_{\cL}
\prod_{l \subset \cL}Y_l\,.
\label{eq:SGF_angle_edge_coupling}
\ee

\medskip
\noindent {\bf Link couplings at the stationary point.}
\medskip

Now that we have the expression of the SGF in terms of the link couplings, we would like to apply the stationary analysis at the large $j$ limit of the SGF to look into the poles. We first take the Stirling approximation of the factorials:
\be
M!\sim \sqrt{2\pi M}\left( \frac{M}{e} \right)^M = e^{M\ln M +O(M)}\,.
\ee
When $V^2>0$ \footnotemark{}, the $\{6j\}$-symbol represents a tetrahedron embedded in the 3D Euclidean space. It reads \cite{Ponzano:1968se}
\be
\left\{ 
\ba{ccc}
j_1 & j_2 & j_3 \\
 j_4 & j_5 & j_6
\ea\right\}
\sim
\frac{1}{\sqrt{12\pi V}}\cos \left( \sum_{l=1}^6 \ell_l \Theta_l + \frac{\pi}{4} \right) \,,
\ee 
\footnotetext{
We show here that it is safe to disregard the $V^2<0$ contribution to the SGF at large spins. 
When $V^2<0$, the asymptotic expression of the $\{6j\}$-symbol corresponds to a tetrahedron embedded in the 3D Minkowski space. 
This can be seen from the geometrical expression of the volume $V$
\be
V^2=\frac{4 A_{s(e)}^2A_{t(e)}^2}{9 \,\ell_{l}^2}\left( 1-(\vec{\hat{n}}_{s(e)}\cdot \vec{\hat{n}}_{t(e)})^2 \right)\,,
\nn\ee
where $A_{s(e)}$ ({\it resp.} $A_{t(e)}$) is the area of the triangle opposite to the source ({\it resp.} target) vertex of the edge $e$, and $\vec{\hat{n}}_{s(e)}$ ({\it resp.} $\vec{\hat{n}}_{t(e)}$) the normal to this triangle. (We have assumed the edges are oriented, which can be arbitrarily decided.) 
When $V^2<0$, any, hence all, pair of normals satisfy $|\vec{\hat{n}}_{s(e)}\cdot \vec{\hat{n}}_{t(e)}|>1$. This is possible in the Minkowski space, where the angle formed by the normal to two faces, say $a$ and $b$, is calculated by
\be
\Theta_{ab}=\cosh^{-1}(\vec{\hat{n}}_a\cdot \vec{\hat{n}}_b)\,.
\nn\ee
This is called the internal Lorentzian angle. (There is a similar formula for exterior Lorentzian angle, see \cite{Barrett:1993db}). 
The $\{6j\}$-symbol in this case reads \cite{Ponzano:1968se,Barrett:1993db}
\be
\left\{ 
\ba{ccc}
j_1 & j_2 & j_3 \\
 j_4 & j_5 & j_6
\ea\right\}
\sim
\frac{1}{2\sqrt{12\pi |V|}}\cos \phi \exp \left( -\left| \sum_{l=1}^6 \ell_{l} \I \,\Theta_{l} \right| \right)\,,
\label{eq:6J_n_less0}
\nn\ee
where $\phi=\sum_{{l}=1}^6(\ell_{l}-\f12)\Re \,\Theta_{l}$, $\Theta_{l}$ the external dihedral angle about the edge $e$ defined by the Euclidean formula
\be
\cos \Theta_{l} =\frac{\vec{\hat{n}}_{s(e)}\cdot \vec{\hat{n}}_{t(e)}}{\vec{\hat{n}}_{s(e)}^2}=-\vec{\hat{n}}_{s(e)}\cdot \vec{\hat{n}}_{t(e)}\,.
\nn\ee
It follows that $|\cos \Theta_{l}|>1$, which implies that $\Theta_{l}$ is complex.
The smallest spin example of the $\{6j\}$-symbol giving a negative $V^2$ is $\left\{ 
\ba{ccc}
1/2 & 1/2 & 1 \\
1/2 & 1/2 & 1
\ea\right\}=\f16$ with $V^2=-\frac{9}{512}$.  
From the asymptotic expression for $V^2<0$ given above, the $\{6j\}$-symbol is exponentially suppressed even if the triangle inequality is satisfied. Therefore $V^2<0$ will not control the behaviour of the SGF.
}
where $\ell_l= j_l+\f12$ is the edge length of the edge $e$ (with abusive notation), $V$ the volume of the tetrahedron with edge lengths $\{\ell_l\}$ calculated by the Cayley-Menger determinant,
and $\Theta_l$ the external dihedral angle about the edge $e$, $\ie$ the angle between the outward normals to the faces sharing the edge $e$ \footnotemark{}. 
\footnotetext{
Note that the edges of the tetrahedron $T$ are denoted as $e$'s and the links of the tetrahedron graph $\Gamma=(\partial T)^*_1$ dual to the boundary of $T$ are denoted as $l$'s. It is on the links of $\Gamma$ where we associate spin labels $\{j_l\}$, but they represent the lengths of the edges $e$'s on $T$ which are dual to $l$'s. We denote the length $\ell_l$ and dihedral angle $\Theta_l$ with subscript $l$ instead of $e$ to avoid a mixture of notation in the same equation as much as possible in the main text. 
}
It is computed by the edge lengths as
\be
\sin \Theta_l = \frac{3}{2}\frac{V \ell_l}{S_{s(e)}S_{t(e)}}\,,
\ee
where $S_{v}$ is the area of the triangle opposite to the vertex $v$.
The large $j$ limit of the SGF \eqref{eq:SGF_y} is thus \cite{Bonzom:2019dpg}
\be\begin{split}
\cS^{\sl}(\{Y_l\})
&=\sum_{j_1\cdots j_6} \left[ \prod_{n=1}^4\sqrt{\frac{(J_n+1)!}{\prod_{l,l'\in n}k_{ll'}!}} \,\right]
\left\{ 
\ba{ccc}
j_1 & j_2 & j_3 \\
 j_4 & j_5 & j_6
\ea\right\}
\prod_{l=1}^6 Y_l^{2j_l}\\
&\sim \sum_{\{j_l\}}\,
e^{\sum_{n=1}^4 \f12\left(J_n \ln J_n - \sum_{l,l'\in n}k_{ll'}\ln k_{ll'}\right)} e^{ \sum_{l=1}^6  2j_l \left( \ln |Y_l|+i \Psi_{l} \right)}
\frac{1}{2\sqrt{12\pi V}}\sum_{\epsilon=\pm }
e^{i\epsilon \left( \sum_{l=1}^6 \ell_l\Theta_l + \frac{\pi}{4} \right)}
\\
&=\sum_{\epsilon=\pm} \sum_{\{j_l\}} \frac{1}{2\sqrt{12\pi V}} 
e^{\cS_\epsilon(\{Y_l,j_l\})}\,.
\end{split}
\label{eq:large_j_limit_SGF}
\ee
In the second line, we have used $\cos \left( \sum_{l=1}^6 \ell_l \Theta_l + \frac{\pi}{4} \right)= \f12 \sum_{\epsilon=\pm} e^{i\epsilon \left( \sum_{l=1}^6 \ell_l\Theta_l + \frac{\pi}{4} \right)}$. 

As the volume $V$ grows polynomially with the spins, its derivative of spin will contribute to the sub-leading correction of the stationary point. Therefore, to the leading order, one simply needs to consider the stationary point of the exponent term  $\cS_\epsilon(\{Y_l,j_l\})$ of the SGF. 

The real and imaginary part of $\cS_\epsilon$ can be rewritten as 
\begin{align}
\Re[\cS_\epsilon(\{Y_l,j_l\})]
&=\sum_{n=1}^4\f12 \left[ J_n\ln J_n -\sum_{l,l'\in n}k_{ll'}\ln k_{ll'}\right] + \sum_{l=1}^6 j_l \ln |Y_l|^2\,,\label{eq:real_sgf}\\
\I[\cS_\epsilon(\{Y_l,j_l\})]
&=\sum_{l=1}^6 \left[ j_l (2\Psi_l +\epsilon \Theta_l )+\f12\epsilon \Theta_l \right] +\epsilon \frac{\pi}{4}\,.\label{eq:imaginary_sgf}
\end{align}

Thanks to the Schläfli identity, $\sum_{e}^6 j_l \frac{\partial \Theta_l}{\partial j_l}=0$, the phase term has a simple derivative expression $\frac{\partial \sum_{l}j_{l}\Theta_l}{\partial j_l}= \Theta_l$. The saddle point $\frac{\partial \cS_\epsilon}{\partial j_l }$ can be separated into the real part and the imaginary part.
Using some trigonometry relations, the result reads (neglecting sub-leading contributions)\cite{Bonzom:2015ova,Bonzom:2019dpg}
\begin{align}
\frac{\partial \Re [\cS_\epsilon ]}{\partial j_e}=0 \,\,
&\rightarrow\,\,
|Y_l|^2\simeq \sqrt{
\frac{k_{ll_1}k_{ll_2}}{k_{l_1l_2}J_{s(l)}}
\frac{k_{l\tilde{l}_1}k_{l\tilde{l}_2}}{k_{\tilde{l}_1\tilde{l}_2}J_{t(l)}}
}
\equiv \tan \frac{\phi_{s(l)}}{2} \tan \frac{\phi_{t(l)}}{2}\,,
\label{eq:saddle_edge_norm} \\ 
\frac{\partial \I [\cS_\epsilon ]}{\partial j_l}=0 \,\,
&\rightarrow\,\,
\Psi_l = -\frac{\epsilon}{2}\Theta_l\,.	
\label{eq:saddle_edge_phase}
\end{align}
we have identified the length of edge $e$ with the spin values $j_l$ as $j_l\gg \f12$. 
$\phi_{s(l)}$ is the internal angle opposite to the edge $e$ in the triangle dual to the source node $s(l)$ of the link $l$, likewise for $\phi_{t(l)}$, as shown in fig.\ref{fig:edge_for_two_triangle}.
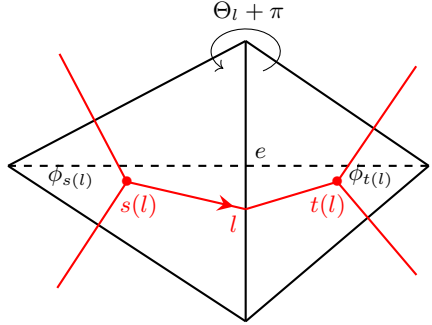
\begin{figure}[h!]
\centering 
\begin{tikzpicture}[scale=2.3, one end extended/.style={shorten >=-#1},one end extended/.default=1cm,]
\coordinate (O1) at (0,0,0);

\coordinate (A1) at (0,0.561,0);
\coordinate (B1) at (0,-1.061,0);
\coordinate (C1) at (-1.566,-0.354,-0.5);
\coordinate (D1) at (0.866,-0.354,-0.5);

\coordinate (aa) at (-1.566/2,-0.35,-0.25);
\coordinate (bb) at (0.866/2,-0.35,-0.25);

\coordinate (AC) at ($(C1)!(aa)!(A1)$);
\coordinate (BC) at ($(C1)!(aa)!(B1)$);
\coordinate (AD) at ($(D1)!(bb)!(A1)$);
\coordinate (BD) at ($(D1)!(bb)!(B1)$);

\coordinate (O) at ($(A1)!0.6!(B1)$);

\draw[thick] (A1) -- (B1) node[pos=0.4,right]{$e$};
\draw[thick] (A1) -- (C1);
\draw[thick] (A1) -- (D1);
\draw[thick] (B1) -- (C1);
\draw[dashed,thick] (C1) -- (D1);
\draw[thick] (D1) -- (B1);

\draw[red,thick,postaction={decorate},decoration={markings,mark={at position 0.9 with {\arrow[scale=1.5,>=stealth]{>}}}}] (aa) -- (O) node[pos=0.9,below]{$l$} node[pos=0.1,below]{$s(l)$};
\draw[red,thick] (O) -- (bb) node[pos=0.9,below]{$t(l)$};
\draw [red,thick,one end extended] (aa) -- (AC);
\draw [red,thick,one end extended] (aa) -- (BC);
\draw [red,thick,one end extended] (bb) -- (AD);
\draw [red,thick,one end extended] (bb) -- (BD);

\draw[red] (aa) node{$\bullet$};
\draw[red] (bb) node{$\bullet$};

    \draw ([shift=(-10:0.37)]C1) node{$\phi_{s(l)}$};

    \draw ([shift=(190:0.34)]D1) node{$\phi_{t(l)}$};
\draw[->] ([shift=(-60:0.2cm)]A1) arc (-60:230:0.2cm and 0.13cm) node[midway,above]{$\Theta_l+\pi$};
\end{tikzpicture}
\caption{Two triangles (non-planar) sharing one edge $e$ and its 2D oriented dual graph ({\it in red}). The node $s(l)$ is dual to the left triangle and $t(l)$ dual to the right one. $\phi_{s(l)}$ and $\phi_{t(l)}$ are the internal angle opposite to the edge $e$ within the left and right triangle respectively. Their relation with the norm of the link coupling at the saddle point is given in \eqref{eq:saddle_edge_norm}. $\Theta_l$ is the external dihedral angle about the edge $e$. Its relation with the phase of the link coupling at the saddle point is given in \eqref{eq:saddle_edge_phase}.}
\label{fig:edge_for_two_triangle}
\end{figure}
The saddle point corresponds to the pole of the SGF, which is also the Fisher zero for the Ising partition function on a tetrahedron graph \cite{Bonzom:2019dpg}. It clearly expresses the conformal geometry of the (classical limit of) the tetrahedron: The norm of the link coupling $|Y_l|$ corresponds to the pair of internal angles $\left(\phi_{s(l)},\phi_{t(l)}\right)$ opposite to the edge $e$ in the two triangles sharing $e$, while the phase $\Psi_l$ corresponds to half of the external dihedral angle about the edge $e$. 

\medskip
\noindent {\bf Back to angle couplings. }
\medskip

Suppose we have a solution to the link couplings $\{Y_l\}$, we aim at solving the angle couplings $\{X_{ll'}\}$ from $\{Y_l\}$. 
As there are 12 angle couplings $X_{ll'}$'s, hence 24 real variables, but only 6 link couplings $Y_l$'s, hence 12 real variables, there is a family of solutions for $\{X_{ll'}\}$ in terms of $\{Y_l\}$, and each solution can be parametrized by 12 real parameters or 6 complex parameters. 

A natural assumption is that each link coupling gains equal contribution from the source and the target of the link. Referring to \eqref{eq:angle_equal_link}, it means
\be
\frac{X_{ll_1} X_{l_2l}}{X_{l_1l_2}}
=
\frac{X_{l\tilde{l}_1} X_{\tilde{l}_2l}}{X_{\tilde{l}_1\tilde{l}_2}}
=Y_l\,.
\label{eq:symm_j_to_k}
\ee
One then gets a symmetric solution to the angle coupling
\be
X_{ll'} = \sqrt{ Y_l Y_{l'} }\,,
\label{eq:sol_zz_symm}
\ee
which indeed gives back \eqref{eq:angle_to_edge}. 
This result can be viewed as obtained by splitting the power $2j_l$ of each link coupling $Y_l$ in \eqref{eq:angle_equal_link}, obeying the constraint \eqref{eq:k_to_j}, into
\be
2j_l=\frac{k_{ll_1}+k_{ll_2}}{2}+\frac{k_{l\tilde{l}_1}+k_{l\tilde{l}_2}}{2}\,,
\ee
and write $Y_l^{2j_l}=\sqrt{Y_l}^{k_{ll_1}}\sqrt{Y_l}^{k_{ll_2}} \sqrt{Y_l}^{k_{l\tilde{l}_1}}\sqrt{Y_l}^{k_{l\tilde{l}_2}}$. 
After splitting all the link couplings in this way, 
the angle coupling $X_{ll'}$ is a collection of the terms with the power $k_{ll'}$. 

With the same logic, one can obtain the most general solution by arbitrarily distributing the contributions of the power to the source and target of each link. That is to separate the power $2j_l$ as
\be
2j_l=\lambda_l^{s(l)} (k_{ll_1}+k_{ll_2})+\lambda_l^{t(l)}(k_{l\tilde{l}_1}+k_{l\tilde{l}_2})\quad\longrightarrow\quad
\left|\ba{l}
\frac{X_{ll_1} X_{l_2l}}{X_{l_1l_2}}
=Y_l^{2\lambda_l^{s(l)}}
\\[0.2cm]
\frac{X_{l\tilde{l}_1} X_{\tilde{l}_2l}}{X_{\tilde{l}_1\tilde{l}_2}}
=Y_l^{2\lambda_l^{t(l)}}
\ea\right.\,,\quad
\lambda_l^{s(l)},\lambda_l^{t(l)}\in[0,1]\quad \text{with } \lambda_l^{s(l)}+\lambda_l^{t(l)}=1\,.
\label{eq:gen_j_to_k}
\ee
We call this set of parameters $\{\lambda_l^n\}$, totally 12, ``scaling factors'' as they scale the contribution of link couplings from the two links forming the angle. 
The solution to the angle coupling $X_{ll'}$ parametrized with $\{\lambda_l^n\}$ is
\be
X_{ll'} =  Y_l^{\lambda_l^n}Y_{l'}^{\lambda_{l'}^n}\,,\quad
l,l'\in n\,.
\label{eq:sol_zz_gen}
\ee
Separating into the norm and the phase, we have
\be
|X_{ll'}|=|Y_l|^{\lambda_l^n}|Y_{l'}|^{\lambda_{l'}^n}\,,\quad
\Phi_{ll'}=\mod \left(\lambda_l^n \Psi_l + \lambda_{l'}^n \Psi_{e'}\,, \pi\right)\,.
\label{eq:sol_zz_gen_norm_phase}
\ee
Apart from writing the scaling factors into the real exponents, one can also write them into the multiplication coefficients, which are complex, as
\be
X_{ll'} = \frac{1}{\alpha_l \alpha_{l'}} Y_l Y_{l'}\,,\quad
\text{with } 
\alpha_l = Y_l^{1-\lambda_l^n}\,,\, \alpha_{l'} = Y_{l'}^{1-\lambda_{l'}^n}\,.
\ee
Note that, as we have fixed the range of the phase $\Phi_l\in [0,\pi)$, there is no branch ambiguity for $\{\alpha_l\}$. This choice of the phase range is related to the geometrical interpretation of the link couplings which we will describe below. 
The most symmetric solution \eqref{eq:sol_zz_symm} corresponds to $\lambda_l^{s(l)}=\lambda_l^{t(l)}=\f12$, or equivalently $\alpha_l = \sqrt{Y_l}, \alpha_{l'}=\sqrt{Y_{l'}}$. 

\medskip
\noindent {\bf Angle couplings with the ``geometric gauge'' at the stationary point.}
\medskip

Now we can transform from the link couplings $\{Y_l\}$ to the angle couplings $\{X_{ll'}\}$. Recall that there is no unique solution for $\{X_{ll'}\}$ from $\{Y_l\}$, but we can choose a family of parameters to fix $\{X_{ll'}\}$ so that it possesses a simple geometrical interpretation. 
Comparing the formula structure of \eqref{eq:norm_and_phase} and \eqref{eq:saddle_edge_norm}, 
we choose a {\it geometric gauge} (using the notations as in fig.\ref{fig:Y_l})
\be
\lambda_l^{s(l)}=\frac{\ln \frac{k_{ll_1}k_{ll_2}}{k_{l_1l_2}J_{s(l)}}}
{ \ln \frac{k_{ll_1}k_{ll_2}}{k_{l_1l_2}J_{s(l)}} 
+ \ln \frac{k_{l\tilde{l}_1}k_{l\tilde{l}_2}}{k_{\tilde{l}_1\tilde{l}_2}J_{t(l)}}}
\simeq \frac{\ln \left(\tan \frac{\phi_{s(l)}}{2}\right)}
{\ln\left( \tan \frac{\phi_{s(l)}}{2} \tan \frac{\phi_{t(l)}}{2}\right)}
\label{eq:lambda_for_closure}
\ee
and define
\be
|X_{ll'}|:=\sqrt{\frac{k_{ll'}}{J_n}}\simeq\sqrt{\tan \frac{\phi_{ll''}}{2} \tan \frac{\phi_{l'l''}}{2}} \,,\quad
l,l',l''\in n\,.
\label{eq:angle_couple_geometry}
\ee
$\phi_{ll''}$ and $\phi_{l'l''}$ are the remaining two internal angles in the triangle other than $\phi_{ll'}$ formed with $e$ and $e'$, as shown in fig.\ref{fig:triangle}.
Reversely, we can solve for these angles from the angle couplings,
\be
\phi_{ll'}\simeq 2 \tan^{-1} \frac{|X_{l'l''}||X_{ll''}|}{|X_{ll'}|}\,,\quad
\phi_{ll'}\in [0,\pi)
\,.
\ee
This solution corresponds to the scale factors
\be
\lambda_l^{s(l)}=\frac{\ln \frac{k_{ll_1}k_{ll_2}}{k_{l_1l_2}J_{s(l)}}}
{ \ln \frac{k_{ll_1}k_{ll_2}}{k_{l_1l_2}J_{s(l)}} 
+ \ln \frac{k_{l\tilde{l}_1}k_{l\tilde{l}_2}}{k_{\tilde{l}_1\tilde{l}_2}J_{t(l)}}}
\simeq \frac{\ln \left(\tan \frac{\phi_{s(l)}}{2}\right)}
{\ln\left( \tan \frac{\phi_{s(l)}}{2} \tan \frac{\phi_{t(l)}}{2}\right)}\,.
\label{eq:lambda_for_closure}
\ee
It indeed satisfies $\lambda_l^{s(l)}+\lambda_l^{t(l)}=1$.
\begin{figure}[h!]
\centering
\begin{tikzpicture}[scale=0.8,one end extended/.style={shorten >=-#1},one end extended/.default=1cm]
\coordinate (A) at (0,0);
\coordinate (B) at (5*4/4,0);
\coordinate (C) at (2*4/4,3.464*4/4);

\coordinate (c) at ($ (A)!.5!(B) $);
\coordinate (b) at ($ (A)!.5!(C) $);
\coordinate (a) at ($ (B)!.5!(C) $);

\path[name path = Aa] (A) -- (a);
\path[name path = Bb] (B) -- (b);
\path[name path = Cc] (C) -- (c);
\path [name intersections = {of = Aa and Bb,by=O}];

\coordinate (c) at ($ (A)!(O)!(B) $);
\coordinate (b) at ($ (A)!(O)!(C) $);
\coordinate (a) at ($ (B)!(O)!(C) $);

\draw[thick] (A) -- node[pos=0.6,below]{$e$}(B) -- node[pos=0.4,above right]{$e'$}(C) -- node[pos=0.6,left]{$e''$} cycle;
\draw[red,thick,one end extended] (O) -- (a) node[pos=1.9,left]{$l'$} ;
\draw[red,thick,one end extended] (O) -- (b) node[pos=1.8,above right]{$l''$} ;
\draw[red,thick,one end extended] (O) -- (c) node[pos=1.9,right]{$l$} ;

    \draw ([shift=(30:0.85)]A) node{$\phi_{ll''}$};

    \draw ([shift=(150:0.85)]B) node{$\phi_{ll'}$};

    \draw ([shift=(-80:0.65)]C) node{$\phi_{l'l''}$};
    
    \draw[red] (O) node{$\bullet$};

\end{tikzpicture}
\caption{Three links ({\it in red}) $l,l',l''$ incident to one node. The triangle ({\it in black}) dual to the node is bounded by edges $e, e', e''$. The length of the edge $e$ is given by the spin $j_l$, likewise for $e'$ and $e''$. The angle formed by the links $l,l'$ is $\phi_{ll'}$, likewise for $\phi_{l'l''}$ and $\phi_{ll''}$. Their relation with the norm of the angle coupling at the saddle point is given in \eqref{eq:angle_couple_geometry}.}
\label{fig:triangle}
\end{figure}
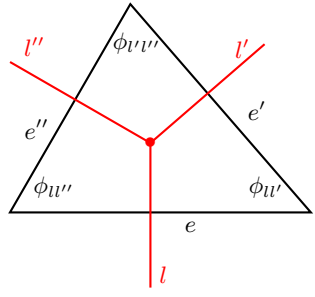
\begin{figure}
\centering
\begin{tikzpicture}[one end extended/.style={shorten >=-#1},one end extended/.default=1cm]
\coordinate (A) at (0,0,0);
\coordinate (B) at (5*4/4,0,0);
\coordinate (C) at (2*4/4,3.464*4/4,0);
\coordinate (D) at (5*4/5,3.464*4/5,-1);

\coordinate (c) at ($ (A)!.5!(B) $);
\coordinate (b) at ($ (A)!.5!(C) $);
\coordinate (a) at ($ (B)!.5!(C) $);

\path[name path = Aa] (A) -- (a);
\path[name path = Bb] (B) -- (b);
\path[name path = Cc] (C) -- (c);
\path [name intersections = {of = Aa and Bb,by=O}];

\coordinate (c) at ($ (A)!(O)!(B) $);
\coordinate (b) at ($ (A)!(O)!(C) $);
\coordinate (a) at ($ (B)!(O)!(C) $);

\coordinate (AB) at ($ (A)!.8!(B) $);
\coordinate (AC) at ($ (A)!.8!(C) $);
\coordinate (BC) at ($ (B)!.8!(C) $);

\draw[very thick] (A) -- node[pos=0.6,below]{$e_1$}(B) -- node[pos=0.4,above right]{$e_2$}(C) -- node[pos=0.6,left]{$e_3$} cycle;
\draw[gray] (C) -- (D);
\draw[gray] (B) -- (D);
\draw[gray,dashed] (A) -- (D);
\draw[red,thick,one end extended,decoration={markings,mark=at position 0.85 with {\arrow[scale=1.3,>=stealth]{>}}},postaction={decorate}] (O) -- (a) node[pos=1.6,below right] {$l_2$} node[pos=0.3,above]{$z_2$};
\draw[red,thick,one end extended,decoration={markings,mark=at position 0.85 with {\arrow[scale=1.3,>=stealth]{>}}},postaction={decorate}] (O) -- (b) node[pos=1.5,above right]{$l_3$} node[pos=0.3,below]{$z_3$};
\draw[red,thick,one end extended,decoration={markings,mark=at position 0.85 with {\arrow[scale=1.3,>=stealth]{>}}},postaction={decorate}] (O) -- (c) node[pos=1.9, right]{$l_1$} node[pos=0.3,right]{$z_1$};

\draw[->] ([shift=(-10:0.2cm)]AB) arc (-10:-250:0.2cm and 0.25cm) node[pos=0.4,below]{$\Theta_{l_1}+\pi$};
\draw[dashed] ([shift=(10:0.2cm)]AB) arc (10:30:0.2cm and 0.25cm);

\draw[->,rotate=20] ([shift=(-190:0.2cm)]BC) arc (-190:60:0.25cm and 0.2cm) node[pos=0.9,right]{$\Theta_{l_2}+\pi$};

\draw[->,rotate=-20] ([shift=(-270:0.2cm)]AC) arc (-270:-50:0.25cm and 0.2cm) node[pos=0.2,left]{$\Theta_{l_3}+\pi$};
\draw[dashed] ([shift=(-310:0.2cm)]AC) arc (-310:-335:0.25cm and 0.2cm);

\draw[red] (O) node{$\bullet$};

	\end{tikzpicture}
\caption{A triangle bounded by edges $e_1,e_2,e_3$ ({\it thick}) on the boundary of a tetrahedron. It is dual to a node with three outgoing links $l_1, l_2, l_3$ incident to it ({\it in red}). $l_1$ is associated with a spinor $|z_1\rangle$, $l_2$ with $|z_2\rangle$ and $l_3$ with $|z_3\rangle$ at their source. They form three angle couplings $[z_{l_1}|z_{l_2}\rangle=|X_{l_1l_2}|e^{i\Phi_{l_1l_2}}, [z_{l_2}|z_{l_3}\rangle=|X_{l_2l_3}|e^{i\Phi_{l_2l_3}}, [z_{l_3}|z_{l_1}\rangle=|X_{l_3l_1}|e^{i\Phi_{l_3l_1}}$. Each edge $e\in\{e_1,e_2,e_3\}$ is shared with another triangle in the tetrahedron, and they form an external dihedral angle $\Theta_l$. The relation between the phases of the angle couplings and the dihedral angles at the saddle point is given in \eqref{eq:angle_coupling_phase_geometry}.}
\label{fig:triangle_in_tetra}
\end{figure}
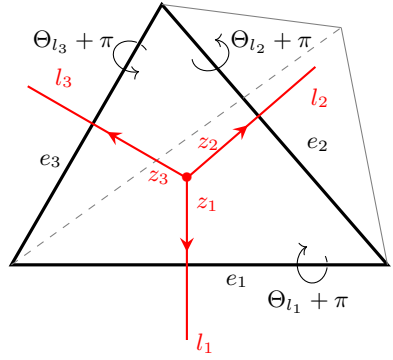

The norms of the angle couplings defined in \eqref{eq:angle_couple_geometry} satisfy the closure constraint, which corresponds to the fact that the three internal angles of the triangle dual to a node sum to $\pi$.
\be
 |X_{ll'}|^2+|X_{l'l''}|^2+|X_{ll''}|^2=1\,,\quad \Longleftrightarrow \quad 
\phi_{ll'}+\phi_{l'l''}+\phi_{ll''}=\pi\,,\quad
l,l',l''\in n\,.
\label{eq:angle_coupling_constraint}
\ee
Note that this solution is only valid for $|X_{ll'}|\leq 1$, which is indeed satisfied by the definition \eqref{eq:angle_couple_geometry} of $|X_{ll'}|$.

Given the scaling factors \eqref{eq:lambda_for_closure}, the solution to $\{\Phi_{ll'}\}$ are also set according to \eqref{eq:sol_zz_gen_norm_phase}, although they possess a more complicated expression hence more involved geometrical interpretation. With no loss of generality, we choose the orientation of the links $l$ and $l'$ to be outgoing from the node $n$, then the phase is 
\be
\Phi_{ll'}
\simeq \mod\left(  \frac{1}{2}\left(
\lambda_l^n \Theta_l 
 + \lambda_{l'}^n \Theta_{e'} 
 \right),\pi\right) \,,\quad
 \text{ with }
 \lambda_l^n= \frac{\ln \left( \tan \frac{\phi_{s(l)}}{2} \right)}{\ln \left( \tan \frac{\phi_{s(l)}}{2}\tan \frac{\phi_{t(l)}}{2} \right)} \,,\quad
 \lambda_{l'}^n= \frac{\ln \left( \tan \frac{\phi_{s(e')}}{2} \right)}{\ln \left( \tan \frac{\phi_{s(e')}}{2}\tan \frac{\phi_{t(e')}}{2} \right)} \,.
 \label{eq:angle_coupling_phase_geometry}
\ee

In summary, the angle couplings encode the conformal geometry of the tetrahedron, similar to the link couplings but in a more ``mixed'' fashion determined by the scaling $\{\lambda_l^n\}$. The norms $\{|X_{ll'}|\}$ give the information about the internal angles between edges, while the phases $\{\Phi_{ll'}\}$ give the information about the dihedral angles between triangles. 

To determine the shape of a tetrahedron, one merely needs the norm of the couplings or the phase of the couplings. For sake of simplicity, we will make use of the norms only.

\medskip
\noindent {\bf Classical and semi-classical correspondences of the SGF.}
\medskip

Now that we know the relation between the SGF and the $\{6j\}$-symbols both rigorously and asymptotically, we can also look into the classical and semi-classical correspondences of the SGF and their relation with that of the $\{6j\}$-symbol, $\ie$ the GHY boundary action and the Hartle-Sorkin boundary action in Regge calculus \cite{Hartle:1981cf}, respectively.

We first re-arrange the large $j$ limit of the SGF in terms of the $\{6j\}$-symbols. For each link $l$, we again denote the three links incident to the source $s(l)$ as $l,l_1,l_2$ and those incident to the target $t(l)$ as $l,\tilde{l}_1,\tilde{l}_2$, which is illustrated in fig.\ref{fig:Y_l}.  Then \eqref{eq:large_j_limit_SGF} can be written as
\be\begin{split}
\cS(\{Y_l\})
\sim &
\sum_{\epsilon=\pm }\sum_{\{j_l\}}
 \frac{e^{i\epsilon  \left( \frac{\pi}{4} + \f12 \Theta_l\right) }}{2\sqrt{12\pi V}}\\
&\exp\left[
\sum_{l=1}^6 j_l \left( \ln \sqrt{\frac{J_{s(l)} \left( J_{s(l)}-2j_{l} \right)}
{\left( J_{s(l)}-2j_{l_1} \right)\left( J_{s(l)}-2j_{l_2} \right)}
\frac{J_{t(l)}\left( J_{t(l)}-2j_{l} \right)}
{\left( J_{t(l)}-2j_{\tilde{l}_1} \right)\left( J_{t(l)}-2j_{\tilde{l}_2} \right)}}
 + \ln|Y_l|^2 + i\left(\epsilon \Theta_l + 2\Psi_l \right)
 \right)\right]\\
\simeq & \sum_{\epsilon=\pm }\sum_{\{j_l\}} 
\frac{e^{i\epsilon \left( \frac{\pi}{4} + \f12 \Theta_l\right) }}{2\sqrt{12\pi V}}\,
\prod_{l=1}^6 \exp\left(i \epsilon j_l \Theta_l \right)
\exp\left(-j_l \ln \frac{\tan \frac{\phi_{s(l)}}{2} \tan \frac{\phi_{t(l)}}{2}}{Y_l^2}\right)\,,
\end{split}
\ee
where we have used \eqref{eq:saddle_edge_norm} and the notation in fig.\ref{fig:edge_for_two_triangle} for the internal angles in the third line. 
Now we can approximate the summation of spins by integration from the minimal admissible spin to infinity. The volume term $\frac{1}{\sqrt{V}}$ scales as $\frac{\lambda^{-\f32}}{\sqrt{V_{\min}}}$ with the scale factor $\lambda\in [1,\infty)$, which is sub-dominant compared to the exponential terms thus can be moved out from the integrand and approximated by $\frac{1}{\sqrt{V_{\min}}}$. Therefore, the SGF can be approximately written as
\be\begin{split}
\cS(\{Y_l\})
\sim & \sum_{\epsilon=\pm }
\frac{e^{i\epsilon \left( \frac{\pi}{4} + \f12 \Theta_l\right)}}{2\sqrt{12\pi V_{\min}}}\,
\int_{\min j_l\rightarrow 0}^\infty \left(\prod_{l=1}^6  \rd j_l\right) \,
\exp\left(i \epsilon \sum_{l=1}^6 j_l \Theta_l \right)
\exp\left(-\sum_{l=1}^6 j_l \ln \frac{\tan \frac{\phi_{s(l)}}{2} \tan \frac{\phi_{t(l)}}{2}}{Y_l^2}\right) \\
=& \sum_{\epsilon=\pm }
\frac{e^{i\epsilon \left( \frac{\pi}{4} + \f12 \Theta_l\right) }}{2\sqrt{12\pi V_{\min}}}\,
\prod_{l=1}^6 \frac{1}{\ln\left( \tan \frac{\phi_{s(l)}}{2} \tan \frac{\phi_{t(l)}}{2} \right)- \ln |Y_l|^2 
+ i\left( \epsilon \Theta_l + 2 \Psi_l\right)}\,.
\end{split}
\label{eq:laplace_of_SGF}
\ee  
In doing the integration, we have viewed $j_l$'s independent of the dihedral angles $\Theta_l$'s and the internal angles $\phi_n$'s by neglecting the sub-leading contribution. In other words, the integration is over the scale factor. To obtain the second line, we have also taken the low spin limit $j_l=0$ for the lower bound and assumed that the domain of the link couplings are constrained by $|Y_l|^2 \leq \tan \frac{\phi_{s(l)}}{2} \tan \frac{\phi_{t(l)}}{2}$ in order to obtain a finite result. Therefore, the large $j$ limit, or the semi-classical correspondence, of the SGF can be seen as a Laplace transform of (the exponential of) the Regge action of a tetrahedron. It is also clear from \eqref{eq:laplace_of_SGF} that the critical couplings \eqref{eq:saddle_edge_norm} and \eqref{eq:saddle_edge_phase} do sit on the pole of the the SGF. 

The same logic follows in the continuum theory of the boundary - the classical correspondence of the SGF is related to a Laplace transform of (the exponential of) the GHY boundary action. The GHY boundary term is changed by a scale factor $\lambda = \exp(\Omega_0/2)$ under rescaling. Then one can single out the scaling factor and rewrite the GHY term as
\be
-2 \int_{\partial \cM} \rd^2 x \,\sqrt{h} K 
= -2\lambda \int_{\partial \cM}\rd^2 x \,\hat{K}
=-2\lambda \int_{\partial \cM}\rd^2 x \,\hat{h}^{\mu\nu}\nabla_{\mu} \hat{n}_{\nu} \,,
\ee
where $\hat{h}_{\mu\nu}$ is the normalized metric with determinant one. (The differential $\rd^2 x$ remains the same since we have been assigning no dimension on $x$'s. See Subsection \ref{sec:3d_gravity}.)
Formally denote a complex function $f(y,c_{\partial})$ of $y$ symboling the coupling and $c_{\partial}$ symboling the conformal geometry of $\partial \cM$, which has a positive real part $\Re\,f>0$. Then the continuum, or the classical, correspondence of \eqref{eq:laplace_of_SGF} is formally expressed as
\be
\cS^{\sl}(y, c_{\partial})=\int_0^\infty \rd \lambda \,\exp(-i \lambda \int_{\partial \cM} \rd^2x \, 2\hat{K}) \, \exp(-\lambda f(y,c_{\partial}))
= \frac{1}{\Re\,f + i (\int_{\partial\cM}\rd^2 x \,2\hat{K} + \I \,f)}\,.
\label{eq:classical_SGF}
\ee

We conclude the relation of the usual GHY boundary term with the scale-invariant boundary action and their discrete version and quantum correspondence as follows ($C$ denotes $\frac{e^{i\left( \frac{\pi}{4} + \f12 \Theta_l\right) }}{2\sqrt{12\pi V_{\min}}}$ and we only consider $\epsilon=1 $ for simplicity). 
\be\ba{ccccc}
-2 \int_{\partial \cM}\rd^2 x \,\sqrt{h}K  &
\xrightleftharpoons[\text{continuum}]{\text{discrete}} &
\sum_{l=1}^6 \ell_l \Theta_l &
\xrightleftharpoons[\text{large }j\text{ limit}]{\text{quantum}} &
6j\text{-symbol} \\
\rotatebox[origin=c]{270}{$\xrightarrow{\text{scale invariant part}}$} &&
\rotatebox[origin=c]{270}{$\xrightarrow{\text{Laplace transform}}$} &&
\rotatebox[origin=c]{270}{$\xrightarrow{\text{generating function}}$} \\
-2 \int_{\partial \cM}\rd^2x\,\hat{K} &
\xrightleftharpoons[\text{continuum}]{\text{discrete, Laplace transform}} &
C \,\prod_{l=1}^6 \ln \frac{Y_l^2}{\tan \frac{\phi_{s(l)}}{2} \tan \frac{\phi_{t(l)}}{2} e^{i \Theta_l}} &
\xrightleftharpoons[\text{large }j\text{ limit}]{\text{quantum}} &
\text{SGF}
\ea
\label{diag:6j_SGF}
\ee

In this subsection, we have analyzed the geometrical interpretation of the new vertex amplitudes given by the SGF using the stationary point analysis. Only the angle information of a tetrahedron is stored in the SGF either expressed in terms of the link couplings or the angle couplings.

\subsection{Propagator and geometric gluing}
\label{sec:glue}

We have shown that the SGF, as a scale-invariant counterpart of the $\{6j\}$-symbol, describes the local conformal geometry of each elementary 3D block in the Ponzano-Regge state-integral. Since the scales of the building blocks are not fixed, the gluing process becomes non-trivial compared to the original Ponzano-Regge state-sum model, where the gluing is performed by matching the full boundary geometry - both the shape and size of the glued triangles - of the blocks. This new way of gluing is described by this new edge amplitude
\be
\cA_{e^*}
=\,{}_0F_3(;2,2,\f12; \f{\left(\sum_{l\in n}  \langle z_l^{s(e^*)}|w_l^{t(e^*)}]\right)^2}{4}  ) 
\nn\ee
 To gain a global picture of the state-integral model, we here analyze the geometrical meaning of the edge amplitude.

The edge amplitude is also called the propagator and it describes how two neighbouring tetrahedra are glued. Note that the measure of the spinors $\rd \mu(z)\equiv \rd^4 z\, e^{-\langle z|z \rangle}$ includes a Gaussian weight, it is necessary to take it into account so that the norm $|P(\{z_l\},\{w_l\})|$ of the propagator remains finite (we have assumed that $\{z_l\}$ are from the tetrahedron dual to $s(e^*)$ and $\{w_l\}$ are from that dual to $t(e^*)$ and omit the superscript for simplicity). We rewrite the propagator as
\be 
P(\{z_l\},\{w_l\})=\,{}_0F_3(;2,2,\f12;\f{\z^2}{4}) e^{-\r^2}\,,\quad
 \z= \sum_{l\in n}\langle z_l|w_l]\,,\quad
 \r^2:=\sum_{l\in n}
\left(\langle z_l|z_l\rangle 
+ \langle w_l|w_l\rangle\right)\,. 
\label{eq:propagator}
\ee 
Similar to the vertex amplitude, we are interested in the stationary points of the propagator \eqref{eq:propagator}. 
We consider $w_{l,l',l''}$ with $ l,l',l''\in n$ to be the configurations of the propagator and obtain the stationary points by the vanishing derivative of $\bar{w}_{l,l',l''}$. For instance, the derivatives of $\bar{w}^{0,1}_l$ give
\begin{align}
\frac{\partial P(\{z_l\},\{w_l\})}{\partial \bar{w}_l^0}
&= \zb_l^1\,\sum_{k=0}^{\infty}
\frac{2k}{(k+1)!^2(2k)!}\z^{2k-1} e^{-\r^2} - w_l^0\, {}_0F_3(;2,2,\f12;\f{\z^2}{4}) e^{-\r^2} \nn \\
&=\zb_l^1\,  {}_0F_3(;3,3,\f32;\f{\z^2}{4}) \,\z \,e^{-\r^2} 
- w_l^0\,{}_0F_3(;2,2,\f12;\f{\z^2}{4})  e^{-\r^2} 
=0\,,\nn\\
\frac{\partial P(\{z_l\},\{w_l\})}{\partial \bar{w}_l^1}
&= -\zb_l^0\,\sum_{k=0}^{\infty}
\frac{2k}{(k+1)!^2(2k)!}\z^{2k-1} e^{-\r^2} - w_l^1\, {}_0F_3(;2,2,\f12;\f{\z^2}{4})  e^{-\r^2} \nn \\
&=-\zb_l^0\, {}_0F_3(;3,3,\f32;\f{\z^2}{4})\,\z\, e^{-\r^2} - w_l^1\,{}_0F_3(;2,2,\f12;\f{\z^2}{4})  e^{-\r^2}
=0\,.\nn
\end{align}
These saddle point formulas are also valid for spinors $w_{l'},w_{l''}$ thus there are totally six such formulas, which can be summarized as follows. Denote $A:= {}_0F_3(;3,3,\f32;\f{\z^2}{4})\,\z\, e^{-\r^2}$ and $B:={}_0F_3(;2,2,\f12;\f{\z^2}{4})  e^{-\r^2}$, then
\be\left|\ba{llll}
A |z_l] &+& B |w_l\rangle & = 0\\[0.12cm]
A |z_{l'}]&+& B |w_{l'}\rangle & = 0\\[0.12cm]
A |z_{l''}]& +&B |w_{l''}\rangle & = 0
\ea\right.
\quad \longrightarrow\quad
\left|\ba{lll}
B^2 [w_l|w_{l'}\rangle &=& A^2 \langle z_l|z_{l'}]\\[0.12cm]
B^2 [w_{l'}|w_{l''}\rangle &=& A^2 \langle z_{l'}|z_{l''}]\\[0.12cm]
B^2 [w_{l''}|w_{l}\rangle &=& A^2 \langle z_{l''}|z_{l}]
\ea\right..
\label{eq:Az_Bw}
\ee
The corresponding angle couplings from different tetrahedra have the same ratio $A^2/B^2$. 
Recall the solution \eqref{eq:angle_couple_geometry} of the angle couplings in the geometric gauge 
\be
|X_{ll'}|=\sqrt{\frac{k_{ll'}}{J_n}}=\sqrt{\tan \frac{\phi_{ll''}}{2} \tan \frac{\phi_{l'l''}}{2}} \,,\quad
l,l',l''\in n\,.
\nn\ee
The norms of the angle couplings encode the internal angles of the triangle they sit at, and they satisfy the closure constraint $|X_{ll'}|^2+|X_{l'l''}|^2+|X_{ll''}|^2=1$, thus
\be
|B|^4=|B|^4 \left(|[w_l|w_{l'}\rangle|^2 +|[w_{l'}|w_{l''}\rangle|^2 + |[w_{l''}|w_{l}\rangle|^2\right)
=
| A|^4 \left(|\langle z_l|z_{l'}]|^2 +|\langle z_{l'}|z_{l''}]|^2 + |\langle z_{l''}|z_{l}]|^2\right)=| A|^4\,.
\ee
Therefore, $|{}_0F_3(;3,3,\f32;\f{\z^2}{4})\,\z|$ and $|{}_0F_3(;2,2,\f12;\f{\z^2}{4})|$ must be the same at the stationary points. The solution is not unique, as can be immediately seen from the plot in fig.\ref{fig:hypergeo}, where the norm square of the two generalized hypergeometric functions is shown. 
\begin{figure}[h!]
	\centering
	\includegraphics[width=0.8\textwidth]{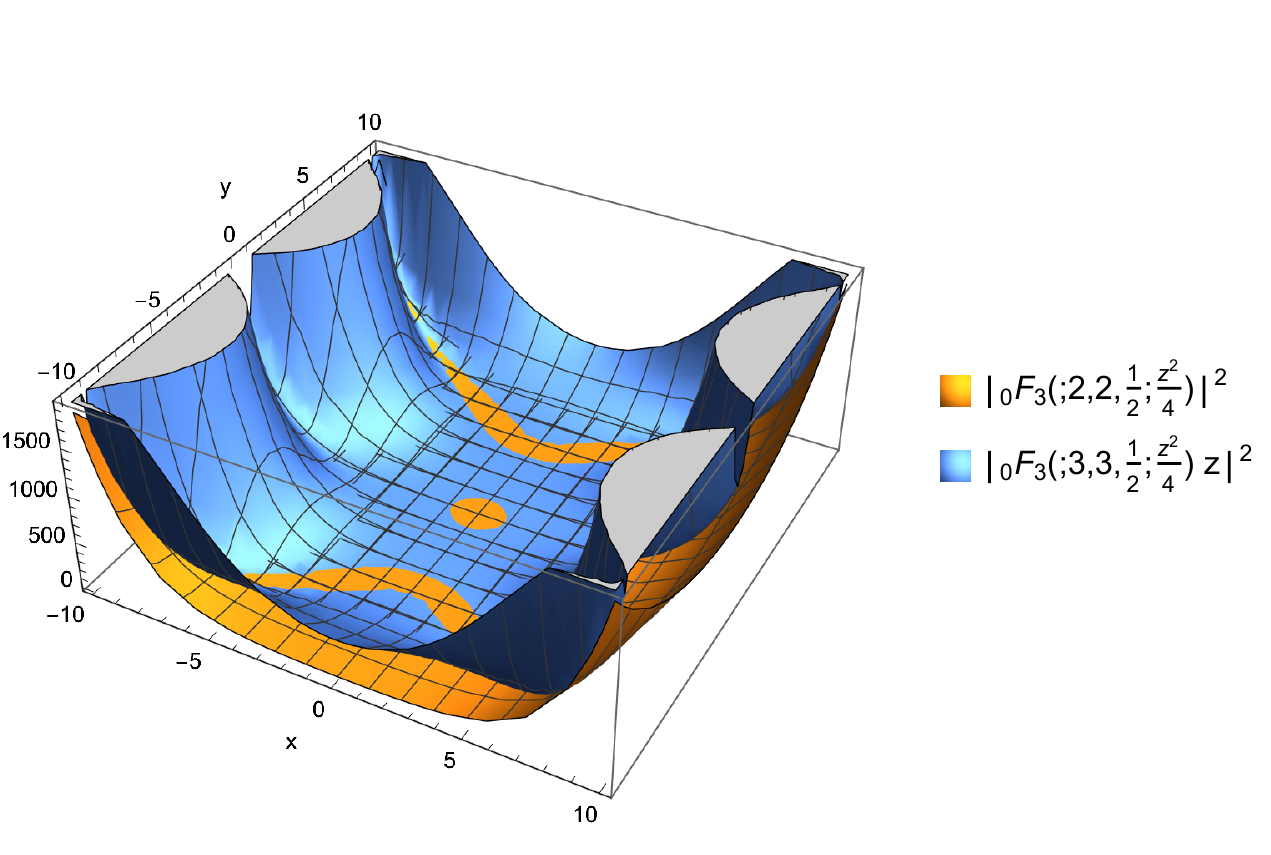}
	\caption{Norm square of the generalized hypergeometric function $|{}_0F_3(;3,3,\f32;\f{\z^2}{4})\,\z|^2$ and $|{}_0F_3(;2,2,\f12;\f{\z^2}{4})|^2$ as the function of $\z=x+iy$.}
	\label{fig:hypergeo}
\end{figure}
Plug this condition back into \eqref{eq:Az_Bw}, it is easy to conclude that the angle couplings $[w_l|w_{l'}\rangle$ and $[ z_l|z_{l'}\rangle$ are different only by a phase, similarly for $[w_{l'}|w_{l''}\rangle$ and $[z_{l'}|z_{l''}\rangle$, $[w_{l''}|w_{l}\rangle$ and $[z_{l''}|z_{l}\rangle$, and the phase differences are the same. 
More strictly, $|z_l]$ and $|w_l\rangle$ are different by a phase, the same valid for links $l'$ and $l''$, thus
\be
\langle z_l | z_l \rangle = \langle w_l | w_l \rangle\,,\quad
\langle z_{l'} | z_{l'} \rangle = \langle w_{l'} | w_{l'} \rangle\,,\quad
\langle z_{l''} | z_{l''} \rangle = \langle w_{l''} | w_{l''} \rangle\,.
\label{eq:same_norm}
\ee
Geometrically, the stationary points of the propagator enforce that the two triangles to be glued have the same shape. 
Let us emphasize that this interpretation solely results from our choice of the geometric gauge solution \eqref{eq:angle_couple_geometry} of the angle couplings in terms of the link couplings. A general solution \eqref{eq:sol_zz_gen_norm_phase} does not lead to $|A|=|B|$, thus it would not guarantee that the shape of the triangles to be glued share the same shape. 
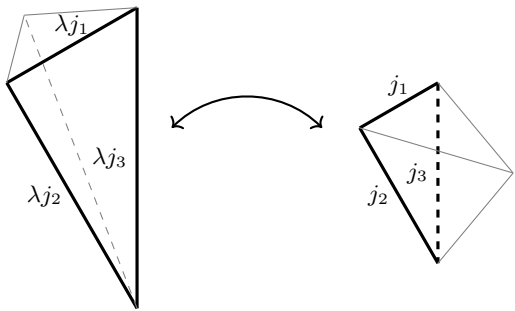
\begin{figure}[h!]
\centering
	\begin{tikzpicture}
\coordinate (A) at (0,0);
\coordinate (B) at (0,-4);
\coordinate (C) at ([shift=(-150:2)]A);
\coordinate (D) at (-1.5,-0.1);
\def\s{0.6};

\coordinate (a) at (4,-1);
\coordinate (b) at (4,-1-\s*4);
\coordinate (c) at ([shift=(-150:2*\s)]a);
\coordinate (d) at (5,-2.2);

\draw[very thick] (A) -- node[midway,left]{$\lambda j_3$} (B);
\draw[very thick] (B) -- node[midway,left]{$\lambda j_2$} (C);
\draw[very thick] (C) -- node[midway,above]{$\lambda j_1$} (A);
\draw[gray] (A) -- (D);
\draw[gray] (C) -- (D);
\draw[gray, dashed] (B) -- (D);

\draw[dashed, very thick] (a) -- node[midway,left]{$j_3$} (b);
\draw[very thick] (b) -- node[midway,left]{$j_2$} (c);
\draw[very thick] (c) -- node[midway,above]{$j_1$} (a);
\draw[gray] (a) -- (d);
\draw[gray] (c) -- (d);
\draw[gray] (b) -- (d);

\coordinate (l) at ([shift=(180:2.5)]c);
\coordinate (r) at ([shift=(180:0.5)]c);

\draw[<->,thick]  (l) to [out=45,in=135]  (r);
	\end{tikzpicture}
\caption{Stationary point of the propagator geometrical means the two triangles from adjacent tetrahedra to be glued (in thick) only need to share the same shape but can have different sizes. }
\end{figure}

In summary, the Ponzano-Regge state-integral model as constructed above encodes only the conformal geometry of the triangulation blocks either in the vertex amplitudes or the edge amplitudes. A vertex amplitude $\cA_{v^*}$ describes the shape of the tetrahedron dual to $v^*$; an edge amplitude $\cA_{e^*}$ describes that two adjacent tetrahedra can be glued together by identifying the shape of two triangles, each from one of the tetrahedra. When there exists a boundary $\partial \cM$ of the manifold, the total amplitude built in this way encodes also the conformal geometry on the triangulation of $\partial \cM$ since the boundary structure is described by the scaleless spin network state $\psi^{\sl}_\Gamma$.

\subsection{Ponzano-Regge state-integral versus state-sum models}

Let us recall the original spinfoam state-sum and the new spinfoam state-integral expression, 
\be\begin{split}
\cA[{\bf T},\psi_\Gamma ]
&= \sum_{\{j_{f^*}\}}\, \prod_{f^*}d_{j_{f^*}}\, 
\prod_{e^*}(-1)^{\sum_{i=1}^3 j_i}\,
\prod_{v^*} 
\Mat{ccc}{j_1 & j_2 & j_3 \\ j_4 & j_5 & j_6}_{v^*}\\
&=\int\left[\rd \mu(z)\right] \, 
\prod_{f^*}\left(\langle z^{f^*}|z^{f^*}\rangle-1 \right)\,
\prod_{e^*} \,{}_0F_3(;2,2,\f12; \f{\left(\sum_{l\in n}  \langle z_l^{s(e^*)}|w_l^{t(e^*)}]\right)^2}{4}  )\,
\prod_{v^*}\cS^{\sl}_{v^*}(\{z_l^{v^*},\zt_l^{v^*}\})\,.
\end{split}
\label{eq:two_amplitudes}
\ee
In the state-sum expression, when no boundary is present, the edge amplitude can be absorbed into the vertex amplitude, while the edge amplitude is always explicit in the state-integral expression. 
This leads to two different ways to understand spinfoam. 
In the former model, after choosing a triangulation ${\bf T}$ of $\cM$, we first associate the representation data (spins) on the one-skeleton of ${\bf T}$ then construct the vertex and face amplitudes with these representation data. In the latter model, in contrast, we first isolate all the elementary blocks, $\ie$ tetrahedra, after the triangulation then associate the representation data (spinors) to the boundary of each isolated block, followed by constructing vertex amplitude for each isolated block, edge amplitude through gluing these elementary blocks and finally face amplitude for each edge after gluing.

One may also absorb the edge amplitude into the vertex amplitude for the state-integral by integrating out, with no loss of generality, the spinors $\{w_l^{t(e^*)}\}$ from the target tetrahedron in each gluing. Unfortunately, we have not found a close form for this expression. Furthermore, since the vertex amplitude is no longer trivial, it would be changed after this absorption, which potentially changes the geometrical interpretation. 

On the other hand, leaving the edge amplitude un-absorbed allows us to separate the data from different blocks so that the saddle point analysis can be done for each vertex amplitude independently. Moreover, the saddle point analysis on the vertex amplitudes is compatible with that on the edge amplitudes.
The saddle point of each vertex amplitude \eqref{eq:saddle_edge_norm} relates the spinor configuration to spins within a single block, while the saddle point of the edge amplitudes \eqref{eq:same_norm} relates the spinors from different blocks. 
These saddle points can be obtained simultaneously and the result effectively relates the (ratios of) spins from different blocks (with spinors as the mediums). 

Geometrically, the gluing condition in the state-integral model is looser compared to that in the state-sum model, since the former only requires that the triangles to be glued have the same shape while the latter restricts that the triangles should be of the same shape and size. At first glance, it seems the state-integral allows more configurations and should produce a different total amplitude. However, the size un-matched configurations can not survive under the spinor integration, thus the total amplitude comes only from the size-wise and shape-wise matched configurations, same as the case of the state-sum model. We have used the same property of the spinor integration in constructing the state-integral to move the contour integral from the vertex amplitude \eqref{eq:vertex_amplitude_to_SGF} to the edge amplitude \eqref{eq:new_ldge_w_contour}. 

Another difference between the state-sum and state-integral models is the source of the divergence in the expressions. In the state-sum model, the vertex amplitude damps as $j^{-3/2}$. The divergence comes from the infinite sum of the spin labels and the total amplitude diverges as $\sqrt{j}$. In the state-integral model, in contrast, the integration does not lead to divergence thanks to the Gaussian weight while the vertex amplitudes give divergence since there are poles in the vertex amplitudes. This is because the vertex amplitude, as a generating function of the $\{6j\}$-symbols, contains the summation of spins, thus the divergence can be viewed as from the large spin contribution. To see that it is the case, we Taylor expand the vertex amplitude and look at the pole,
\be
\cA_{v^*}=\frac{1}{(1+x)^2}= \sum_{j\in \N/2} (2j+1) (-x)^{2j} 
\,\xrightarrow{x=-1}\, 
\sum_{j\in\N/2} (2j+1)\,,
\nn\ee 
where $x$ denotes the cycle-sums as given in \eqref{eq:evaluate_SGF}. This illustrates that the divergence of the vertex amplitude in the state-integral model is also given by large spins, which is consistent with the state-sum model.

\section*{Outlook}

We have introduced a new framework of the spinfoam model for 3D quantum gravity based on the spinor representation of $\SU(2)$. The continuum nature of spinor variables allows us to represent the spinfoam as a state-integral, rather than a state-sum in the original Ponzano-Regge model where the spin representation of $\SU(2)$ was used. The integral expression would probably make the computation of $\eg$ correlations and transition amplitudes more controllable. 
More importantly, the state-integral framework inherits the scale-invariant nature of pure gravity in 3D.  It describes a quantum gravity model with scale-invariant boundary quantum geometry, which can be seen as an integration over the conformal classes of boundary geometry. 
We expect that this framework would serve as a better starting point to study the coarse-graining or renormalization behaviour of 3D quantum gravity, and would be useful to investigate the quasi-local CFT/gravity duality. 

Apart from these possibly exciting applications, we list the following directions to better understand this new formalism and possible generalization:
\begin{itemize}
	\item {\it The geometrical interpretation of spinors.}
In order to understand the geometrical interpretation of the newly constructed spinfoam amplitude, it is enough to unravel the geometrical meaning of the holomorphic inner product of spinors, $\ie$ the angle couplings. It would be interesting to further understand the geometrical interpretation of the spinor itself. Bearing in mind the relation of the spinors and flux vectors, most importantly the equivalence of the spinor norm and the norm of the flux vector it describes (see \eqref{eq:flux_and_spinor}), we expect that the spinors encode the length/scale information, which is washed out by constructing the holomorphic inner product. 
	\item {\it The Wheeler de-Witt equations of the SGF.}
Recall that in the original Ponzano-Regge state-sum model, the invariance of the amplitude under Pachner moves is guaranteed by the Biedenharn-Elliott identity of the $\{6j\}$-symbols, which is closely related to the recursion relations of the $\{6j\}$-symbols. These recursion relations can be viewed as generated by the Wheeler de-Witt equations of the spin network states. It would be interesting to explore whether the invariance of the newly constructed amplitude formalism under Pachner moves can be directly proven by (possibly the deformation of) the Wheeler-de Witt equations of the SGF. If possible, it would provide a new way that is similar to the original state-sum model to see the topological invariance of the new spinfoam formalism, directly from the scale-invariant factorization of the amplitude, $\ie$ the second line of \eqref{eq:two_amplitudes}. As an intermediate step, it may be useful to rewrite the Wheeler de-Witt equation in terms of the link couplings so that variables are more decoupled. Similar to the differential equations found in \cite{Bonzom:2011nv}, one is able to find four such equations, each associated to one triangle on the boundary of the tetrahedron and only three are independent. An example is (referring to the notation in fig.\ref{fig:spinnetwork})
\be\begin{split}
&3 \left[
  Y_1\left(\frac{\partial \cS}{\partial Y_2 }\frac{\partial \cS}{\partial Y_6 }
- \frac{\partial \cS}{\partial Y_3 }\frac{\partial \cS}{\partial Y_5 }\right)
+ Y_2\left(\frac{\partial \cS}{\partial Y_3 }\frac{\partial \cS}{\partial Y_4 }
- \frac{\partial \cS}{\partial Y_1 }\frac{\partial \cS}{\partial Y_6 }\right)
+ Y_3\left(\frac{\partial \cS}{\partial Y_1 }\frac{\partial \cS}{\partial Y_5 }
- \frac{\partial \cS}{\partial Y_2 }\frac{\partial \cS}{\partial Y_4 }\right)
\right]\\
-&2 \cS \left[
  Y_1\left(\frac{\partial^2 \cS}{\partial Y_2 \partial Y_6 }
- \frac{\partial^2 \cS}{\partial Y_3 \partial Y_5 }\right)
+   Y_2\left(\frac{\partial^2 \cS}{\partial Y_3 \partial Y_4 }
- \frac{\partial^2 \cS}{\partial Y_1 \partial Y_6 }\right)
+  Y_3\left(\frac{\partial^2 \cS}{\partial Y_1 \partial Y_5 }
- \frac{\partial^2 \cS}{\partial Y_2 \partial Y_4 }\right)
\right]
=0\,,
\end{split}
\label{eq:diff_link_SGF}
\ee
where $\cS$ is the SGF and $Y_i$ the link coupling of link $l_i$.
On the other hand, as the SGF not only describes a flat tetrahedron but a scale-invariant one, it should be possible to determine a different differential equation that generates the scale-invariance of the SGF. The equation to be found describes the dilatation behaviour of the SGF, thus playing a dynamic role and can also be viewed as the Wheeler de-Witt equation. The flow of the differential equation should generate the symmetry of the link couplings \eqref{eq:angle_to_edge} in terms of spinors, namely the $\SL(2,\bC)$ symmetry and the ``anti-rescaling'' symmetry \eqref{eq:symm_SGF}. 
	\item {\it Spinfoam model for conformal quantum gravity. }
The building blocks of the new spinfoam are scale-invariant, thus almost but not yet conformal invariant. It would be interesting to push forward to construct a spinfoam with conformal blocks, whose boundary states are conformal invariant. Unfortunately, we have no clue whether it can be constructed with the spinor representation or what structures of the vertex amplitude we should expect. It may be more reliable to start the journey from the discretization of conformal gravity action and construct the partition function applying the action principle \cite{Freidel:1998pt}. Such a spinfoam model, if exists, would be a more suitable framework to study the CFT/gravity duality at the discrete level.

	\item {\it Group field theory for 3D quantum gravity with spinors. }
It has been well-known that the Ponzano-Regge models can be generated by a group field theory (GFT) with $\SU(2)$ group and the amplitude is interpreted as some Feynman graph evaluation \cite{Boulatov:1992vp}. GFT reformulation with spinors of the Ponzano-Regge model can capture the full $\SU(2)$ structure rather than $\SO(3)$ when using the flux and holonomy variables \cite{Freidel:2005bb,Freidel:2005me}. An attempt was made in \cite{Dupuis:2011fx}
where the amplitude is not yet separated into local blocks. In order to build a similar connection of the vertex amplitude with the interaction of the GFT and of the edge amplitude with the propagator of the GFT, a better starting point is a spinfoam formulation with local amplitudes, thus the result of Proposition \ref{prop:state_integral_2}. We expect that this GFT, if found, also has a scale-invariant nature, thus it may be better suited to study the renormalization properties in quantum gravity. 

	\item {\it Generalization to a 4D spinfoam. }
The spinfoam model built with spinors has been applied to 4D BF theory \cite{Dupuis:2011dh,Banburski:2014cwa,Chen:2016aag}. However there remains spin dependence or $\SU(2)$ holonomy dependence of the local amplitude in these existing models. It would be interesting to discover a 4D spinfoam with local amplitudes given only by the spinor variables and explore what these local amplitudes represent geometrically. A natural guess is that the vertex amplitude could be constructed by the generating function of the $15j$-symbols \cite{Huang:1974gs}, which can also be written into a closed form with a loop structure as the SGF. 

	\item {\it Including the cosmological constant. }
The spinor variables can be deformed to describe the loop gravity with a negative cosmological constant \cite{Dupuis:2014fya,Bonzom:2022bpv}, which recovers the loop gravity framework described by deformed holonomy and flux variables \cite{Bonzom:2014wva} (see also \cite{Bonzom:2014bua} for its quantization). It would be interesting to see if the deformed spinors can be used to construct the spinfoam that recovers the Turaev-Viro model \cite{Turaev:1992hq}, whose building blocks are the $q$-deformed $\{6j\}$-symbols. One of the first things to do is to construct the ``$q$-deformed scaleless spin network statl'' on the hyperbolic geometry. We expect that the spinorial framework of LQG and the spinfoam model can be generalized to a $q$-deformed version with the use of quantum groups and describe the quantum gravity with a non-vanishing cosmological constant. 
\end{itemize}

\section*{Acknowledgement}

This work is part of the project initiated in the Perimeter Quantum Gravity workshop that happened in the winter of 2018. The authors thank the organizers of the workshop. 
Research at Perimeter Institute is supported by the Government of Canada through the Department of Innovation, Science and Economic Development Canada and by the Province of Ontario through the Ministry of Research, Innovation and Science.
QP was supported by an NSERC Discovery grant awarded by Maïté Dupuis during this work.
%

\appendix

\renewcommand\thesection{\Alph{section}}

\section{Spinorial phase space for loop gravity}
\label{app:spinorial}
In this appendix, we give a quick review of the spinorial phase space of loop gravity and its quantization to construct the Hilbert space of LQG. 
Let us introduce the spinor variables
\be
\kz:=\left(
\ba{c}
z^0\\z^1
\ea
 \right)\,,\quad
\bz:=\left(
\ba{cc}
\zb^0, & \zb^1
\ea
 \right)\,,\quad
z^0,z^1\in\bC \,.
\label{eq:spinor}
\ee
We also introduce the dual spinor by a dual map $\varsigma$
\be
\kzd\equiv |\varsigma z \rangle := \mat{cc}{0&-1\\1&0}\kzb
= \left(
\ba{c}
-\zb^1\\ \zb^0
\ea
 \right)\,,\quad
\bzd\equiv \langle \varsigma z|:= 
\bzb \mat{cc}{0&1\\-1&0} =
\left(
\ba{cc}
-z^1, & z^0
\ea
 \right)\,
 \label{eq:dual_spinor}.
\ee
A spinor transforms covariantly under the $\SU(2)$ action $\kz\rightarrow g\kz$ in the fundamental representation of $\SU(2)$. Thus the $\SU(2)$, hence $\SL(2,\bC)$ by complexification, invariant objects can be naturally formed by the inner product of two spinors 
\be
\left<w|z\right>=\left[z|w\right]=\sum_A \bar{w}^Az^A\,,\quad
\left[w|z\right>=-w^1 z^0 + w^0 z^1\,, \quad
\left<w|z\right] = -\bar{w}^0\zb^1+\bar{w}^1\zb^0\,.
\ee
$|z]$ is dual to $|z\rangle$ in the sense that they are orthogonal via the inner product, $\ie$ $[z|z\rangle=0$. 

Consider a pair of spinors $\{|z\rangle,|\zt\rangle\}$, whose components are provided with the Poisson brackets
\be
\{z^A,\zb^B\}=i\delta^{AB}\,,\quad
\{\zt^A,\ztb^B\}=-i\delta^{AB}\,,\quad
\{z^A,\zt^B\}=\{z^A,\ztb^B\}=\{\zb^A,\zt^B\}=\{\zb^A,\ztb^B\}=0\,,\quad
\,,A,B=0,1\,.
\label{eq:poisson_spinor}
\ee
We introduce the Hermitian matrices ${\bf X}:=|z\rangle\langle z|$ and ${\bf \tX}:=|\zt\rangle\langle\zt|$, then project them onto the identity and the Pauli matrices
\be\ba{ccc}
|z\rangle \langle z| = |\vec{X}| \,\id - \vec{X}\cdot \vec{\sigma} \,\quad
\text{with }\, &|\vec{X}| \equiv \f12 \langle z|z\rangle\,, \quad
\vec{X} &\equiv -\f12\langle z|\vec{\sigma}|z\rangle \in \R^3\,,\\[0.15cm]
|\zt\rangle\langle \zt| = |\vec{\tX}| \,\id - \vec{\tX}\cdot \vec{\sigma} \,\quad
\text{with }\, &|\vec{\tX}| \equiv \f12\langle\zt|\zt\rangle\,, \quad
\vec{\tX} &\equiv -\f12\langle\zt|\vec{\sigma}|\zt\rangle\in \R^3 \,,
\ea
\label{eq:flux_and_spinor}
\ee
where $|\vec{X}|$ ({\it resp. }$|\vec{\tX}|$) is the norm of the vector $\vec{X}$ ({\it resp. }$\vec{\tX}$). 
We relate the non-tilde spinor variables and the tilde spinor variables by an $\SU(2)$ action
\be
g|z\rangle=|\zt\rangle\,,\quad
g|z]=|\zt] \,,\quad
g\in\SU(2)\,.
\label{eq:SU2_action}
\ee
Thus ${\bf X}$ and ${\bf \tX}$ are related by an $\SU(2)$ adjoint action: ${\bf \tX}=g{\bf X}g^{-1}$.
\eqref{eq:SU2_action} determines the $\SU(2)$ group element $g$ uniquely to be
\be
g=\frac{|\zt\rangle\langle z| + |\zt][z|}{\sqrt{\langle z|z\rangle\langle \zt|\zt\rangle}}
\,.
\ee 
It also implies that the spinors $|z\rangle$ and $|\zt\rangle$ satisfy a norm-matching constraint
\be
\cN:=\langle z|z\rangle  - \langle \zt|\zt\rangle\,,
\label{eq:norm_matching}
\ee
which generates a $\bU(1)$ transformation on spinors $\{\cN,\kz\}=-i\kz$, $\{\cN,|\zt\rangle\}=-i|\zt\rangle$. The finite gauge transformation reads 
\be
\kz \xrightarrow{\bU(1)} e^{i\theta}\kz\,,\quad
|\zt\rangle \xrightarrow{\bU(1)} e^{i\theta}\, |\zt\rangle\,,\quad
\bz \xrightarrow{\bU(1)} e^{-i\theta}\bz\,,\quad
\langle\zt| \xrightarrow{\bU(1)} e^{-i\theta}\langle\zt|\,.
\label{eq:U1_transf}
\ee

The Poisson brackets \eqref{eq:poisson_spinor} can thus be rewritten with the new variables $(\vec{X},g)$ (or equivalently $(\vec{\tX},g)$) as
\be\ba{lll}
\{X^i,g\}=\frac{i}{2}g \sigma^i \,,&
\{X^i,X^j\}=\epsilon^{ijk}X^k\,,&
\{g,g\}\stackrel{\cN=0}{\sim}0\,,\\[0.15cm]
\{\tX^i,g\}=\frac{i}{2}\sigma^i g\,,&
\{\tX^i,\tX^j\}=-\epsilon^{ijk}\tX^k\,,&
\{X^i,\tX^j\}=0\,,
\ea\ee
which are the Poisson structure of the loop gravity phase space with vanishing cosmological constant, $\ie$ the $T^*\SU(2)$ phase space. $\vec{X}$ (or $\vec{\tX}$) is called the flux and $g$ the holonomy. Note that $({\bf X},g)$ are invariant under the $\bU(1)$ transformation \eqref{eq:U1_transf}, thus the loop gravity phase space can be reconstructed completely as a symplectic reduction of the spinor space $\bC^2\times \bC^2$:
\be
T^*\SU(2)\backslash \{|\vec{X}|=0\}= \bC^2\times \bC^2 \backslash \{\left<z|z\right>=0,\left<\zt|\zt\right>=0\}//\bU(1)\,.
\ee

Consider an oriented graph $\Gamma$ with $|L|$ links and $|N|$ nodes. 
For each link $l$, we assign a non-tilde spinors $|z_l\rangle$ to its source node $s(l)$ and a tilde spinors $|\zt_l\rangle$ to its target node $t(l)$. They are related by an $\SU(2)$ action thus satisfying the norm matching constraint
\be
g_l |z_l\rangle=|\zt_l\rangle\,,\quad
\langle z_l|z_l\rangle=\langle \zt_l|\zt_l\rangle\,.
\ee
$g_l$'s can be viewed as an assignment to the links. 
The kinematical phase space of $\Gamma$ is defined as the collection of $(z_l,\zt_l)\in \bC^2\times\bC^2$ for each link imposing the closure constraints $\vec{\cC}_v$ for each node:
\be
\vec{\cC}_n=\sum_{l|n=s(l)}\left<z_l|\vec{\sigma}|z_l\right>
-
\sum_{l|n=t(l)}\langle\zt_l|\vec{\sigma}|\zt_l\rangle\,.
\ee

The phase space defined with spinors allow us to have a $\bU(N)$ reformulation of LQG after quantization \cite{Freidel:2009ck,Freidel:2010tt,Borja:2010rc}. The essential idea is to change the building blocks from degrees of freedom on links (which are holonomies or fluxes) to those on nodes (which are spinors). For each ($N$-valent) node $n$, we define an $N\times N$ asymmetric matrix ${\bf F}$ and an $N\times N$ symmetric matrix ${\bf E}$ with complex entries representing the correlation of spinors associated to the different half-links incident to the same node.  (we will use the Latin indices $a,b,c,d$ in the subscript to denote the legs of the node):
\be\ba{ll}
E_{ab}= \langle z_a|z_b\rangle\,,\quad 
& E_{ba}=\overline{E}_{ab} \,,\\[0.15cm]
F_{ab}= [z_a|z_b\rangle\,,\quad
& F_{ba}=-F_{ab}\,,\\[0.15cm]
\overline{F}_{ab}=\langle z_b|z_a] \,,\quad 
& \overline{F}_{ba}=-\overline{F}_{ab}\,,
\ea
\label{eq:EF}
\ee
where the bar denotes the complex conjugate. 
These $\SU(2)$ observables form a closed algebra. With no loss of generality, consider that all the links incident to the node $n$ are outgoing (hence only non-tilde spinors are attached to $v$), then the Poisson brackets of the components read 
\be\begin{split}
\{E_{ab},E_{cd}\}&=i(\delta_{ad}E_{cb} - \delta_{bc}E_{ad})\,,\\
\{E_{ab},F_{cd}\}&=i(\delta_{ad}F_{bc}-\delta_{ac}F_{bd} )\,,\\
\{E_{ab},\overline{F}_{cd}\}&=i(\delta_{bd}\overline{F}_{ac} -\delta_{bc}\overline{F}_{ad})\,,\\
\{F_{ab},F_{cd}\}&=0\,,\quad
\{\overline{F}_{ab},\overline{F}_{cd}\}=0\,,\\
\{F_{ab},\overline{F}_{cd}\}&= i ( \delta_{ac}E_{db} - \delta_{ad}E_{cb} + \delta_{bd}E_{ca} - \delta_{bc}E_{da} )\,.
\end{split}
\label{eq:Poisson_EF}
\ee
It can be seen from the first Poisson bracket that the components of the matrix ${\bf E}$ form a $\u(N)$ algebra. The full algebra \eqref{eq:Poisson_EF} is called the $\so^*(2N)$ \cite{Girelli:2017dbk}.

Upon quantization, the spinors become the annihilation operators $z^A\rightarrow a^A$ and the creation operator $\zb^A \rightarrow a^{A\dagger}$ such that satisfy the commutator 
\be
[a^A,a^{B\dagger}]=\delta^{AB}\id \,,\quad
[a^A,a^B]=[a^{A\dagger},a^{B\dagger}]=0\,.
\ee
This being said the phase space for an $N$-valent node is quantized to be a set of $2N$ harmonic oscillators. 
We also quantize the observable matrix ${\bf E}$, ${\bf F}$ and ${\bf \overline{F}}$ in the following way.
\be\ba{lll}
E_{bc}=\langle z_b|z_c \rangle  & \longrightarrow & \hat{E}_{bc}=a_b^{A\dagger}a_{c}^A\,,\\[0.15cm]
F_{bc}=[ z_b|z_c\rangle & \longrightarrow & \hat{F}_{bc}=\epsilon^{AB}a_b^Aa_c^B\,,\\[0.15cm]
\overline{F}_{bc}=\langle z_c|z_b ] & \longrightarrow & \hat{F}_{bc}^\dagger = \epsilon^{AB} a_b^{A\dagger}a_c^{B\dagger}\,.
\ea\ee
We have used the normal ordering when necessary. The commutators naturally inherit from the Poisson brackets \eqref{eq:Poisson_EF}. 
It is natural to obtain the Fock states $|n^0,n^1\rangle_{\HO}$ which diagonalize the occupation number operator $N^A:=a^{A\dagger}a^{A}$
\be
N^A|n^0,n^1\rangle_{\HO} = n^A |n^0,n^1\rangle_\HO\,.
\ee
This basis is equivalent to the magnetic number basis $|j,m\rangle\in \cV^j$ with the relation between the eigenvalue as
\be
j=\f12 (n^0 + n^1)\,,\quad
m=\f12 (n^0 - n^1)\,.
\ee
The action of $(a^A, a^{B\dagger})$ on the magnetic number basis allows the jumping between different spin representations. Explicitly,
\be
a^A|j,m\rangle = \sqrt{j+(-1)^{A}m}\, |j-\f12, m-\f12+A\rangle\,,\quad
a^{A\dagger}|j,m\rangle = \sqrt{j+(-1)^{A}m+1}\,|j+\f12,m+\f12-A\rangle\,.
\label{eq:a_on_jm}
\ee
By the definition of the coherent state basis through the magnetic number basis \eqref{eq:jz_basis}, it is easy to find that $a^A$ acts on $|j,z\rangle$ as a multiplication operator while $a^{A\dagger}$ acts as a derivative operator on $|j,z\rangle$ as given in \eqref{eq:a_on_jz}.

The LQG Hilbert space $L_2(\SU(2), \rd g)^{|L|}//\SU(2)^{|N|}$ is standardly understood as spanned by the spin network states labelled by spins. The quantization of the spinorial phase space allows us to span the same Hilbert space by the coherent or scaleless spin network states labelled by spinors. To this end, we also need to introduce a Haar measure. It is given by the Haar measure $\rd \mu(z)$ of the Bargmann space $\cF_2=L_2^{\rm hol}(\bC^2,\rd\mu)$, the space of holomorphic squared integrable functions, over spinor \cite{Livine:2011gp}:
\be
\rd\mu(z):=\frac{1}{\pi^2}e^{-\left<z|z\right>}\rd z^0 \rd z^1\,.
\ee
It is a measure invariant under the $\SU(2)$ transformation $\rd\mu(g z)=\rd\mu(z),\,\forall\, g\in \SU(2)$. 
The space $\cF_2$ can be decomposed into the direct sum of $(2j+1)$-dimensional subspace: $\cF_2=\oplus_{j\in \N/2}\cV^j$, with the orthonormal basis of each spin $j$ subspace given by 
\be
e_m^j(z):=\frac{(z^0)^{j+m}(z^1)^{j-m}}{\sqrt{(j+m)!(j-m)!}}\,.
\ee
The Hilbert space of one link is thus equivalently given by $\cH_l=\cF_2\times\cF_2//\bU(1)$. Readers can find more details for the Bargmann space reconstruction of the Hilbert space in \cite{Livine:2011gp}. By taking the closure constraint for nodes into consideration, we conclude that the Hilbert space spanned by the coherent or scaleless spin network states can be represented as
$L^{\rm hol}_2(\bC^2,\rd \mu)^{2|L|}//(\SU(2)^{|N|}\times\bU(1)^{|L|})$. 

\section{Wheeler-de Witt equation of the SGF}
\label{app:WdW_SGF}
In this appendix, we review the Wheeler-de Witt equations of the physical state of the tetrahedron graph, which is the $\{6j\}$-symbol in the spin network basis and the SGF in the scaleless spin network basis. Detailed analysis can be found in \cite{Bonzom:2011hm,Bonzom:2011nv}.

Throughout this appendix, we refer to the tetrahedron graph in fig.\ref{fig:spinnetwork}. Consider the Hamiltonian
\be
\hat{H}_{126}\, \psi_\phys(g_l) = 0 \,,\quad
\hat{H}_{126}:= g_1g_6g_2^{-1} -\id\,.
\label{eq:WdW_g126}
\ee
By the discrete nature of the spin labels defining a spin network basis, the Wheeler-de Witt equation can be represented by a recursion relation of the $\{6j\}$-symbols. 
 It was proposed in \cite{Bonzom:2011hm} a spin 1 Hamiltonian (denoted with a superscript $(1)$) by projecting the rotation matrix $R(g_6 g_1 g_2^{-1})$ onto the fluxes $\vec{X}_2$ and $\vec{X}_6$ associated to the node where links $l_2$ and $l_6$ meet,
\be
H_{126}^{(1)}=\vec{X}_6 \cdot \left( \id - \Ad(g_6 g_1 g_2^{-1}) \right)\, \vec{X}_2 
= \vec{X}_6 \cdot \vec{X}_2  - \vec{X}_6 \cdot \Ad(g_1) \,\vec{X}\,.
\label{eq:H_126_spin1}
\ee 
It corresponds to the recursion relation of the $\{6j\}$-symbol involving a shift of a spin label by one \cite{Bonzom:2011hm}:
\be
A_{+1}(j_1)\Mat{ccc}{j_1+1 & j_2 & j_3 \\ j_4 & j_5 & j_6}
+A_0^{(1)}(j_1)\Mat{ccc}{j_1 & j_2 & j_3 \\ j_4 & j_5 & j_6}
+A_{-1}(j_1)\Mat{ccc}{j_1+1 & j_2 & j_3 \\ j_4 & j_5 & j_6}
=0\,,
\ee
with coefficients
\begin{align}
A_0^{(1)}(j_1)
&=(-1)^{j_2+j_4+j_6}\Mat{ccc}{j_2 & j_2 & 1 \\ j_6 & j_6 & j_4}
+(-1)^{2j_1+j_2+j_3+j_5+j_6}(2j_1+1)\Mat{ccc}{j_1 & j_1 & 1 \\ j_2 & j_2 & j_3}
\Mat{ccc}{j_1 & j_1 & 1 \\ j_6 & j_6 & j_5}\,, \\ 
A_{\pm 1}(j_1)
&= (-1)^{2j_1+j_2+j_3+j_5+j_6+1}
(2(j_1\pm 1)+1)
\Mat{ccc}{j_1\pm 1 & j_1 & 1 \\ j_2 & j_2 & j_3}
\Mat{ccc}{j_1\pm 1 & j_1 & 1 \\ j_6 & j_6 & j_5}\,.
\end{align}
The reason for obtaining a recursion relation as such with an arguments changed by 1 is that the fluxes transform in the adjoin representation of $\SU(2)$, thus the spin 1 representation. Fluxes can be reproduced with spinor variables \eqref{eq:flux_and_spinor}. The inner product of fluxes $\vec{X}_2$ and $\vec{X}_6$ can be written as the inner product of spinors $\zt_2$ and $\zt_6$ associated to the same node,
\be
\vec{X}_6\cdot \vec{X}_2 
\equiv \langle \zt_2|\zt_6\rangle \langle \zt_6|\zt_2\rangle 
- \langle \zt_2|\zt_6][\zt_6|\zt_2\rangle\,.
\label{eq:flux_from_spinors_26}
\ee
The use of spinors makes it possible to formulate a spin $1/2$ Hamiltonian since spinors transform in the fundamental representation of $\SU(2)$. Such a Hamiltonian is expected (and was shown in \cite{Bonzom:2011nv}) to generate a recursion relation of the $\{6j\}$-symbols with arguments changed by $1/2$. As proposed in \cite{Bonzom:2011nv}, the spin $1/2$ Hamiltonian (denoted with a superscript $(1/2)$) represents the unchanged inner product of the spinors $\zt_2$ and $\zt_6$ after transported around the plaquette surrounded by links $e_1,e_2,e_6$,
\be
H^{(1/2)}_{126}:=\langle\tz_6|\tz_2] [\tz_6 | g_6g_1g_2^{-1} -\id |\tz_2\rangle \,.
\label{eq:H_126_spin1o2}
\ee
The quantization of spinors leads to creation and annihilation operators. Since there are two separate scalar terms in the formula, two Hamiltonian operators are available upon quantization. The left ordering Hamiltonian is \cite{Bonzom:2011nv}
\be
\hat{H}^{(1/2)}_{126}|_L = \langle \tilde{a}_6|\tilde{a}_2] \left([a_6|g_1|a_2\rangle -[\tilde{a}_6|\tilde{a}_2\rangle \right) 
=  \hat{F}^{\dagger}_{26} \left(\frac{1}{d_{j_1}}\left( \hat{F}_{61}\hat{E}_{12} + \hat{E}_{16}\hat{F}_{12} \right) - \hat{F}_{62}\right)\,,
\ee
and the right ordering Hamiltonian is
\be
\hat{H}^{(1/2)}_{126}|_R = \left([a_6|g_1|a_2\rangle -[\tilde{a}_6|\tilde{a}_2\rangle \right) 
\langle \tilde{a}_6|\tilde{a}_2] 
= \left(\frac{1}{d_{j_1}}\left( \hat{F}_{61}\hat{E}_{12} + \hat{E}_{16}\hat{F}_{12} \right) - \hat{F}_{62}\right) \hat{F}^{\dagger}_{26}\,.
\ee
It turns out that only the right ordering Hamiltonian can generate the desired recursion relation. $\hat{H}^{(1/2)}_{126}|_R $ annihilates the tetrahedral spin network state evaluated on the identity
\be
\hat{H}^{(1/2)}_{126}|_R \, s_{\tet}^{\{j_l\}}(\id)=0\,.
\label{eq:Hamiltonian_1/2}
\ee
By expanding the terms and applying the annihilation and creation operators on the magnetic number basis \eqref{eq:a_on_jm}, \eqref{eq:Hamiltonian_1/2} reproduces a spin $1/2$ recursion relation on the $\{6j\}$-symbols \cite{Bonzom:2011nv}
\be
A_{+\f12}(j_1)
\left\{\ba{ccc}j_1+\f12 & j_2-\f12 & j_3 \\ j_4 & j_5 & j_6-\f12 \ea\right\}
+A_0^{(1/2)}(j_1)
\left\{\ba{ccc} j_1 & j_2 & j_3 \\ j_4 & j_5 & j_6 \ea\right\}
+A_{-\f12}(j_1)
\left\{\ba{ccc} j_1-\f12 & j_2-\f12 & j_3 \\ j_4 & j_5 & j_6-\f12 \ea\right\}
=0\,,
\ee
with coefficients
\begin{align}
A_0^{(1/2)}(j_1)
&= (-1)^{j_2+j_4+j_6+1}\Mat{ccc}{j_1 & j_2 -\f12 & \f12 \\ j_6-\f12 & j_6 & j_4} \\
A_{\pm \f12}(j_1)
&=
(-1)^{2j_1+j_2+j_3+j_5+j_6-\f12 \pm \f12} (2(j_1\pm\f12)+1)
\Mat{ccc}{j_1\pm\f12 & j_1 & \f12 \\ j_2 & j_2-\f12 & j_3}
\Mat{ccc}{j_1\pm\f12 & j_1 & \f12 \\ j_6 & j_6-\f12 & j_5}\,.
\end{align}
The same Hamiltonian can, at the same time, annihilate the tetrahedral scaleless spin network evaluated on the identity, thus the SGF,
\be
\hat{H}^{(1/2)}_{126}|_R \, s^{\sl}_{\tet}(\id)=0\,.
\ee
Similarly, expanding terms and applying the annihilation and creation operators on the coherent state basis \eqref{eq:a_on_jz} and \eqref{eq:a_on_jz_dual}, this reproduces a differential equation of the SGF \cite{Bonzom:2011nv},
\be\begin{split}
&\left( \sum_{A,B=0,1}
\left(\frac{\partial}{\partial \tz_1^A}\otimes z_6^A \right)
\left( \frac{\partial}{\partial z_1^B}\otimes z_2^B \right)
\right)
\left(2+\f12 \sum_{l\in n_{246}}\sum_A z^A \frac{\partial}{\partial z^A} \right) 
-\left[\frac{\partial}{\partial \tz_6}\middle\vert\frac{\partial}{\partial \tz_2}\right>
\left( 1+\sum_A z_1^A \frac{\partial}{\partial z_1^A} \right)\\
&+[\tz_1|z_6\rangle [z_1|z_2\rangle 
\left( 2+\f12 \sum_{l\in n_{246}}\sum_A z^A \frac{\partial}{\partial z^A} \right)
\left( 2+\f12 \sum_{l\in n_{123}}\sum_A z^A \frac{\partial}{\partial z^A} \right)
\left( 2+\f12 \sum_{l\in n_{156}}\sum_A z^A \frac{\partial}{\partial z^A} \right)
\cS^{\sl}(\{z_l,\tz_l\})=0\,,
\end{split}\ee
which is noted as the Wheeler-de Witt equation of the SGF. There are totally four such equations, each associated to one plaquette of the tetrahedral graph, while only three are independent. As expected, these expressions look more cumbersome than \eqref{eq:diff_link_SGF} since the SGF is in a more coupled fashion when written in terms of the angle coupling rather than the link couplings. 

\section{The Ponzano-Regge model in terms of coherent blocks}
\label{app:coherent_PR}

In the section, we prove that the Ponzano-Regge amplitude \eqref{eq:group_form} can be written in a way that the vertex amplitudes are given by the coherent states \eqref{eq:coherent_SN} evaluated on identity and edge amplitudes in a simple Gaussian form. This is the starting point of the Proposition \ref{prop:state_integral_2} and 3D version of \cite{Dupuis:2011dh}.

Let us recall the notations we will use to write the amplitudes, which are the same as those in Proposition \ref{prop:state_integral_2}.

\begin{prop}
The spinfoam model can be expressed as an integral in terms of coherent blocks
\be
\cA_{\bf T}[\cM,\psi_\Gamma^{\text{cohe}} ]=\int \left[\rd \mu(z)\right] \prod_{f^*} \cA_{f^*}[z_{f^*}] \prod_{e^*}\cA_{e^*}[z_{e^*}]\prod_{v^*}\cA_{v^*}[z_{v^*}]
\nn\ee
with the vertex, edge and face amplitude written as
\begin{align}
\cA_{v^*}
&=
s_{\tet}^{\cohe}(\id)
=\int_{\SU(2)^4}\prod_{n=1}^4 \rd h_n\, e^{\sum_{e=1}^6 [\zt_l|
h_{t(l)}^{-1}h_{s(l)}
|z_l\rangle}\,,
\label{eq:vertex_amplitude}
\\
\cA_{e^*}
&=
e^{\sum_{l\in n} \langle z_l^{s(e^*)}|z_l^{t(e^*)}]}\,,
\label{eq:edge_amplitude_Schwinger}
\\
\cA_{f^*}
&=\langle z^{f^*,T_1}|z^{f^*,T_1}\rangle-1\,.
\label{eq:face_amplitude}
\end{align}
\label{prop:state_integral_1}
\end{prop}
\begin{proof}
Recall that the coherent spin network state for a tetrahedron graph is 
\be
\sum_{\{j_l\}}
s_{\tet}^{\cohe}
=\int_{\SU(2)^4}\prod_{n=1}^4 \rd h_n\, e^{\sum_{l=1}^6 [\zt_l|
h_{t(l)}^{-1}g_lh_{s(l)}
|z_l\rangle}\,.
\ee
The evaluation on identity gives the vertex amplitude \eqref{eq:vertex_amplitude}.
To glue the vertex amplitudes associated to adjacent tetrahedra, we make use of the identity in the representation space $\cV^j\otimes \cV^{*j}$ spanned by the coherent states \cite{Livine:2011gp}
\be
\id_{\cV^j\otimes\cV^{*j}}=\frac{1}{(2j)!}\int \rd \mu(z) |j,z\rangle \langle j,z|
=\frac{1}{(2j)!}\int \rd \mu(z) |j,z][ j,z|
\,.
\label{eq:id_Vj}
\ee
To do the gluing, we also use the following identity,
\be\begin{split}
&\int\rd\mu(w_1)\int\rd\mu(w_2)\, e^{[ z|g|w_1\rangle+\langle w_1|w_2]+ [ w_2|h|z'\rangle}\\
=&\sum_{j,k,q\in\N/2}\frac{1}{(2j)!(2k)!(2q)!}\int\rd\mu(w_1)\int\rd\mu(w_2)\,
[ j,z|g|j,w_1\rangle \langle k,w_1|k,w_2] [ q,w_2|h|q,z'\rangle \\
=&
\sum_{j\in\N/2}\frac{1}{(2j)!}[ j,z|gh|j,z'\rangle
=e^{[ z|gh|z'\rangle}\,.
\end{split}\nn\ee
We have used \eqref{eq:id_Vj} to obtain the third line.
This product rule can be applied to contract terms between adjacent tetrahedra, say $T_1$ and $T_2$, connected with the triangle whose 2D dual is a node $n$ and 3D dual is an oriented dual edge $e^*$. Say the source dual vertex $s(e^*)$ is dual to $T_1$ and the target $t(e^*)$ is dual to $T_2$. Consider the graphs $(\partial T_1)^*_1$ and $(\partial T_2)^*_1$ both including the node $n$. To identify the triangles from $T_1$ and $T_2$ is to identify the three links $l,l',l''\in n$ from the two graphs. However, the spinors associated to links from different graphs are different. For instance, consider a link $l\in n$, the spinor $z_l^{T_1}$ (or $z_l^{s(e^*)}$ with the dual language) on $(\partial T_1)^*_1$ is not the same as the spinor $z_l^{T_2}$ (or $z_l^{t(e^*)}$) on $(\partial T_2)^*_1$. Integrating over the relevant terms from vertex amplitudes of $s(e^*)$ and $t(e^*)$, and the edge amplitude of $e^*$, one gets (we write only integration for one link $l$ for short)
\be\begin{split}
&\int \rd\mu(z_l^{s(e^*)})\int \rd\mu(z_l^{t(e^*)})\, 
e^{[z^{n_1} |\left(h_{n_1}^{T_1}\right)^{-1}h^{T_1}_{n}|z_l^{s(e^*)}\rangle }\,
e^{\langle z_l^{s(e^*)}|z_l^{t(e^*)}]} \,
e^{[z^{t(e^*)}|\left(h_n^{T_2}\right)^{-1}h_{n_2}^{T_2}|z^{n_2}\rangle}\\
=&e^{[ z^{n_1}|\left(h_{n_1}^{T_1}\right)^{-1}h_{n}^{T_1}\left(h_n^{T_2}\right)^{-1}h_{n_2}^{T_2}|z^{n_2}\rangle}
=e^{[z^{n_1}|\left(h_{n_1}^{T_1}\right)^{-1}h_nh_{n_2}^{T_2}|z^{n_2}\rangle}\,,
\label{eq:cal_1}
\end{split}\ee
with $h_n\equiv h^{T_1}_{n}\left(h_n^{T_2}\right)^{-1}$. 

The contraction \eqref{eq:cal_1} can be repeatedly performed along a closed chain $(e_1^*e_2^*...e_M^*e_1^*)\in T$* (or equivalently $(t_1t_2\cdots t_Mt_1\in T$) surrounding an edge $e$ shared by $M$ tetrahedra, as illustrated in fig.\ref{fig:order_tetra}. The edge amplitude can be absorbed into the vertex amplitudes. As a result, one simply flips of the spinors $[z_l^{t(e^*)}|$ associated to the target $t(e^*)$ of $e^*$ to their dual $\langle z_l^{t(e^*)}|$ and identify the spinors for the same link from different tetrahedra. In this way, then the gluing of a close chain of tetrahedra reads
\be
\int\prod_{i=1}^M\rd\mu(z^{v_i})
e^{\langle z^{n_1}|
\left(h_{n_1}^{T_1}\right)^{-1}h_{n_2}^{T_1}
|z^{n_2}\rangle
+\langle z^{n_2}|
\left(h_{n_2}^{T_2}\right)^{-1}h_{v_3}^{T_2}
|z^{v_3}\rangle
+\cdots
+\langle z^{n_M}|
\left(h_{n_M}^{T_n}\right)^{-1}h_{n_1}^{T_n}
|z^{n_1}\rangle
}
=\int\rd\mu(z^{n_1})
e^{\langle z^{n_1}|G_{e}|z^{n_1}\rangle}
\label{eq:gluing_flat}
\ee
with $G_{e}\equiv \left(h_{n_1}^{T_1}\right)^{-1}h_{n_2}h_{v_3}\cdots h_{n_M}h_{n_1}^{T_n}$.
Different from the magnetic number basis, each dual face is weighted by a factor $\langle z^{n_1}|G_{e}|z^{n_1}\rangle$ depending on the the spinor attached to the (randomly chosen) base node $n_1$ and the Wilson loop around the dual face, or a factor $(\langle z^{n_1}|z^{n_1}\rangle -1)$ (or ($\langle z^{f^*,T_1}|z^{f^*,T_1} \rangle -1$) depending only on the spinor attached to the $n_1$ \cite{Dupuis:2011fz}.  
At the end of the day, one gets the partition function expressed as
\be\begin{split}
Z_{\bf T}[\cM,\partial\cM]
&=\int_{\SU(2)}
\left(\prod_{n\in \Gamma}\rd h_{n}\right)
\prod_{f^*\in {\bf T}^*}
\int\rd\mu(z^{n_1})
(1+\langle z^{n_1}|G_{e}|z^{n_1}\rangle)
e^{\langle z^{n_1}|G_{e}|z^{n_1}\rangle}\\
&
=\int_{\SU(2)}
\left(\prod_{n\in \Gamma}\rd h_{n}\right)
\prod_{f^*\in {\bf T}^*}
\int\rd\mu(z^{n_1})
(\langle z^{n_1}|z^{n_1}\rangle -1\rangle)
e^{\langle z^{n_1}|G_{e}|z^{n_1}\rangle}\\
&=
\left(\int_{\SU(2)}\prod_{t\in T}\rd h_{t}\right)
\prod_{e \in {\bf T}} \delta(G_{e}) 
=\left(\int_{\SU(2)} \prod_{e^*}\rd g_{e^*}\right)\prod_{f^*}\delta(\overrightarrow{\prod}_{e^*\in \partial f^*}g_{e^*})
\,,
\end{split}
\label{eq:amplitude_spinor}
\ee
where an integration by part is used to get the second line.
This matches the results by using the spin network basis as shown in Section \ref{sec:local_holography_SN}. 
\end{proof}




\bibliographystyle{bib-style}
\bibliography{TQFT}

\providecommand{\href}[2]{#2}\begingroup\raggedright\begin{thebibliography}{10}

\bibitem{Baez:1999sr}
J.~C. Baez, ``{An Introduction to spin foam models of quantum gravity and BF
  theory},'' Lect. Notes Phys. {\bf 543} (2000) 25--94,
  \href{http://arXiv.org/abs/gr-qc/9905087}{{\texttt{arXiv:gr-qc/9905087}}}.

\bibitem{Livine:2010zx}
E.~R. Livine, ``{The Spinfoam Framework for Quantum Gravity},'' other thesis,
  10, 2010.

\bibitem{Perez:2012wv}
A.~Perez, ``{The Spin Foam Approach to Quantum Gravity},'' Living Rev. Rel.
  {\bf 16} (2013) 3,
  \href{http://arXiv.org/abs/1205.2019}{{\texttt{arXiv:1205.2019}}}.

\bibitem{Regge:2000wu}
T.~Regge and R.~M. Williams, ``{Discrete structures in gravity},'' J. Math.
  Phys. {\bf 41} (2000) 3964--3984,
  \href{http://arXiv.org/abs/gr-qc/0012035}{{\texttt{arXiv:gr-qc/0012035}}}.

\bibitem{Reisenberger:1996pu}
M.~P. Reisenberger and C.~Rovelli, ``{'Sum over surfaces' form of loop quantum
  gravity},'' Phys. Rev. D {\bf 56} (1997) 3490--3508,
  \href{http://arXiv.org/abs/gr-qc/9612035}{{\texttt{arXiv:gr-qc/9612035}}}.

\bibitem{Engle:2007wy}
J.~Engle, E.~Livine, R.~Pereira, and C.~Rovelli, ``{LQG vertex with finite
  Immirzi parameter},'' Nucl. Phys. B {\bf 799} (2008) 136--149,
  \href{http://arXiv.org/abs/0711.0146}{{\texttt{arXiv:0711.0146}}}.

\bibitem{Reisenberger:2000fy}
M.~Reisenberger and C.~Rovelli, ``{Spin foams as Feynman diagrams},'' in {\em
  {25th Johns Hopkins Workshop on Current Problems in Particle Theory: 2001: A
  Relativistic Spacetime Odyssey. Experiments and Theoretical Viewpoints on
  General Relativity and Quantum Gravity}}, pp.~431--448.
\newblock 2, 2000.
\newblock
  \href{http://arXiv.org/abs/gr-qc/0002083}{{\texttt{arXiv:gr-qc/0002083}}}.

\bibitem{Reisenberger:2000zc}
M.~P. Reisenberger and C.~Rovelli, ``{Space-time as a Feynman diagram: The
  Connection formulation},'' Class. Quant. Grav. {\bf 18} (2001) 121--140,
  \href{http://arXiv.org/abs/gr-qc/0002095}{{\texttt{arXiv:gr-qc/0002095}}}.

\bibitem{Barrett:2008wh}
J.~W. Barrett and I.~Naish-Guzman, ``{The Ponzano-Regge model},'' Class.\
  Quant.\ Grav. {\bf 26} (2009) 155014,
  \href{http://arXiv.org/abs/0803.3319}{{\texttt{arXiv:0803.3319}}}.

\bibitem{Freidel:2005me}
L.~Freidel and E.~R. Livine, ``{3D Quantum Gravity and Effective Noncommutative
  Quantum Field Theory},'' Phys. Rev. Lett. {\bf 96} (2006) 221301,
  \href{http://arXiv.org/abs/hep-th/0512113}{{\texttt{arXiv:hep-th/0512113}}}.

\bibitem{Bonzom:2014wva}
V.~Bonzom, M.~Dupuis, F.~Girelli, and E.~R. Livine, ``{Deformed phase space for
  3d loop gravity and hyperbolic discrete geometries},''
  \href{http://arXiv.org/abs/1402.2323}{{\texttt{arXiv:1402.2323}}}.

\bibitem{Dupuis:2020ndx}
M.~Dupuis, L.~Freidel, F.~Girelli, A.~Osumanu, and J.~Rennert, ``{On the origin
  of the quantum group symmetry in 3d quantum gravity},''
  \href{http://arXiv.org/abs/2006.10105}{{\texttt{arXiv:2006.10105}}}.

\bibitem{Girelli:2007tt}
F.~Girelli, H.~Pfeiffer, and E.~M. Popescu, ``{Topological Higher Gauge Theory
  - from BF to BFCG theory},'' J. Math. Phys. {\bf 49} (2008) 032503,
  \href{http://arXiv.org/abs/0708.3051}{{\texttt{arXiv:0708.3051}}}.

\bibitem{Baratin:2014era}
A.~Baratin and L.~Freidel, ``{A 2-categorical state sum model},'' J. Math.
  Phys. {\bf 56} (2015), no.~1, 011705,
  \href{http://arXiv.org/abs/1409.3526}{{\texttt{arXiv:1409.3526}}}.

\bibitem{Asante:2019lki}
S.~K. Asante, B.~Dittrich, F.~Girelli, A.~Riello, and P.~Tsimiklis, ``{Quantum
  geometry from higher gauge theory},'' Class. Quant. Grav. {\bf 37} (2020),
  no.~20, 205001,
  \href{http://arXiv.org/abs/1908.05970}{{\texttt{arXiv:1908.05970}}}.

\bibitem{Girelli:2021zmt}
F.~Girelli and P.~Tsimiklis, ``{Discretization of 4D Poincar\'e BF theory: From
  groups to 2-groups},'' Phys. Rev. D {\bf 106} (2022), no.~4, 046003,
  \href{http://arXiv.org/abs/2105.01817}{{\texttt{arXiv:2105.01817}}}.

\bibitem{Ponzano:1968se}
G.~Ponzano and T.~E. Regge, ``Semiclassical limit of Racah coefficients,''.

\bibitem{Ooguri:1991ni}
H.~Ooguri, ``{Partition functions and topology changing amplitudes in the 3-D
  lattice gravity of Ponzano and Regge},'' Nucl. Phys. B {\bf 382} (1992)
  276--304,
  \href{http://arXiv.org/abs/hep-th/9112072}{{\texttt{arXiv:hep-th/9112072}}}.

\bibitem{Freidel:2004vi}
L.~Freidel and D.~Louapre, ``{Ponzano-Regge model revisited I: Gauge fixing,
  observables and interacting spinning particles},'' Class.\ Quant.\ Grav. {\bf
  21} (2004) 5685--5726,
  \href{http://arXiv.org/abs/hep-th/0401076}{{\texttt{arXiv:hep-th/0401076}}}.

\bibitem{Freidel:2000uq}
L.~Freidel, ``{A Ponzano-Regge model of Lorentzian 3-dimensional gravity},''
  Nucl. Phys. B Proc. Suppl. {\bf 88} (2000) 237--240,
  \href{http://arXiv.org/abs/gr-qc/0102098}{{\texttt{arXiv:gr-qc/0102098}}}.

\bibitem{Davids:2000kz}
S.~Davids, ``{A State sum model for (2+1) Lorentzian quantum gravity},'' other
  thesis, 10, 2000.

\bibitem{Freidel:2005bb}
L.~Freidel and E.~R. Livine, ``{Ponzano-Regge model revisited III: Feynman
  diagrams and effective field theory},'' Class.\ Quant.\ Grav. {\bf 23} (2006)
  2021--2062,
  \href{http://arXiv.org/abs/hep-th/0502106}{{\texttt{arXiv:hep-th/0502106}}}.

\bibitem{Girelli:2015ija}
F.~Girelli and G.~Sellaroli, ``{3d Lorentzian loop quantum gravity and the
  spinor approach},'' Phys. Rev. D {\bf 92} (2015), no.~12, 124035,
  \href{http://arXiv.org/abs/1506.07759}{{\texttt{arXiv:1506.07759}}}.

\bibitem{Bonzom:2014bua}
V.~Bonzom, M.~Dupuis, and F.~Girelli, ``{Towards the Turaev-Viro amplitudes
  from a Hamiltonian constraint},'' Phys. Rev. D {\bf 90} (2014), no.~10,
  104038, \href{http://arXiv.org/abs/1403.7121}{{\texttt{arXiv:1403.7121}}}.

\bibitem{Turaev:1992hq}
V.~Turaev and O.~Viro, ``{State sum invariants of 3 manifolds and quantum 6j
  symbols},'' Topology {\bf 31} (1992) 865--902.

\bibitem{Freidel:2004nb}
L.~Freidel and D.~Louapre, ``{Ponzano-Regge model revisited II: Equivalence
  with Chern-Simons},''
  \href{http://arXiv.org/abs/gr-qc/0410141}{{\texttt{arXiv:gr-qc/0410141}}}.

\bibitem{Witten:1988hc}
E.~Witten, ``{(2+1)-Dimensional Gravity as an Exactly Soluble System},'' Nucl.
  Phys. B {\bf 311} (1988) 46.

\bibitem{Crane:2001kf}
L.~Crane, ``{A New approach to the geometrization of matter},''
  \href{http://arXiv.org/abs/gr-qc/0110060}{{\texttt{arXiv:gr-qc/0110060}}}.

\bibitem{Gurau:2010nd}
R.~Gurau, ``{Lost in Translation: Topological Singularities in Group Field
  Theory},'' Class. Quant. Grav. {\bf 27} (2010) 235023,
  \href{http://arXiv.org/abs/1006.0714}{{\texttt{arXiv:1006.0714}}}.

\bibitem{Costantino:2011gen}
F.~Costantino and J.~Marche, ``Generating series and asymptotics of classical
  spin networks,'' arXiv preprint arXiv:1103.5644 (2011).

\bibitem{Ooguri:1991ib}
H.~Ooguri and N.~Sasakura, ``{Discrete and continuum approaches to
  three-dimensional quantum gravity},'' Mod. Phys. Lett. A {\bf 6} (1991)
  3591--3600,
  \href{http://arXiv.org/abs/hep-th/9108006}{{\texttt{arXiv:hep-th/9108006}}}.

\bibitem{Boulatov:1992vp}
D.~Boulatov, ``{A Model of three-dimensional lattice gravity},'' Mod. Phys.
  Lett. A {\bf 7} (1992) 1629--1646,
  \href{http://arXiv.org/abs/hep-th/9202074}{{\texttt{arXiv:hep-th/9202074}}}.

\bibitem{Carrozza:2013oiy}
S.~Carrozza, {\em {Tensorial methods and renormalization in Group Field
  Theories}}.
\newblock PhD thesis, Orsay, LPT, 2013.
\newblock \href{http://arXiv.org/abs/1310.3736}{{\texttt{arXiv:1310.3736}}}.

\bibitem{Dowdall:2009eg}
R.~Dowdall, H.~Gomes, and F.~Hellmann, ``{Asymptotic analysis of the
  Ponzano-Regge model for handlebodies},'' J.\ Phys.\ A {\bf 43} (2010) 115203,
  \href{http://arXiv.org/abs/0909.2027}{{\texttt{arXiv:0909.2027}}}.

\bibitem{Goeller:2019zpz}
C.~Goeller, E.~R. Livine, and A.~Riello, ``{Non-Perturbative 3D Quantum
  Gravity: Quantum Boundary States and Exact Partition Function},'' Gen.\ Rel.\
  Grav. {\bf 52} (2020), no.~3, 24,
  \href{http://arXiv.org/abs/1912.01968}{{\texttt{arXiv:1912.01968}}}.

\bibitem{Bonzom:2010ar}
V.~Bonzom and M.~Smerlak, ``{Bubble divergences from cellular cohomology},''
  Lett. Math. Phys. {\bf 93} (2010) 295--305,
  \href{http://arXiv.org/abs/1004.5196}{{\texttt{arXiv:1004.5196}}}.

\bibitem{Bonzom:2010zh}
V.~Bonzom and M.~Smerlak, ``{Bubble divergences from twisted cohomology},''
  Commun. Math. Phys. {\bf 312} (2012) 399--426,
  \href{http://arXiv.org/abs/1008.1476}{{\texttt{arXiv:1008.1476}}}.

\bibitem{Bonzom:2012bn}
V.~Bonzom and E.~R. Livine, ``{Generating Functions for Coherent
  Intertwiners},'' Class. Quant. Grav. {\bf 30} (2013) 055018,
  \href{http://arXiv.org/abs/1205.5677}{{\texttt{arXiv:1205.5677}}}.

\bibitem{Bonzom:2015ova}
V.~Bonzom, F.~Costantino, and E.~R. Livine, ``{Duality between Spin networks
  and the 2D Ising model},'' Commun.\ Math.\ Phys. {\bf 344} (2016), no.~2,
  531--579,
  \href{http://arXiv.org/abs/1504.02822}{{\texttt{arXiv:1504.02822}}}.

\bibitem{Schwinger:1965an}
J.~Schwinger, ``On Angular Momentum, USAEC Report NYO-3071 (1952); reprinted in
  LC Biedenharn and H. van Dam,(editors), Quantum Theory of Angular Momentum,''
  1965.

\bibitem{Bargmann:1962zz}
V.~Bargmann, ``{On the Representations of the Rotation Group},'' Rev. Mod.
  Phys. {\bf 34} (1962)
829--845.

\bibitem{Freidel:2010tt}
L.~Freidel and E.~R. Livine, ``{U(N) Coherent States for Loop Quantum
  Gravity},'' J. Math. Phys. {\bf 52} (2011) 052502,
  \href{http://arXiv.org/abs/1005.2090}{{\texttt{arXiv:1005.2090}}}.

\bibitem{Dupuis:2010iq}
M.~Dupuis and E.~R. Livine, ``{Revisiting the Simplicity Constraints and
  Coherent Intertwiners},'' Class. Quant. Grav. {\bf 28} (2011) 085001,
  \href{http://arXiv.org/abs/1006.5666}{{\texttt{arXiv:1006.5666}}}.

\bibitem{Dupuis:2011dh}
M.~Dupuis and E.~R. Livine, ``{Holomorphic Simplicity Constraints for 4d
  Riemannian Spinfoam Models},'' J. Phys. Conf. Ser. {\bf 360} (2012) 012046,
  \href{http://arXiv.org/abs/1111.1125}{{\texttt{arXiv:1111.1125}}}.

\bibitem{Varshalovich:1988qu}
D.~A. Varshalovich, A.~N. Moskalev, and V.~K. Khersonskii, {\em Quantum theory
  of angular momentum}.
\newblock World Scientific, 1988.

\bibitem{Freidel:2012ji}
L.~Freidel and J.~Hnybida, ``{On the exact evaluation of spin networks},'' J.
  Math. Phys. {\bf 54} (2013) 112301,
  \href{http://arXiv.org/abs/1201.3613}{{\texttt{arXiv:1201.3613}}}.

\bibitem{Bonzom:2011nv}
V.~Bonzom and E.~R. Livine, ``{A New Hamiltonian for the Topological BF phase
  with spinor networks},'' J. Math. Phys. {\bf 53} (2012) 072201,
  \href{http://arXiv.org/abs/1110.3272}{{\texttt{arXiv:1110.3272}}}.

\bibitem{Witten:1998qj}
E.~Witten, ``{Anti-de Sitter space and holography},'' Adv. Theor. Math. Phys.
  {\bf 2} (1998) 253--291,
  \href{http://arXiv.org/abs/hep-th/9802150}{{\texttt{arXiv:hep-th/9802150}}}.

\bibitem{Witten:1998wy}
E.~Witten, ``{AdS / CFT correspondence and topological field theory},'' JHEP
  {\bf 12} (1998) 012,
  \href{http://arXiv.org/abs/hep-th/9812012}{{\texttt{arXiv:hep-th/9812012}}}.

\bibitem{Dittrich:2017hnl}
B.~Dittrich, C.~Goeller, E.~Livine, and A.~Riello, ``{Quasi-local holographic
  dualities in non-perturbative 3d quantum gravity I -- Convergence of multiple
  approaches and examples of Ponzano--Regge statistical duals},'' Nucl.\ Phys.\
  B {\bf 938} (2019) 807--877,
  \href{http://arXiv.org/abs/1710.04202}{{\texttt{arXiv:1710.04202}}}.

\bibitem{Dittrich:2017rvb}
B.~Dittrich, C.~Goeller, E.~R. Livine, and A.~Riello, ``{Quasi-local
  holographic dualities in non-perturbative 3d quantum gravity II -- From
  coherent quantum boundaries to BMS$_3$ characters},'' Nucl.\ Phys.\ B {\bf
  938} (2019) 878--934,
  \href{http://arXiv.org/abs/1710.04237}{{\texttt{arXiv:1710.04237}}}.

\bibitem{Dittrich:2018xuk}
B.~Dittrich, C.~Goeller, E.~R. Livine, and A.~Riello, ``{Quasi-local
  holographic dualities in non-perturbative 3d quantum gravity},'' Class.\
  Quant.\ Grav. {\bf 35} (2018), no.~13, 13LT01,
  \href{http://arXiv.org/abs/1803.02759}{{\texttt{arXiv:1803.02759}}}.

\bibitem{OLoughlin:2000yww}
M.~O'Loughlin, ``{Boundary actions in Ponzano-Regge discretization, quantum
  groups and AdS(3)},'' Adv. Theor. Math. Phys. {\bf 6} (2003) 795--826,
  \href{http://arXiv.org/abs/gr-qc/0002092}{{\texttt{arXiv:gr-qc/0002092}}}.

\bibitem{Achucarro:1989ch}
A.~Achucarro and P.~K. Townsend, ``A Chern--Simons action for three-dimensional
  anti-de Sitter supergravity theories,'' in {\em Supergravities in Diverse
  Dimensions: Commentary and Reprints (In 2 Volumes)}, pp.~732--736.
\newblock World Scientific, 1989.

\bibitem{Cimasoni2007di}
D.~Cimasoni and N.~Reshetikhin, ``Dimers on surface graphs and spin structures.
  I,'' Communications in Mathematical Physics {\bf 275} (2007), no.~1,
  187--208.

\bibitem{Freidel:2002dw}
L.~Freidel and D.~Louapre, ``{Diffeomorphisms and spin foam models},'' Nucl.
  Phys. B {\bf 662} (2003) 279--298,
  \href{http://arXiv.org/abs/gr-qc/0212001}{{\texttt{arXiv:gr-qc/0212001}}}.

\bibitem{Freidel:2002xb}
L.~Freidel and E.~R. Livine, ``{Spin networks for noncompact groups},'' J.
  Math. Phys. {\bf 44} (2003) 1322--1356,
  \href{http://arXiv.org/abs/hep-th/0205268}{{\texttt{arXiv:hep-th/0205268}}}.

\bibitem{Bonzom:2009zd}
V.~Bonzom, E.~R. Livine, and S.~Speziale, ``{Recurrence relations for spin foam
  vertices},'' Class. Quant. Grav. {\bf 27} (2010) 125002,
  \href{http://arXiv.org/abs/0911.2204}{{\texttt{arXiv:0911.2204}}}.

\bibitem{Livine:2002ak}
E.~R. Livine, ``{Projected spin networks for Lorentz connection: Linking spin
  foams and loop gravity},'' Class. Quant. Grav. {\bf 19} (2002) 5525--5542,
  \href{http://arXiv.org/abs/gr-qc/0207084}{{\texttt{arXiv:gr-qc/0207084}}}.

\bibitem{Alexandrov:2002br}
S.~Alexandrov and E.~R. Livine, ``{SU(2) loop quantum gravity seen from
  covariant theory},'' Phys. Rev. D {\bf 67} (2003) 044009,
  \href{http://arXiv.org/abs/gr-qc/0209105}{{\texttt{arXiv:gr-qc/0209105}}}.

\bibitem{Livine:2006ix}
E.~R. Livine, ``{Towards a Covariant Loop Quantum Gravity},''
  \href{http://arXiv.org/abs/gr-qc/0608135}{{\texttt{arXiv:gr-qc/0608135}}}.

\bibitem{Hartle:1981cf}
J.~Hartle and R.~Sorkin, ``{Boundary Terms in the Action for the Regge
  Calculus},'' Gen. Rel. Grav. {\bf 13} (1981) 541--549.

\bibitem{Roberts:1998zka}
J.~Roberts, ``{Classical 6j-symbols and the tetrahedron},'' Geom. Topol. {\bf
  3} (1999), no.~1, 21--66,
  \href{http://arXiv.org/abs/math-ph/9812013}{{\texttt{arXiv:math-ph/9812013}}}.

\bibitem{Schulten:1975sem}
K.~Schulten and R.~G. Gordon, ``Semiclassical approximations to 3 j-and 6
  j-coefficients for quantum-mechanical coupling of angular momenta,'' Journal
  of Mathematical Physics {\bf 16} (1975), no.~10, 1971--1988.

\bibitem{Freidel:2002mj}
L.~Freidel and D.~Louapre, ``{Asymptotics of 6j and 10j symbols},'' Class.\
  Quant.\ Grav. {\bf 20} (2003) 1267--1294,
  \href{http://arXiv.org/abs/hep-th/0209134}{{\texttt{arXiv:hep-th/0209134}}}.

\bibitem{Barrett:1993db}
J.~W. Barrett and T.~Foxon, ``{Semiclassical limits of simplicial quantum
  gravity},'' Class. Quant. Grav. {\bf 11} (1994) 543--556,
  \href{http://arXiv.org/abs/gr-qc/9310016}{{\texttt{arXiv:gr-qc/9310016}}}.

\bibitem{Livine:2011gp}
E.~R. Livine and J.~Tambornino, ``{Spinor Representation for Loop Quantum
  Gravity},'' J. Math. Phys. {\bf 53} (2012) 012503,
  \href{http://arXiv.org/abs/1105.3385}{{\texttt{arXiv:1105.3385}}}.

\bibitem{Borja:2010rc}
E.~F. Borja, L.~Freidel, I.~Garay, and E.~R. Livine, ``{U(N) tools for Loop
  Quantum Gravity: The Return of the Spinor},'' Class. Quant. Grav. {\bf 28}
  (2011) 055005,
  \href{http://arXiv.org/abs/1010.5451}{{\texttt{arXiv:1010.5451}}}.

\bibitem{Livine:2007vk}
E.~R. Livine and S.~Speziale, ``{A New spinfoam vertex for quantum gravity},''
  Phys. Rev. D {\bf 76} (2007) 084028,
  \href{http://arXiv.org/abs/0705.0674}{{\texttt{arXiv:0705.0674}}}.

\bibitem{Livine:2007ya}
E.~R. Livine and S.~Speziale, ``{Consistently Solving the Simplicity
  Constraints for Spinfoam Quantum Gravity},'' EPL {\bf 81} (2008), no.~5,
  50004, \href{http://arXiv.org/abs/0708.1915}{{\texttt{arXiv:0708.1915}}}.

\bibitem{Westbury:1998ge}
B.~W. Westbury, ``A generating function for spin network evaluations,'' Banach
  Center Publications {\bf 42} (1998) 447--456.

\bibitem{Dittrich:2013jxa}
B.~Dittrich and J.~Hnybida, ``{Ising model from intertwiners},'' Ann. Inst. H.
  Poincare D Comb. Phys. Interact. {\bf 3} (2016), no.~4, 363--380,
  \href{http://arXiv.org/abs/1312.5646}{{\texttt{arXiv:1312.5646}}}.

\bibitem{Bonzom:2019dpg}
V.~Bonzom and E.~R. Livine, ``{Self-duality of the 6j-symbol and Fisher zeros
  for the Tetrahedron},''
  \href{http://arXiv.org/abs/1905.00348}{{\texttt{arXiv:1905.00348}}}.

\bibitem{Huang:1974gs}
C.-S. Huang and A.~T. Wu, ``{Structure of the 12j and 15j coefficients in the
  bargmann approach},'' J. Math. Phys. {\bf 15} (1974) 1490--1493.

\bibitem{Labarthe:1975yf}
J.~Labarthe, ``{Generating Functions for the Coupling Recoupling Coefficients
  of SU(2)},'' J. Phys. A {\bf 8} (1975) 1543--1561.

\bibitem{Bonzom:2011hm}
V.~Bonzom and L.~Freidel, ``{The Hamiltonian constraint in 3d Riemannian loop
  quantum gravity},'' Class. Quant. Grav. {\bf 28} (2011) 195006,
  \href{http://arXiv.org/abs/1101.3524}{{\texttt{arXiv:1101.3524}}}.

\bibitem{Dupuis:2011fz}
M.~Dupuis and E.~R. Livine, ``{Holomorphic Simplicity Constraints for 4d
  Spinfoam Models},'' Class. Quant. Grav. {\bf 28} (2011) 215022,
  \href{http://arXiv.org/abs/1104.3683}{{\texttt{arXiv:1104.3683}}}.

\bibitem{Freidel:1998pt}
L.~Freidel and K.~Krasnov, ``{Spin foam models and the classical action
  principle},'' Adv. Theor. Math. Phys. {\bf 2} (1999) 1183--1247,
  \href{http://arXiv.org/abs/hep-th/9807092}{{\texttt{arXiv:hep-th/9807092}}}.

\bibitem{Dupuis:2011fx}
M.~Dupuis, F.~Girelli, and E.~R. Livine, ``{Spinors and Voros star-product for
  Group Field Theory: First Contact},'' Phys. Rev. D {\bf 86} (2012) 105034,
  \href{http://arXiv.org/abs/1107.5693}{{\texttt{arXiv:1107.5693}}}.

\bibitem{Banburski:2014cwa}
A.~Banburski, L.-Q. Chen, L.~Freidel, and J.~Hnybida, ``{Pachner moves in a 4d
  Riemannian holomorphic Spin Foam model},'' Phys. Rev. D {\bf 92} (2015),
  no.~12, 124014,
  \href{http://arXiv.org/abs/1412.8247}{{\texttt{arXiv:1412.8247}}}.

\bibitem{Chen:2016aag}
L.-Q. Chen, ``{Bulk amplitude and degree of divergence in 4d spin foams},''
  Phys. Rev. D {\bf 94} (2016), no.~10, 104025,
  \href{http://arXiv.org/abs/1602.01825}{{\texttt{arXiv:1602.01825}}}.

\bibitem{Dupuis:2014fya}
M.~Dupuis, F.~Girelli, and E.~R. Livine, ``{Deformed Spinor Networks for Loop
  Gravity: Towards Hyperbolic Twisted Geometries},'' Gen. Rel. Grav. {\bf 46}
  (2014), no.~11, 1802,
  \href{http://arXiv.org/abs/1403.7482}{{\texttt{arXiv:1403.7482}}}.

\bibitem{Bonzom:2022bpv}
V.~Bonzom, M.~Dupuis, F.~Girelli, and Q.~Pan, ``{Local Observables in
  $\operatorname{SU}_q(2)$ Lattice Gauge Theory},''
  \href{http://arXiv.org/abs/2205.13352}{{\texttt{arXiv:2205.13352}}}.

\bibitem{Freidel:2009ck}
L.~Freidel and E.~R. Livine, ``{The Fine Structure of SU(2) Intertwiners from
  U(N) Representations},'' J. Math. Phys. {\bf 51} (2010) 082502,
  \href{http://arXiv.org/abs/0911.3553}{{\texttt{arXiv:0911.3553}}}.

\bibitem{Girelli:2017dbk}
F.~Girelli and G.~Sellaroli, ``{SO*(2N) coherent states for loop quantum
  gravity},'' J. Math. Phys. {\bf 58} (2017), no.~7, 071708,
  \href{http://arXiv.org/abs/1701.07519}{{\texttt{arXiv:1701.07519}}}.

\end{thebibliography}\endgroup

\end{document}